%% file: main.tex
\documentclass{article}

\usepackage{dilab_arxiv}

\usepackage{amsfonts}
\usepackage{ifthen}
\usepackage{textcomp}
\usepackage{mathtools}
\usepackage{bbm,dsfont} 
\usepackage{algorithm}
\usepackage{wrapfig}
\usepackage{algpseudocode} 
\usepackage{multirow}
\usepackage{enumitem} 
\usepackage[section]{placeins} 
\usepackage{tikz}
\usetikzlibrary{arrows.meta,positioning,fit,backgrounds,calc,decorations,decorations.pathreplacing,shapes.geometric}

\usepackage{CJKutf8} 

\usepackage[capitalise,nameinlink]{cleveref}

\let\standardEqref\eqref

\input{math_commands}
\input{figures/SNIPPET_ttt_macros.tex}
\input{figures/SNIPPET_headline_macros.tex}

\providecommand{\Description}[2][]{}

\let\eqref\standardEqref

\theoremstyle{plain}

\newtheorem{proposition}{Proposition}

\theoremstyle{definition}

\theoremstyle{remark}

\newcommand{\compilehidecomments}{false}
\ifthenelse{ \equal{\compilehidecomments}{true} }{%
	\newcommand{\yu}[1]{}
    \newcommand{\longbo}[1]{}
    
    \colorlet{ruishuocolor}{black}
}{
	\newcommand{\yu}[1]{{\color{cyan}[\text{Yu:} #1]}}
    \newcommand{\longbo}[1]{{\color{orange}[\text{Longbo:} #1]}}
    \definecolor{ruishuocolor}{RGB}{0,155,80}
    
}

\title{GPU-CFR: 80x Faster Counterfactual Regret Minimization by Compiling the Game to Static Dataflow and CUDA Graph Replay}
\runningtitle{Compiling CFR to Static Dataflow and CUDA Graph Replay}
\date{\today}

\paperlogo{\includegraphics[height=1.5cm]{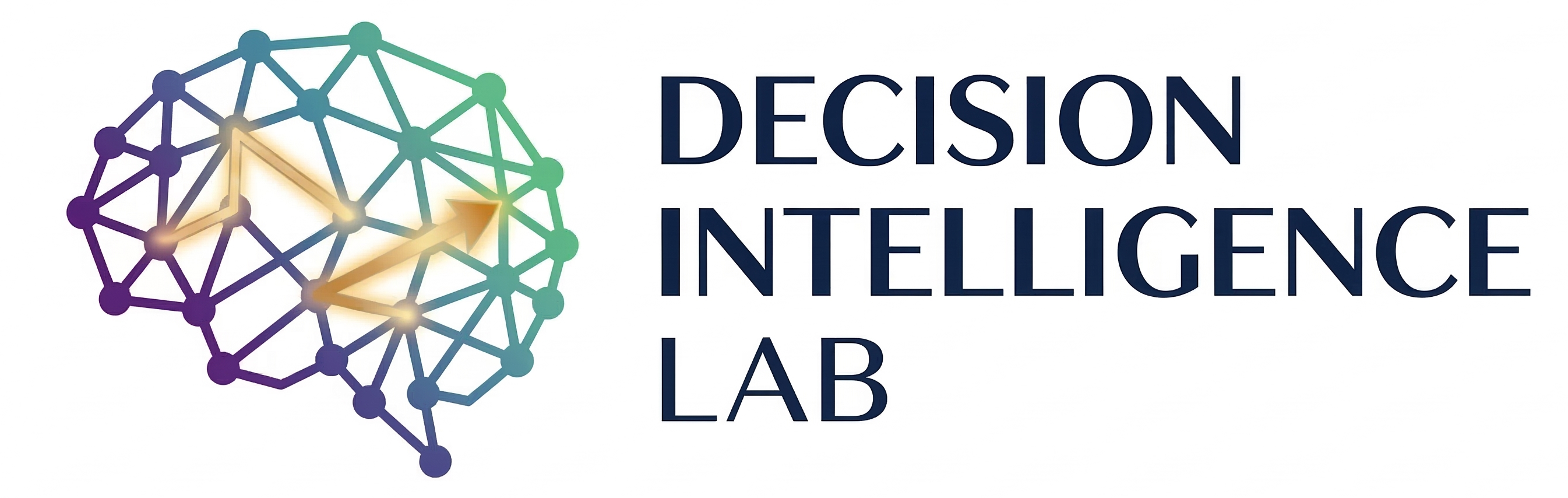}}

\author{
  Boning Li$^{1}$ and Longbo Huang$^{1\,\text{\faEnvelope}}$
  \\[0.3em]\normalfont
  $^1$Institute for Interdisciplinary Information Sciences, Tsinghua University
  \\
  \text{\faEnvelope}\ Correspondence: longbohuang@tsinghua.edu.cn 
}

\begin{document}

\maketitle
\thispagestyle{fancy}

\begin{abstract}
\input{sections/0_abstract}
\end{abstract}

\maketitle

\input{sections/1_introduction}
\input{sections/2_related_work}
\input{sections/3_preliminaries}
\input{sections/4_method}
\input{sections/4b_theory}
\input{sections/5_setup}
\input{sections/6_results}

\input{sections/8_conclusion}

\clearpage

\bibliography{references}
\bibliographystyle{dilab_ref}

\makeappendixtoc
\input{sections/A_hunl_construction}
\input{sections/B_protocol}
\input{sections/C_correctness}
\input{sections/D_extra_results}
\input{sections/E_libratus_headtohead}
\input{sections/G_profile}

\input{sections/H_compiler}

\input{sections/F_proofs}
\input{sections/I_future_work}

\end{document}

%% file: math_commands.tex
\providecommand{\R}{\mathbb{R}}

\providecommand{\nodes}{\mathcal{H}}
\providecommand{\termnodes}{\mathcal{Z}}
\providecommand{\infosets}{\mathcal{I}}
\providecommand{\acts}{A}
\providecommand{\reach}{\pi}
\providecommand{\creach}{\pi^{-i}}
\providecommand{\cfv}{v}
\providecommand{\regret}{R}
\providecommand{\strat}{\sigma}
\providecommand{\avgstrat}{\bar{\sigma}}
\providecommand{\expl}{\mathrm{expl}}

\newcommand{\numnodes}{N}
\newcommand{\numedges}{E}
\newcommand{\numslots}{A_{\mathrm{tot}}}
\newcommand{\numinfosets}{I}

\newcommand{\fwd}{\mathrm{F}}             
\newcommand{\bwd}{\mathrm{B}}             
\newcommand{\acc}{\mathrm{R}}             
\newcommand{\ulp}{\mathrm{ulp}}           

%% file: figures/SNIPPET_ttt_macros.tex
\newcommand{\tttRatioMin}{3.8}
\newcommand{\tttRatioMax}{44}

%% file: figures/SNIPPET_headline_macros.tex
\newcommand{\kimGraphRange}{29.8--80.4}
\newcommand{\kimGraphMedian}{44.1}
\newcommand{\kimEagerRange}{17.4--23.5}
\newcommand{\kimCpuRange}{2.2--51.1}
\newcommand{\kimCpuMedian}{11.5}
\newcommand{\liteGraphLargestRange}{14--258}
\newcommand{\liteGraphTurn}{258}
\newcommand{\graphCpuLargestRange}{3.1--13.5}
\newcommand{\cpuLiteLargestRange}{4.4--25.8}
\newcommand{\eagerGraphRange}{1.6--3.6}
\newcommand{\eagerMsRange}{0.401--0.933}
\newcommand{\graphMsTwoLargest}{0.380--0.556}
\newcommand{\turnSteadyMs}{0.397}
\newcommand{\turnFirstSolveSec}{4.936}
\newcommand{\turnKimThousandSec}{11.893}
\newcommand{\turnLiteThousandSec}{102.160}

%% file: sections/0_abstract.tex
Counterfactual regret minimization (CFR) is the standard solver for imperfect-information
extensive-form games. It is also one of the few large numerical
workloads that still runs faster on CPUs than on GPUs. Each iteration sweeps a game
tree with up to billions of states in millions of small, interdependent gather and
scatter steps issued through a generic tree interface. On a GPU every kernel finishes
in microseconds, so kernel launches and framework dispatch dominate the run time, and
prior GPU implementations have lost to optimized CPU code. We observe that for a fixed
game, everything about a CFR iteration except the numerical values is known before the
first iteration runs. We propose GPU-CFR, a compiler and runtime built on this
observation. It compiles any two-player zero-sum perfect-recall game once into static dataflow: flat
edge and information-set arrays, precomputed indices, and depth-level batched passes
fix the entire operation sequence, and only solver state changes between iterations.
Static chance folding, depth-level execution blocks, and a dual-lane reach buffer cut
the number of framework operations by up to 18.1$\times$. Because shapes, indices, and
buffer addresses never change, CUDA Graph Replay records the iteration once and
replays it with a single graph launch. On one A100, across an eight-game suite that
spans card games, dice games, and board games, GPU-CFR runs \kimGraphRange$\times$
faster than the fastest prior GPU CFR on the same accelerator (median
\kimGraphMedian$\times$), and \liteGraphLargestRange$\times$ faster than LiteEFG, one of the
fastest open-source CPU implementations, on the four largest games. The compiled
representation carries most of that margin: on eight CPU threads with no accelerator
it is already \kimCpuRange$\times$ faster than the GPU baseline. On the CPU the
optimized path reproduces the reference iterates bitwise, and tree construction and
graph capture pay for themselves within the first solve. GPU-CFR beats every CPU and
GPU baseline on the mid-to-large games of the suite without changing the update rule.
Code is available at \url{https://github.com/lbn187/GPU-CFR}.

%% file: sections/1_introduction.tex
\section{Introduction}
\label{sec:intro}

Counterfactual regret minimization (CFR) \cite{zinkevich2007regret} is the algorithm
of choice for computing equilibria in large games of hidden information, a class that
includes poker, bargaining, and adversarial security settings. Interestingly, tabular
CFR has resisted the move to accelerators, and its fastest implementations still run
on CPUs. Its execution
model explains why. One iteration touches every node of a game tree in millions of tiny
dependent steps. Each step chases a pointer, branches on the node type, and updates a
handful of floats. Every read and write is therefore a data-dependent gather or scatter
\cite{kim2026parallelizing}, and a wide processor finds nothing to batch in such a
walk. Deployed implementations therefore still execute CFR as a node-wise tree
traversal through a generic game interface
\cite{lanctot2019openspiel,liu2024liteefg,steinberger2019pokerrl,li2026real}.

This cost matters because the solves are time-bound. Modern game-playing agents must
solve a game tree with tens of millions of states within a few seconds during play
\cite{burch2014solving,ganzfried2015endgame,moravcik2016refining,brown2017safe,moravvcik2017deepstack,brown2018superhuman,brown2019superhuman,vsustr2019monte,brown2018depth,steinberger2019pokerrl},
so the cost of one solver iteration bounds what any such agent can do. Tabular CFR and
its CFR$^{+}$ refinement \cite{tammelin2014solving} produced the essential solution of
heads-up limit hold'em \cite{bowling2015heads,tammelin2015solving} and drive
superhuman poker AIs
\cite{brown2018superhuman,moravvcik2017deepstack,brown2019superhuman}. The tabular form
remains the field's workhorse: the reference against which every approximation is
checked \cite{lanctot2009monte,brown2019deep,schmid2019variance}, the source of the
exact best-response evaluations that certify solution quality
\cite{johanson2011accelerating,li2026av,li2026agents}, and the inner loop of online re-solving
\cite{burch2014solving,brown2017safe,sustr2020sound,li2025efficient,li2026real}.

That demand has driven many attempts to move the workload onto accelerators, and none has
produced a GPU solver whose iteration time is below that of an efficient CPU
implementation \cite{kim2026parallelizing,DBLP:journals/corr/abs-2508-06559}. The
milestone poker systems sidestepped the GPU entirely and used supercomputer-scale CPU
parallelism instead \cite{bowling2015heads,tammelin2015solving,brown2018superhuman}.
Porting the traversal to a GPU shows why. Every kernel runs for microseconds, so the
time goes into launching kernels and dispatching framework operations
\cite{guide2020cuda,paszke2019pytorch}, and arithmetic is a small fraction of it. The fastest prior GPU CFR
is the sequence-form \cite{koller1997representations} implementation of Kim (2026)
\cite{kim2026parallelizing}. It expresses each tree level as a sparse matrix product,
yet its iteration still takes longer than an optimized CPU solver. Its iteration also
cannot be recorded as a CUDA graph, because the sparse products and per-level index
traffic it relies on are not capturable (\cref{app:profile}). The obstacle lies in the
representation the algorithm runs on, and the workload itself is unusually regular. When the same game is solved for many
iterations, the tree topology, information sets, chance behavior, tensor shapes, and
data dependencies stay constant. Only solver state changes: strategies, reaches,
regrets, values, and the iteration weight. This raises a systems question. How much
of tabular CFR's cost is an artifact of its execution plan, and how much of it can a
compiler remove without touching the algorithm?

\begin{figure}[!htb]
    \centering
    \includegraphics[width=0.62\linewidth]{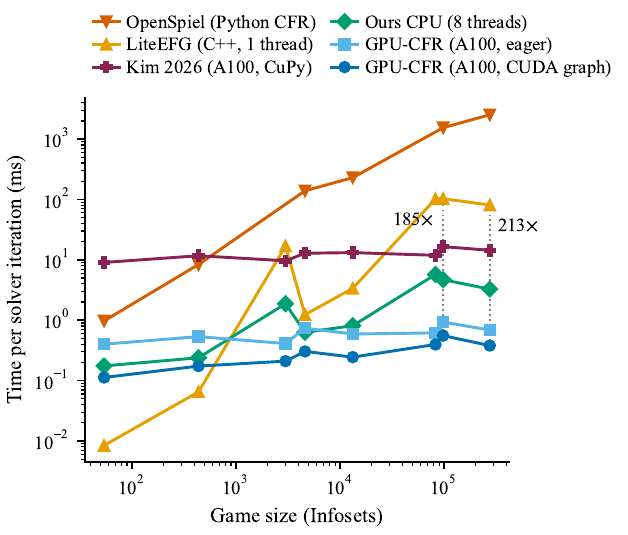}
    \caption{Steady-state time per CFR iteration versus game size (both axes
    logarithmic; lower is better); backends and variants as in \cref{tab:matrix}.}
    \Description{Log-log scatter plot of steady-state milliseconds per iteration
    against game size in Infosets for six systems. The graph-replayed GPU-CFR curve
    stays nearly flat while baseline curves rise with game size, giving the widest
    margins on the largest games.}
    \label{fig:hero}
\end{figure}

To remove these obstacles, we propose GPU-CFR, which treats a fixed game as a program
to be compiled. GPU-CFR walks a game once (two
players, zero-sum, perfect recall) and emits flat edge,
information-set slot, and depth arrays. An iteration then becomes a fixed sequence of
depth-level batched gather and scatter operations over those arrays. Three techniques
shape the compiled iteration. Static chance folding moves all strategy-independent
chance work to build time. Depth-level execution blocks collapse each traversal stage
into one batched operation per tree level. A dual-lane reach buffer advances both
players' reach probabilities in a single operation, and a sentinel slot removes the
last per-edge branch. The resulting sequence has fixed shapes, indices, and buffer
addresses, so CUDA Graph Replay \cite{guide2020cuda} records it once and replays it
with one graph launch per iteration. The design separates the two sources of run-time
cost. The compiled representation reduces the work issued to the tensor framework, and
graph replay reduces the cost of issuing it.

The compiled iteration pays off most where the trees are largest. \Cref{fig:hero} plots
steady-state time per iteration against game size for every backend we test. The CPU
libraries and the prior GPU implementation slow down as the tree grows. The
graph-replayed GPU-CFR curve stays nearly flat, so the margin widens with scale. On
one A100, a heads-up no-limit Texas hold'em (HUNL) turn subgame with 83{,}040
information sets (Infosets) runs at \turnSteadyMs\,ms per
steady-state CFR$^{+}$ iteration. Against the matched A100 implementation of
Kim (2026)~\cite{kim2026parallelizing}, graph replay is \kimGraphRange$\times$ faster
across the eight-game suite. Against LiteEFG \cite{liu2024liteefg}, one of the
fastest open-source CPU implementations, it is
\liteGraphLargestRange$\times$ faster on the four largest games. Most of this margin
comes from the compiled representation itself: the same compiled solver on eight CPU
threads already beats the A100 baseline by \kimCpuRange$\times$. GPU-CFR in turn
outperforms every CPU path we test, including that eight-thread arm and its one-socket
scaling in \cref{app:profile}.

We answer three systems questions. (Q1) How can a full CFR iteration be compiled into
fixed tensor dataflow while preserving its update semantics? (Q2) Which gains come
from the compiled representation and which come from CUDA Graph Replay? (Q3) When does
the lower iteration cost translate into lower wall-clock time for solution quality and
for repeated solves? Our contributions are:
\begin{enumerate}
    \item \textbf{A game-to-dataflow compiler for tabular CFR.} GPU-CFR maps a fixed
game to flat arrays and a depth-level execution schedule, including static chance
folding and a dual-lane reach representation
(\cref{fig:archpipeline,fig:archdataflow}).
    \item \textbf{CUDA Graph Replay with a causal performance decomposition.} We
separate representation cost from framework dispatch and graph-launch cost through
operation counts, eager/graph ablations, scaling, and repeated-solve measurements.
    \item \textbf{State-of-the-art GPU performance for tabular CFR.} GPU-CFR runs
\kimGraphRange$\times$ faster than the matched A100 baseline across the eight-game
suite and outperforms every tested CPU path (\cref{sec:results}).
\end{enumerate}

%% file: sections/2_related_work.tex
\section{Related Work}
\label{sec:related}

\paragraph{The CFR family and exact equilibrium solving.}
CFR \cite{zinkevich2007regret} minimizes counterfactual regret independently at every
information set. In two-player zero-sum games the average strategy profile converges
to a Nash equilibrium \cite{nash1950equilibrium}. CFR$^{+}$
\cite{tammelin2014solving} replaces regret matching \cite{hart2000simple} with regret matching$^{+}$
and linear averaging, and was the engine behind the essential solution of heads-up
limit hold'em \cite{tammelin2015solving,bowling2015heads}.
Subsequent variants tune the update rule: discounting \cite{brown2019solving},
dynamic discounting \cite{xu2024dynamic}, predictive/optimistic updates connecting
CFR to Blackwell approachability
\cite{farina2021faster,farina2019optimistic,farina2019stable,xu2024minimizing}, and
refinements of regret matching itself \cite{farina2023regret,meng2026a}.
First-order methods reach equilibria by a different route, smoothing
\cite{hoda2010smoothing,nesterov2005excessive} on the sequence form
\cite{koller1997representations}. Burch et al.\ \cite{burch2019revisiting} analyze
whether the two players are updated \emph{simultaneously} or \emph{alternately},
which changes iterate quality at matched iteration counts. Our work targets the
execution layer shared by the tabular CFR family. The compiled slot representation
supports vanilla CFR, CFR$^{+}$, DCFR \cite{brown2019solving}, and predictive
CFR$^{+}$ \cite{farina2021faster} under both update orders.

\paragraph{Scaling CFR by approximation.}
When the full tree is too large, the community samples it, shrinks it, prunes it, or
learns it. Monte Carlo CFR estimates counterfactual values from sampled trajectories
\cite{lanctot2009monte}, with public-chance and other variance-controlled variants
for solving and evaluation
\cite{johanson2012efficient,burch2012efficient,schmid2019variance,li2026correlated,li2026av}.
Abstraction maps the game onto a smaller proxy
\cite{gilpin2007lossless,waugh2009abstraction,kroer2018unified,li2024rl,li2025efficient,li2026effective,li2026abstraction},
a line of work motivated directly by CFR's per-iteration cost
\cite{johanson2012efficient,brown2015hierarchical}. Regret-based and online pruning
skip low-value branches during the traversal \cite{brown2015regret,li2025efficient}.
Deep CFR and its successors replace tabular regrets with neural approximation
\cite{brown2019deep,DBLP:journals/corr/abs-1901-07621}, and search-based and
language-model agents combine learned models with online play
\cite{brown2020combining,schmid2023student,li2026pokerskill,wang2026solver}. These methods reduce the game or the
work sampled from it. Our compiler addresses the regime in which the full tree is
available and repeatedly executed. Every one of these methods still runs a tabular CFR
sweep on whatever tree remains, so a faster sweep composes with all of them.

\paragraph{Systems and frameworks for game solving.}
OpenSpiel \cite{lanctot2019openspiel} is the field's reference library. Its CFR
implementations favor readability and breadth of game support. LiteEFG
\cite{liu2024liteefg} compiles extensive-form games into a graph representation
executed by a single-threaded C++ backend. We use it as the primary compiled CPU
baseline. PokerRL \cite{steinberger2019pokerrl} vectorizes CFR over hand ranges
\cite{johanson2012efficient,moravvcik2017deepstack}, a HUNL-specific representation.
Its released Libratus endgames are our external poker reference in
\cref{app:libratus}. Accelerated best-response computation
\cite{johanson2011accelerating} speeds up the \emph{evaluation} half of the loop with
game-specific structure. Our exploitability evaluator is a generic analogue that runs
on the same compiled arrays as the solver. Parallel CFR has been explored through
distributed sampling and equilibrium computation
\cite{jackson2013slumbot,brown2015hierarchical}, on CPU clusters
\cite{DBLP:conf/ispa/ZhangLZXW21a}, and in real-time solvers that partition the
traversal across CPU threads under per-decision budgets \cite{li2026real}. The
milestone poker systems used supercomputer-scale CPU parallelism
\cite{bowling2015heads,brown2018superhuman}. Those designs parallelize the
\emph{traversal}. We change the representation the traversal runs on, so the CPU arm
of our implementation needs no CFR-specific threading of its own (\cref{sec:setup}).

Accelerators already drive game \emph{simulation} at scale \cite{koyamada2023pgx}.
GPU CFR has recent precedent
\cite{kim2026parallelizing,DBLP:journals/corr/abs-2508-06559}, and the GPU path of
the fastest such system still runs slower than an optimized CPU solver. We instead compile
the game into depth-stratified dense tensors and replay the whole iteration as one
captured CUDA graph. This compilation step is what makes stock PyTorch
\cite{paszke2019pytorch} and NVIDIA CUDA graphs \cite{guide2020cuda} applicable to
CFR at all.

%% file: sections/3_preliminaries.tex
\section{Preliminaries}
\label{sec:prelim}

A two-player zero-sum extensive-form game is a finite tree of \emph{histories}, also
called \emph{states} \cite{koller1997representations,zinkevich2007regret}. A
history $h \in \nodes$ is the sequence of all actions taken from the root
$\varnothing$, including chance's. In poker it fixes both players' private cards,
the public cards, and the betting sequence. Performing an action $a$ at a non-terminal
history $h$ leads to the history $h \cdot a$. Each non-terminal history is owned by
player 1, player 2, or chance; chance histories have fixed action distributions. A
terminal history $z \in \termnodes \subset \nodes$ has no actions and pays $u(z)$ to
player~1 (player~2 receives $-u(z)$).

Player $i$ cannot distinguish histories that differ only in what is hidden from $i$.
An \emph{information set} (Infoset) $I \in \infosets_i$ is a maximal set of histories that are
indistinguishable to $i$. They share $i$'s private information and the public
sequence, and differ only in the opponent's private information. Every history $h$,
terminal or not, therefore lies in exactly one information set of each player,
$I_1(h)$ and $I_2(h)$. In poker $I_i(h)$ is $i$'s hand, the board, and the betting. Histories in the same $I$ owned by $i$ share an action set
$\acts(I)$. We assume perfect recall \cite{zinkevich2007regret}. A \emph{public node} $s(h)$
\cite{burch2014solving,moravvcik2017deepstack} collects the
histories that share everything both players can see. These objects give the two size counts used throughout. \emph{States} is the
number of non-chance histories, terminals included. \emph{Infosets} is the number of
pairs $(i, I_i(h))$ over those histories, again terminals included.
A behavioral strategy $\strat_i$ assigns each $I \in \infosets_i$ owned by $i$ a
distribution over $\acts(I)$. For a profile $\strat = (\strat_1, \strat_2)$, the reach
probability $\reach^{\strat}(h)$ of history $h$ factors into player-1, player-2, and
chance contributions. The counterfactual reach $\creach(h)$ excludes player $i$'s own
contribution. We report exploitability as NashConv \cite{lanctot2019openspiel} divided
by two, computed by exact best response \cite{johanson2011accelerating}:
\begin{equation}
    \expl(\strat) = \frac{1}{2} \sum_i \left[\max_{\strat_i'} u_i(\strat_i', \strat_{-i}) - u_i(\strat)\right].
    \label{eq:exploitability}
\end{equation}

CFR \cite{zinkevich2007regret} is an iterative self-play procedure over this tree.
Iteration $t$ holds a current profile $\strat^t$, a cumulative regret $\regret^t(I,a)$
for every information-set action, and an average-strategy accumulator. Four steps
turn $\strat^t$ into $\strat^{t+1}$. First, a \emph{forward reach pass} propagates
reach probabilities from the root, one product per edge:
\begin{equation}
    \reach^{\strat}(h\cdot a) = \reach^{\strat}(h)\,\strat_{P(h)}(I(h),a),
    \label{eq:reach}
\end{equation}
where $P(h)$ is the owner of $h$ and chance edges use their fixed probabilities.
Second, a \emph{backward value pass} folds terminal payoffs up the tree, one
strategy-weighted sum per node, and aggregates them into counterfactual values that
weight each history by every contribution except player $i$'s own:
\begin{equation}
\begin{aligned}
    \cfv_i(I,a) &= \sum_{h\in I}\creach(h)\sum_{z\in\termnodes}\reach^{\strat}(h\cdot a, z)\,u_i(z),\\
    \cfv_i(I) &= \sum_{a\in\acts(I)}\strat_i(I,a)\,\cfv_i(I,a).
\end{aligned}
    \label{eq:cfv}
\end{equation}
Third, the solver adds the instantaneous regret of each action to its accumulator:
\begin{equation}
    \regret^{t}(I,a) = \regret^{t-1}(I,a) + \cfv_i^{t}(I,a) - \cfv_i^{t}(I).
    \label{eq:regret}
\end{equation}
Fourth, regret matching \cite{hart2000simple} sets the next strategy proportional to
positive regret, and the reach-weighted current strategy is folded into the average:
\begin{equation}
\begin{aligned}
    \strat^{t+1}(I,a) &= \frac{\regret^{t}_{+}(I,a)}{\sum_{a'}\regret^{t}_{+}(I,a')},\\
    \avgstrat^{T}(I,a) &\propto \sum_{t\le T} w_t\,\reach_i^{\strat^t}(I)\,\strat^t(I,a),
\end{aligned}
    \label{eq:rm}
\end{equation}
with $\regret_{+}=\max(\regret,0)$, a uniform strategy when the denominator is zero,
and $w_t=1$ for uniform or $w_t=t$ for linear averaging \cite{tammelin2014solving}.
The average $\avgstrat^{T}$ converges to a Nash equilibrium at rate
$O(1/\sqrt{T})$ \cite{zinkevich2007regret}. CFR$^{+}$ \cite{tammelin2014solving} changes one line. It clamps the accumulator
itself, $\regret^{t}\gets\max(\regret^{t},0)$ after \cref{eq:regret}, and
converges much faster in practice \cite{tammelin2015solving,burch2019revisiting}. DCFR \cite{brown2019solving} and PCFR$^{+}$
\cite{farina2021faster} (\cref{tab:convergence}) keep this loop and change only how
the two accumulators evolve. In systems terms, every rule is the same sparse computation each iteration: two
sweeps over a fixed irregular tree plus a segment-normalized update per information
set, with only the numbers changing. Players may be updated \emph{simultaneously}
from one strategy profile or \emph{alternately}, where the first player's update is
visible to the second's in the same sweep \cite{burch2019revisiting}. An alternating
iteration is one half-update per player and therefore roughly twice the work. GPU-CFR
supports both orders and both uniform and linear strategy averaging. The timing
matrix fixes these choices for matched comparisons (\cref{sec:setup}).

%% file: sections/4_method.tex
\section{Compiling CFR onto Accelerators}
\label{sec:method}

The design follows one observation: once the game is fixed, the entire structure of
a CFR iteration is known at build time, and only the numbers change. We therefore split the solver into a one-time \emph{compilation} of the game into
flat arrays (\cref{sec:compiled}) and a per-iteration \emph{dataflow} over those
arrays whose operation sequence never changes (\cref{sec:dataflow}). The fixed
sequence makes the iteration a legal target for CUDA-graph capture
(\cref{sec:cudagraph}). \Cref{fig:archpipeline} shows the pipeline on an 11-node
example game with a chance root, two decisions per player, and six terminals, and
\cref{fig:archdataflow} reuses the same game. A layered verification design checks
the compiled schedule, graph replay, and evaluator (\cref{sec:correctness}).

\input{sections/fig_arch_pipeline}
\input{sections/fig_arch_dataflow}
\subsection{Compiled game representation}
\label{sec:compiled}

Compilation walks the game once, numbering nodes in breadth-first order so that
every parent precedes its children, and emits:
\begin{itemize}
    \item \textbf{Node arrays} of length $\numnodes$: owner (player 1/2, chance,
    terminal), depth, and a terminal payoff vector.
    \item \textbf{Edge arrays} of length $\numedges$: for every edge, its parent node,
    child node, and \emph{slot}, the index of the corresponding (information set,
    action) pair in a flat slot space of size $\numslots = \sum_{I} |\acts(I)|$.
    Chance edges carry their fixed probability and no slot.
    \item \textbf{Information-set arrays} of length $\numinfosets$ and $\numslots$:
    per-set action counts and slot offsets, so regret matching operates on flat
    vectors with segment reductions and never touches a ragged structure.
\end{itemize}
All solver state (accumulated regrets $\regret \in \R^{\numslots}$, strategy sums
$\bar{s} \in \R^{\numslots}$, and the iteration scratch below) lives in
preallocated device tensors. On solver construction, root reach entries are set to
one and every buffer is written before it is read. The scratch footprint is
$(5\numnodes + \numslots + \numinfosets + 1)$ floats plus the integer index arrays.
Measured peak GPU memory on the largest game is 183\,MiB allocated and 236\,MiB
reserved, including the CUDA-graph pool and evaluation scratch (\cref{tab:memory}).
That is about 570 bytes per tree node. Float scratch, persistent slot vectors, and
the int64 edge-index arrays all scale linearly in $\numnodes$, $\numedges$, and
$\numslots$, with $\numedges \approx \numnodes$ in trees. The same code path runs
unchanged on CPU and GPU.

\subsection{One iteration as a fixed dataflow}
\label{sec:dataflow}

A CFR iteration is regret matching, a forward reach pass, a backward value pass with
instantaneous-regret accumulation, and the update rule
(\crefrange{eq:reach}{eq:rm}). \Cref{fig:archdataflow} shows the four phases as they
execute. Three structural observations turn the two tree passes into a handful of
batched tensor operations each. Throughout, ``overwrite'' semantics are available
because the game is a tree. Every non-root node has exactly one incoming edge, so a
scatter along edges writes each destination exactly once per iteration, and buffers
never need zeroing within a solve.

\paragraph{(a) Static chance folding.}
Chance behavior is strategy-independent, so for every node $h$ the product of chance
probabilities on the root-to-$h$ path, $\pi_c(h)$, is a build-time constant. We fold
it into a \emph{values template}
\begin{equation}
    \cfv_{\mathrm{tmpl}}(h) \;=\;
    \begin{cases}
        \pi_c(h)\, u(h) & h \in \termnodes,\\
        0 & \text{otherwise},
    \end{cases}
    \label{eq:template}
\end{equation}
from which the backward pass starts every iteration. The invariant is that chance probabilities appear \emph{only} through
\cref{eq:template}. During the dynamic passes, every chance edge acts as multiplier
$1$ (via the sentinel slot below). The backward recurrence
$\cfv(h) = \sum_{(h,a,h')} \strat_{\mathrm{ext}}(h,a)\, \cfv(h')$
then returns at each node $h$ the \emph{chance-weighted} continuation value
$\sum_{z \succeq h} \pi_c(z)\, \reach^{\strat}_{1,2}(h \to z)\, u(z)$. There is
no double counting, because each terminal's full-path chance product enters exactly
once, at the template. This removes every chance multiplication from the
per-iteration forward pass. Chance edges remain in the depth blocks only as sentinel
reads.

\paragraph{(b) Depth-level execution blocks.}
The forward pass must respect topology, since a child's reach depends on its
parent's. In a tree, every child sits exactly one depth level below its parent, so
grouping edges by \emph{parent depth} yields a valid schedule at the coarsest
granularity a depth-synchronous schedule allows. All edges at one depth execute as a
single batched gather, multiply, and scatter, and the number of sequential steps
equals the tree depth. Across our suite this is 4--15 blocks per pass
(\cref{tab:suite}). A generic topological schedule over (stage $\times$ depth)
intersections produced up to 100 blocks on the same trees. The backward pass uses the same blocks
in reverse with \texttt{index\_add} scatters, which tolerate repeated parents within
a block.

\paragraph{(c) Sentinel slot and dual-lane reach buffer.}
Player $i$'s reach lane $\reach_i$ must multiply $i$'s own strategy entries and copy
through everything else:
\begin{equation}
    \reach_i(h') =
    \begin{cases}
        \reach_i(h)\, \strat(\mathrm{slot}(h,a)) & (h,a,h') \text{ owned by } i,\\
        \reach_i(h) & \text{otherwise}.
    \end{cases}
    \label{eq:lanes}
\end{equation}
We implement \cref{eq:lanes} branchlessly. The strategy vector is extended by one \emph{sentinel} slot pinned to $1.0$. For
each edge and each lane, a precomputed index points either at the edge's true slot
(own-player edges) or at the sentinel (all others). The reach buffer holds both lanes contiguously ($2\numnodes$ entries), and
per-block index arrays over the doubled range advance both lanes with a single
gather, multiply, and scatter. The same trick removes the last per-edge branch from regret accumulation. For a
player-$i$ edge $(h, a)$ with slot $q = \mathrm{slot}(h,a)$, the instantaneous
regret contribution is
\begin{equation}
    r(q) \;\mathrel{+}=\; s(e)\,\reach_{-i}(h)\,\big(\cfv(h') - \cfv(h)\big),
    \label{eq:instregret}
\end{equation}
where $e=(h,a,h')$, $s(e)=+1$ for player 1 and $-1$ for player 2,
$\cfv$ stores player-1 values, and $\reach_{-i}(h)$ is read from the
opponent's lane through a precomputed index array. Summing \cref{eq:instregret} over an information-set action with a slot-indexed
\texttt{index\_add} yields $\sum_{h\in I}s(e)\reach_{-i}(h)(\cfv(h')-\cfv(h))$,
the standard counterfactual regret. Chance reach is already in $\cfv$ through
\cref{eq:template} and is not multiplied again.

Per-block index arrays $P_2, C_2, S_2$ are the precomputed dual-lane parent/child/slot
arrays of block $b$, covering every edge (chance and other-player entries point at the
sentinel). $P_{\mathrm{val}}$, $C_{\mathrm{val}}$, $S_{\mathrm{val}}$ cover every edge of the
backward value pass; $P_{\mathrm{reg}}$, $C_{\mathrm{reg}}$, $S_{\mathrm{reg}}$,
$O_{\mathrm{reg}}$
cover player decision edges only, with $O_{\mathrm{reg}}$ indexing the opponent-reach lane at
the parent and $M_{\mathrm{reg}}$ the acting player's own lane there, used by the
averaging update. Because $\cfv$ stores player-1 values, a build-time sign
$s(e) \in \{\pm 1\}$ orients each decision edge's value difference toward its
acting player; $s(b)$ gathers these signs for the block's decision edges.
Chance edges and sentinel slots never receive regret updates.
$\strat_{\mathrm{ext}}$ denotes the strategy vector with the sentinel entry fixed at $1$.

\begin{algorithm}[t]
\caption{One compiled CFR$^{+}$ iteration under simultaneous updates.}
\label{alg:iteration}
\begin{algorithmic}[1]
\State $\strat \gets \mathrm{regret\_matching}(\regret)$ \Comment{flat over $\numslots$ slots; segment-normalized}
\For{$b = 1, \dots, D$} \Comment{forward, both lanes at once, \cref{eq:lanes}}
    \State $\reach[C_2(b)] \gets \reach[P_2(b)] \cdot \strat_{\mathrm{ext}}[S_2(b)]$
\EndFor
\State $\cfv \gets \cfv_{\mathrm{tmpl}}$ \Comment{chance-weighted payoffs, \cref{eq:template}}
\For{$b = D, \dots, 1$} \Comment{backward values}
    \State $\cfv.\mathrm{index\_add}\big(P_{\mathrm{val}}(b),\; \strat_{\mathrm{ext}}[S_{\mathrm{val}}(b)] \cdot \cfv[C_{\mathrm{val}}(b)]\big)$
    \State $\regret.\mathrm{index\_add}\big(S_{\mathrm{reg}}(b),\; s(b) \cdot \reach[O_{\mathrm{reg}}(b)] \cdot (\cfv[C_{\mathrm{reg}}(b)] - \cfv[P_{\mathrm{reg}}(b)])\big)$ \Comment{\cref{eq:instregret}, decision edges only}\label{alg:regretline}
\EndFor
\State $\bar{s} \mathrel{+}= w_t \cdot \reach_{\mathrm{own}} \cdot \strat$;\qquad
       $\regret \gets \max(\regret, 0)$ \Comment{update rule (here: CFR$^{+}$)}
\end{algorithmic}
\end{algorithm}

\cref{alg:iteration} summarizes the iteration, where $w_t$ is the averaging weight.
Alternating updates \cite{burch2019revisiting} run its body twice per iteration,
once per player, and restrict the regret update and the regret matching$^{+}$ clamp
to the acting player's slots. The player order is static, so every index array and
the captured graph stay the same. \Cref{tab:opsteps} lists the realized operation
schedule. One realization detail: the regret scatter of line~\ref{alg:regretline}
is hoisted out of the level loop and applied once, flat over all decision edges after
the value pass. The average-strategy accumulation shares the same flat scatter. The
parent and child values it reads are final by then, so the flat scatter accumulates
the same summands in a different grouping. Only the forward and value recurrences contribute
to the depth-proportional term. The result is a fixed sequence of eight Aten
\cite{paszke2019pytorch} operations per depth level plus a constant part: 64--152
operations per iteration across the suite. A reference implementation of the same
algorithm with per-(stage, depth) blocks, masked chance handling, and per-iteration
allocations needs 110--1{,}742 (\cref{tab:suite}). The Aten-operation count is
exactly reproducible and independent of machine load. It therefore serves both as a
regression gate (\cref{sec:correctness}) and as a benchmarking instrument.

\input{sections/tab_opsteps}

\subsection{CUDA-graph execution}
\label{sec:cudagraph}

Even at 96 operations per iteration, a 434k-state game runs each GPU kernel for only
microseconds, so framework dispatch and kernel launches dominate wall-clock time. The
compiled dataflow executes eagerly at \eagerMsRange\,ms per iteration on an A100
across our suite, largely independent of tree size, the signature of a launch-bound
workload. Because the operation sequence, tensor shapes, and buffer
addresses are all fixed, the entire iteration qualifies for CUDA-graph capture
\cite{guide2020cuda}: record the kernel sequence once, then replay it with one graph
launch per iteration. CUDA graphs remove per-kernel launch and dispatcher overhead
and leave the kernels themselves unchanged, so the same kernels run back to back.
\Cref{app:profile} measures hand-fused kernels and \cref{app:compiler} lists the
compiler passes.

Three details make capture faithful to eager execution. First, CFR$^{+}$'s
linear-averaging weight $w_t = t$ changes every iteration, which a static graph
cannot see. We keep $w$ in a 0-dimensional device tensor that the captured body
increments in place after use. Before each replay batch the host refills it with the
exact starting index, so replay $k$ uses precisely $w = t_0 + k$ and matches eager
execution bit for bit. The counter is float32 in every reported run and represents
$w$ exactly up to $2^{24}$ iterations. Second, capture must be preceded by warmup
iterations on a side stream. In GPU-CFR these are genuine eager iterations that count
toward the iteration budget, and the accounting is preserved across arbitrarily split
solver calls. Third, replaying a graph whose buffers have been reallocated would
silently read stale memory. The solver stores the data pointers of all eight
persistent buffers at capture time and refuses to replay if any has changed. If
capture itself fails, for example on an unsupported driver, the solver warns once and
falls back to eager execution. The fallback runs the same kernels in the same order
and, within one process, computes the identical floats.

Graph replay cuts steady-state iteration time by \eagerGraphRange$\times$
depending on tree size (\cref{tab:matrix}, \cref{app:extra}). The two largest games
run at \graphMsTwoLargest\,ms per iteration. One-time capture cost
(0.037--0.301\,s) lands in the first solver call, and \cref{sec:results} reports
steady-state and end-to-end numbers separately.

\subsection{Layered verification}
\label{sec:correctness}

Verification proceeds at three levels. First, graph replay and compiled eager
execution run the same kernels in the same order and agree bitwise over 30 iterations
within one process.
Second, the compiled dataflow is compared with the pre-optimization implementation.
The two agree bitwise in float32 after 22 iterations on all eight benchmark games
(\cref{tab:suite}; CPU, fixed reduction order) and to $10^{-12}$ in float64 over 30
iterations. Third, independently written implementations provide tolerance-level
checks on values and exploitability.

GPU scatter reductions can accumulate in different orders across independent
processes. Near an indifferent action a last-bit difference can separate subsequent
trajectories: on the HUNL river subgame, 4 of 20 identical runs reach a second
exploitability cluster $0.6\%$ away. \Cref{app:correctness} traces this to
reduction order and quantifies its effect. The verification layers are:
\begin{enumerate}
    \item \textbf{Reference-oracle parity.} A complete copy of the solver from
    before optimization is preserved inside the test suite and compared as above,
    across $\{$CFR, CFR$^{+}\} \times \{$uniform, linear$\}$ averaging.
    \item \textbf{Independent reference solver.} A per-node Python implementation,
    sharing no code with the compiled path, must agree to $10^{-9}$ on the poker
    subgames.
    \item \textbf{Best-response oracles and analytic anchors.} An exact recursive best-response oracle \cite{johanson2011accelerating}, checked
    by brute-force strategy enumeration on small instances, validates the evaluator.
    Analytic cases pin absolute values.
    \item \textbf{Operation-count regression gates.} The per-iteration count of
    Aten operations is asserted against the affine schedule size $c_1 + c_2 D$ of
    \cref{sec:theory}, and graph replay is compared with eager execution over 30
    iterations. These deterministic checks complement wall-clock measurements.
\end{enumerate}
The full test inventory (499 tests) appears in \cref{app:correctness}.

%% file: sections/fig_arch_pipeline.tex
\begin{figure}[!htb]
\begin{minipage}[b]{0.50\linewidth}\centering
\resizebox{\linewidth}{!}{%
\begin{tikzpicture}[
    font=\sffamily\small,
    dec1/.style={circle, draw=teal!60!black, fill=teal!45, minimum size=3.4mm, inner sep=0pt},
    dec2/.style={circle, draw=violet!60!black, fill=violet!40, minimum size=3.4mm, inner sep=0pt},
    chn/.style={diamond, draw=black!65, fill=black!22, minimum size=4.2mm, inner sep=0pt},
    ter/.style={rectangle, draw=black!70, fill=black!55, minimum size=2.5mm, inner sep=0pt},
    edge/.style={black!55, line width=0.9pt},
    flow/.style={-{Stealth[length=2.2mm]}, line width=1.0pt, black!60},
    fold/.style={-{Stealth[length=1.8mm]}, densely dashed, line width=0.8pt, black!55},
    mapline/.style={line width=1.0pt, opacity=0.80, shorten <=1.4pt},
    tick/.style={orange!80!black, line width=1.0pt},
    gtick/.style={orange!95!black, line width=1.6pt},
    note/.style={font=\sffamily\small, text=black!55},
    lanelabel/.style={font=\sffamily\small\bfseries, text=black!60}
]
\foreach \y in {0, -1.7}
  \fill[blue!6] (0.10,\y+0.33) rectangle (7.95,\y-0.33);
\node[chn]  (c0)  at (4.60, 0.00) {};
\node[dec1] (p1a) at (2.60,-0.85) {};
\node[dec1] (p1b) at (6.60,-0.85) {};
\node[dec2] (q2a) at (1.70,-1.70) {};
\node[ter]  (t2a) at (3.50,-1.70) {};
\node[dec2] (q2b) at (5.70,-1.70) {};
\node[ter]  (t2b) at (7.60,-1.70) {};
\node[ter]  (t3a) at (1.30,-2.55) {};
\node[ter]  (t3b) at (2.30,-2.55) {};
\node[ter]  (t3c) at (5.20,-2.55) {};
\node[ter]  (t3d) at (6.20,-2.55) {};
\draw[edge, densely dashed] (c0) -- (p1a);
\draw[edge, densely dashed] (c0) -- (p1b);
\draw[edge] (p1a) -- (q2a); \draw[edge] (p1a) -- (t2a);
\draw[edge] (p1b) -- (q2b); \draw[edge] (p1b) -- (t2b);
\draw[edge] (q2a) -- (t3a); \draw[edge] (q2a) -- (t3b);
\draw[edge] (q2b) -- (t3c); \draw[edge] (q2b) -- (t3d);
\node[chn, minimum size=3.2mm] (legchance) at (0.55,0.62) {};
\node[note, anchor=west] at (0.98,0.62) {chance};
\node[dec1, minimum size=3.0mm] (legp1) at (2.30,0.62) {};
\node[note, anchor=west] at (2.73,0.62) {P1};
\node[dec2, minimum size=3.0mm] (legp2) at (3.48,0.62) {};
\node[note, anchor=west] at (3.91,0.62) {P2};
\node[ter, minimum size=2.4mm] (legterm) at (4.70,0.62) {};
\node[note, anchor=west] at (5.13,0.62) {terminal node};

\def\cw{0.62}
\def\sx{0.72}
\def\slotb{-4.17}   
\def\tmplb{-5.87}   
\foreach \n/\i/\c in {c0/0/black!50, p1a/1/teal!75, p1b/2/teal!75,
                      q2a/3/violet!75, t2a/4/black!50, q2b/5/violet!75,
                      t2b/6/black!50, t3a/7/black!50, t3b/8/black!50,
                      t3c/9/black!50, t3d/10/black!50}
  \draw[mapline, \c] (\n.south) -- (\n.south |- 0,-2.85)
    .. controls (\n.south |- 0,-3.25) and ({\sx+(\i+0.5)*\cw},-3.10) ..
    ({\sx+(\i+0.5)*\cw},\slotb+\cw);
\node[note, align=left, anchor=west] at (0.00,-3.02) {compile\\walk once};
\node[note, anchor=east] at (\sx-0.10,\slotb+0.5*\cw) {nodes};
\foreach [count=\i from 0] \c in
  {black!22, teal!45, teal!45, violet!40, black!55, violet!40, black!55,
   black!55, black!55, black!55, black!55}{
  \draw[black!60, fill=\c] ({\sx+\i*\cw},\slotb) rectangle ++(\cw,\cw);
}
\draw[black!60, densely dashed, fill=white, fill opacity=0.55]
  (\sx,\slotb) rectangle ++(\cw,\cw);
\foreach \s in {0, 1, 3, 7, 11}
  \draw[black!85, line width=1.2pt] ({\sx+\s*\cw},\slotb-0.05) -- ++(0,\cw+0.10);
\foreach \d/\a/\b in {0/0/1, 1/1/3, 2/3/7, 3/7/11}
  \node[note] at ({\sx+(\a+\b)*\cw/2},\slotb-0.26) {$d{=}\d$};
\foreach \a/\b in {1/3, 3/7, 7/11}
  \draw[flow, line width=0.9pt, -{Stealth[length=1.8mm]}]
    ({\sx+\a*\cw+0.06},\slotb-0.56) -- ({\sx+\b*\cw-0.06},\slotb-0.56);
\node[note] at ({\sx+9*\cw},\slotb-0.86) {$D$ blocks};
\node[note, anchor=east] at (\sx-0.10,\tmplb+0.5*\cw) {$\cfv_{\mathrm{tmpl}}$};
\foreach [count=\i from 0] \c in
  {white, white, white, white, black!55, white, black!55,
   black!55, black!55, black!55, black!55}{
  \draw[black!60, fill=\c] ({\sx+\i*\cw},\tmplb) rectangle ++(\cw,\cw);
}
\foreach \s in {0, 1, 3, 7, 11}
  \draw[black!85, line width=1.2pt] ({\sx+\s*\cw},\tmplb-0.05) -- ++(0,\cw+0.10);
\draw[fold] ({\sx+0.04},\slotb)
  .. controls (0.28,-4.60) and (0.28,-4.95) .. ({\sx+0.08},\tmplb+\cw)
  node[note, pos=0.5, left=0.6mm] {fold};

\begin{scope}[yshift=-0.10cm]
\def\wzero{1.66}
\def\itw{1.62}
\def\itg{0.14}
\foreach \i/\f in {0/black!5, 1/black!0, 2/black!5}
  \fill[\f] ({\wzero+\i*(\itw+\itg)},-7.30) rectangle ++(\itw,-1.15);
\node[note] at ({\wzero+0.5*\itw},-7.46) {one iteration};
\node[note, anchor=east] at (0.72,-7.78) {eager};
\draw[black!35] (\wzero,-7.78) -- ({\wzero+3*\itw+2*\itg},-7.78);
\foreach \i in {0,1,2}{
  \foreach \k in {0,...,8}{
    \draw[tick] ({\wzero+0.18+\i*(\itw+\itg)+\k*0.17},-7.65) -- ++(0,-0.26);
  }
}
\node[note, anchor=east] at (0.72,-8.25) {graph};
\draw[black!35] (\wzero,-8.25) -- ({\wzero+3*\itw+2*\itg},-8.25);
\node[draw=orange!80!black, fill=orange!25, rounded corners=1.5pt,
      minimum height=3.4mm, inner sep=1.6pt, anchor=east]
  at ({\wzero-0.08},-8.25) {record};
\foreach \i in {0,1,2}{
  \draw[gtick] ({\wzero+0.18+\i*(\itw+\itg)},-8.06) -- ++(0,-0.38);
}
\draw[flow, line width=0.9pt]
  ({\wzero+3*\itw+2*\itg},-8.25) -- ++(0.26,0);
\node[note, align=left, anchor=west] at (7.10,-8.25) {replay\\$\times T$};
\draw[-{Stealth[length=1.6mm]}, black!40] (5.50,-8.65) -- (6.70,-8.65)
  node[note, pos=0.5, above=0.2mm] {time};
\end{scope}

\begin{scope}[on background layer]
\node[draw=blue!45!black!30, rounded corners=3pt, inner sep=1.2mm, fill=blue!3,
      fit={(-0.22,0.88) (8.14,-6.35)}] (buildlane) {};
\node[draw=orange!70!black!30, rounded corners=3pt, inner sep=1.2mm, fill=orange!4,
      fit={(-0.22,-7.15) (8.14,-9.00)}] (runlane) {};
\end{scope}
\node[lanelabel, anchor=south west] at (buildlane.north west)
  {Build: once per game};
\node[lanelabel, anchor=south west] at (runlane.north west)
  {Run: every iteration};
\end{tikzpicture}}
\caption{The compilation pipeline on an 11-node example game. A build-time walk
flattens the tree into the depth-segmented node array and values template
(colored lines trace each node to its cell); at run time, eager execution
submits $c_1 + c_2 D$ framework operations per iteration (each one or more kernel
launches), graph replay one graph launch.}
\Description{Two-lane diagram. The build-time lane shows a small game tree with
a diamond chance root, colored decision nodes, and square terminals over
alternating depth bands; thick colored lines funnel every node into its
breadth-first cell of a flat array strip divided into depth segments, above a
values-template strip whose terminal cells are pre-filled; a short dashed arrow
folds the chance slot into that template. The run-time lane shows two aligned
timelines over three shaded iteration windows: the eager row has a dense comb
of launch ticks in every window, while the graph row has one record block
followed by a single tick per window and an arrow labeled replay times T.}
\label{fig:archpipeline}
\end{minipage}\hfill%

%% file: sections/fig_arch_dataflow.tex
\begin{minipage}[b]{0.46\linewidth}\centering
\resizebox{\linewidth}{!}{%
\begin{tikzpicture}[
    font=\sffamily\small,
    flow/.style={-{Stealth[length=2.6mm]}, line width=1.4pt, black!55},
    arr/.style={-{Stealth[length=1.6mm]}, line width=0.9pt, black!60},
    carc/.style={-{Stealth[length=1.6mm]}, line width=1.0pt, black!50},
    tarc/.style={-{Stealth[length=1.6mm]}, line width=1.0pt, teal!70!black},
    varc/.style={-{Stealth[length=1.6mm]}, line width=1.0pt, violet!70!black},
    sarr/.style={-{Stealth[length=1.5mm]}, line width=0.9pt, orange!75!black},
    note/.style={font=\sffamily\small, text=black!55},
    phaselabel/.style={font=\sffamily\small\bfseries, text=black!65},
    badge/.style={circle, fill=orange!85!black, text=white,
                  font=\sffamily\small\bfseries, inner sep=1.0pt}
]
\def\cw{0.58}
\def\sx{0.82}
\newcommand{\slotstrip}[3]{
  \foreach \i in {0,...,7}
    \draw[black!60, fill=#3] ({#1+\i*\cw},#2) rectangle ++(\cw,\cw);
}
\newcommand{\nodestrip}[3]{
  \foreach [count=\i from 0] \c in
    {#3, #3, #3, #3, black!55, #3, black!55, black!55, black!55, black!55, black!55}
    \draw[black!60, fill=\c] ({#1+\i*\cw},#2) rectangle ++(\cw,\cw);
  \foreach \s in {0, 1, 3, 7, 11}
    \draw[black!85, line width=1.0pt] ({#1+\s*\cw},#2-0.03) -- ++(0,\cw+0.06);
}
\newcommand{\dualstrip}[2]{
  \foreach \i in {0,...,10}{
    \fill[violet!25] ({#1+\i*\cw},#2) rectangle ++(\cw,0.5*\cw);
    \fill[teal!28]   ({#1+\i*\cw},{#2+0.5*\cw}) rectangle ++(\cw,0.5*\cw);
    \draw[black!60]  ({#1+\i*\cw},#2) rectangle ++(\cw,\cw);
  }
  \foreach \s in {0, 1, 3, 7, 11}
    \draw[black!85, line width=1.0pt] ({#1+\s*\cw},#2-0.03) -- ++(0,\cw+0.06);
}

\node[badge] at (1.00,13.30) {1};
\node[phaselabel, anchor=west] at (1.22,13.30) {regret matching};
\node[note, anchor=east] at (\sx-0.10,12.74) {$\regret$};
\slotstrip{\sx}{12.45}{orange!30}
\node[note, anchor=east] at (\sx-0.10,11.64) {$\strat_{\mathrm{ext}}$};
\slotstrip{\sx}{11.35}{green!25}
\draw[black!60, fill=black!45] ({\sx+8*\cw},11.35) rectangle ++(\cw,\cw);
\node[note] at ({\sx+8.5*\cw},11.13) {$\bot$};
\foreach \i in {0,...,7}
  \draw[arr] ({\sx+(\i+0.5)*\cw},12.40) -- ({\sx+(\i+0.5)*\cw},11.98);
\node[note, anchor=west] at ({\sx+8.3*\cw},12.19) {RM$^{+}$};
\draw[flow] (3.25,11.18) -- (3.25,10.78);

\node[badge] at (1.00,10.66) {2};
\node[phaselabel, anchor=west] at (1.22,10.66) {forward reach $\times D$};
\node[note, anchor=east] at (\sx-0.10,9.31) {$\reach_{1}$};
\node[note, anchor=east] at (\sx-0.10,9.00) {$\reach_{2}$};
\dualstrip{\sx}{8.85}
\newcommand{\fwdarc}[6]{
  \draw[#5] ({\sx+(#1+#2)*\cw},9.49)
    .. controls ({\sx+(#1+#2)*\cw},9.49+#4) and ({\sx+(#3+0.5)*\cw},9.49+#4) ..
    ({\sx+(#3+0.5)*\cw},9.49);
  \fill[#6, draw=black!70, line width=0.5pt]
    ({(\sx+(#1+#2)*\cw)/2+(\sx+(#3+0.5)*\cw)/2-0.08},{9.49+0.75*#4-0.08})
    rectangle ++(0.16,0.16);
}
\fwdarc{0}{0.35}{1}{0.45}{carc}{black!45}
\fwdarc{0}{0.65}{2}{0.60}{carc}{black!45}
\fwdarc{1}{0.35}{3}{0.45}{tarc}{green!45}
\fwdarc{1}{0.65}{4}{0.60}{tarc}{green!45}
\fwdarc{2}{0.35}{5}{0.75}{tarc}{green!45}
\fwdarc{2}{0.65}{6}{0.90}{tarc}{green!45}
\fwdarc{3}{0.35}{7}{0.45}{varc}{green!45}
\fwdarc{3}{0.65}{8}{0.60}{varc}{green!45}
\fwdarc{5}{0.35}{9}{0.75}{varc}{green!45}
\fwdarc{5}{0.65}{10}{0.90}{varc}{green!45}
\draw[flow] (3.25,8.68) -- (3.25,8.28);

\node[badge] at (1.00,8.05) {3};
\node[phaselabel, anchor=west] at (1.22,8.05) {backward values $\times D$};
\node[note, anchor=east] at (\sx-0.10,7.54) {$\cfv_{\mathrm{tmpl}}$};
\nodestrip{\sx}{7.25}{white}
\node[note, anchor=east] at (\sx-0.10,6.54) {$\cfv$};
\nodestrip{\sx}{6.25}{gray!15}
\draw[arr] ({\sx+4.5*\cw},7.20) -- node[note, right=0.4mm] {copy} ({\sx+4.5*\cw},6.88);
\newcommand{\bwdarc}[6]{
  \draw[#5] ({\sx+(#1+0.5)*\cw},6.21)
    .. controls ({\sx+(#1+0.5)*\cw},6.21-#4) and ({\sx+(#2+#3)*\cw},6.21-#4) ..
    ({\sx+(#2+#3)*\cw},6.21);
  \fill[#6, draw=black!70, line width=0.5pt]
    ({(\sx+(#1+0.5)*\cw)/2+(\sx+(#2+#3)*\cw)/2-0.08},{6.21-0.75*#4-0.08})
    rectangle ++(0.16,0.16);
}
\bwdarc{1}{0}{0.35}{0.45}{carc}{black!45}
\bwdarc{2}{0}{0.65}{0.60}{carc}{black!45}
\bwdarc{3}{1}{0.35}{0.45}{tarc}{green!45}
\bwdarc{4}{1}{0.65}{0.60}{tarc}{green!45}
\bwdarc{5}{2}{0.35}{0.75}{tarc}{green!45}
\bwdarc{6}{2}{0.65}{0.90}{tarc}{green!45}
\bwdarc{7}{3}{0.35}{0.45}{varc}{green!45}
\bwdarc{8}{3}{0.65}{0.60}{varc}{green!45}
\bwdarc{9}{5}{0.35}{0.75}{varc}{green!45}
\bwdarc{10}{5}{0.65}{0.90}{varc}{green!45}
\draw[flow] (3.25,5.00) -- (3.25,4.68);

\node[badge] at (1.00,4.50) {4};
\node[phaselabel, anchor=west] at (1.22,4.50) {regret update, flat};
\node[note, anchor=east] at (\sx-0.10,3.99) {$\regret$};
\slotstrip{\sx}{3.70}{orange!30}
\node[note, anchor=east] at (\sx-0.10,2.59) {$\cfv$};
\nodestrip{\sx}{2.30}{gray!15}
\foreach \c/\s in {3/0, 4/1, 5/2, 6/3, 7/4, 8/5, 9/6, 10/7}
  \draw[sarr] ({\sx+(\c+0.5)*\cw},2.92)
    to[bend left=12] ({\sx+(\s+0.5)*\cw},3.66);

\draw[flow] (-0.30,4.30) -- (-0.30,12.50)
  node[note, midway, rotate=90, above=0.6mm] {next iteration};
\begin{scope}[on background layer]
\node[draw=black!30, rounded corners=4pt, fill=black!2, inner sep=1.2mm,
      fit={(-0.62,13.50) (7.38,1.92)}] (graphframe) {};
\end{scope}
\node[note, anchor=south east] at ($(graphframe.south east)+(-0.12,-0.02)$)
  {one CUDA graph};
\end{tikzpicture}}
\caption{One compiled CFR$^{+}$ iteration on the example game of
\cref{fig:archpipeline}, stacked in execution order inside the captured CUDA
graph (rounded frame). Arcs are colored by the acting owner of their edge (gray
chance, teal P1, violet P2); each apex square is the gathered strategy factor
(green) or the sentinel (dark).}
\Description{Four phases stacked top to bottom inside one rounded frame labeled
one CUDA graph. Phase one shows eight parallel arrows from an orange regret
strip to a green strategy strip with an extra dark sentinel cell that receives
no arrow. Phase two shows a reach strip of split teal and violet cells with
staggered arcs from parent cells to child cells, colored gray, teal, or violet
by the acting owner; each arc carries a small green square for a strategy
factor or a dark square for the sentinel. Phase three shows a template strip
copied into a value strip, with the same colored arcs folding child cells back
into parent cells. Phase four shows a gently curved fan of orange arrows
scattering child-value cells into the regret strip. A thick upward arrow on the
left labeled next iteration closes the cycle.}
\label{fig:archdataflow}
\end{minipage}
\end{figure}

%% file: sections/tab_opsteps.tex
\begin{table}[!htb]
\centering
\caption{The realized operation schedule of one CFR$^{+}$ iteration
(\cref{fig:archdataflow}). $\mathrm{ia}$ abbreviates \texttt{index\_add},
$I(q)$ is the information set owning slot $q$, and the build-time sign
$s \in \{\pm 1\}$ orients each edge's regret toward its acting player. Only the
two per-level rows repeat with depth.}
\label{tab:opsteps}
{
\setlength{\tabcolsep}{3pt}
\begin{tabular}{@{}llr@{}}
\toprule
Realized operation & Data & Ops \\
\midrule
\multicolumn{3}{@{}l}{\emph{Regret matching (constant)}} \\
$\strat^{+} \gets \mathrm{clamp}(\regret, 0)$ & $[\numslots]$ & 1 \\
$\mathrm{tot} \gets 0$;\; $\mathrm{tot}.\mathrm{ia}(I(q), \strat^{+})$ & $[\numslots] {\to} [\numinfosets]$ & 2 \\
gather $\mathrm{tot}$ back to slots & $[\numinfosets] \to [\numslots]$ & 1 \\
$m \gets \mathrm{tot} > 0$;\; guard divisor via $\mathrm{where}$ & $[\numslots]$ & 2 \\
$\strat \gets \mathrm{where}(m, \strat^{+}/\mathrm{tot}, \mathrm{unif})$ & $[\numslots]$ & 2 \\
$\strat_{\mathrm{ext}}[{:}\numslots] \gets \strat$ \; (sentinel stays $1$) & $[\numslots{+}1]$ & 1 \\
\midrule
\multicolumn{3}{@{}l}{\emph{Forward reach, per level $\ell = 1, \dots, D$}} \\
$\reach[C_2(\ell)] \gets \reach[P_2(\ell)] \cdot \strat_{\mathrm{ext}}[S_2(\ell)]$ \; (both lanes) & $[2\numnodes]$ & $4 \cdot D$ \\
\midrule
\multicolumn{3}{@{}l}{\emph{Backward values, per level $\ell = D, \dots, 1$}} \\
$\cfv \gets \cfv_{\mathrm{tmpl}}$ \; (chance folded at build) & $[\numnodes]$ & 1 \\
$\cfv.\mathrm{ia}\big(P(\ell), \strat_{\mathrm{ext}}[S(\ell)] \cdot \cfv[C(\ell)]\big)$ & $[\numnodes]$ & $4 \cdot D$ \\
\midrule
\multicolumn{3}{@{}l}{\emph{Regret and average over decision edges (constant)}} \\
$\Delta \gets s \cdot (\cfv[C_{\mathrm{reg}}] - \cfv[P_{\mathrm{reg}}])$ & decision edges & 4 \\
$\regret.\mathrm{ia}\big(S_{\mathrm{reg}}, \reach[O_{\mathrm{reg}}] \cdot \Delta\big)$ & ${\to} [\numslots]$ & 3 \\
$\bar{s}.\mathrm{ia}\big(S_{\mathrm{reg}}, (w \cdot \reach[M_{\mathrm{reg}}] \cdot \strat_{\mathrm{ext}}[S_{\mathrm{reg}}])_{\mathrm{f64}}\big)$ & ${\to} [\numslots]$ & 6 \\
$\regret \gets \max(\regret, 0)$ \; (CFR$^{+}$) & $[\numslots]$ & 1 \\
$w \mathrel{{+}{=}} 1$ \; (inside the captured body) & scalar & 1 \\
\bottomrule
\end{tabular}}
\end{table}

%% file: sections/4b_theory.tex
\section{The Compiled Iteration, Formally}
\label{sec:theory}

This section states the per-iteration dataflow as a composition of three indexed
operators and establishes four properties of it. The folded chance representation
computes the same values as an explicit chance chain. The branchless dual-lane update
computes \cref{eq:lanes} exactly. The depth schedule is the shortest among schedules
that execute edges in dependency-respecting groups. Unrescaled linear averaging has a
representability horizon that fixes the accumulator precision. Proofs are in
\cref{app:proofs}.

\paragraph{Three indexed operators.}
\label{sec:operators}

Fix a compiled game with node set $\nodes$, $|\nodes| = \numnodes$, and edge set
indexed $1, \dots, \numedges$. Write $p(e), c(e)$ for the parent and child of edge
$e$ and $q(e)$ for its slot. Let $d(h)$ be the depth of node $h$ and
\begin{equation}
    B_\ell = \{\, e : d(p(e)) = \ell \,\},
    \quad
    \mathcal{D} = \{\, \ell : B_\ell \neq \emptyset \,\},
    \quad
    D = |\mathcal{D}|,
    \label{eq:blocks}
\end{equation}
so the blocks $\{B_\ell\}_{\ell \in \mathcal{D}}$ partition the edges by parent depth.
Let $\strat_{\mathrm{ext}} \in \R^{\numslots + 1}$ be the strategy vector extended by a
sentinel coordinate pinned at $\strat_{\mathrm{ext}}(\bot) = 1$, and for a lane
$i \in \{1,2\}$ define the compiled slot map
\begin{equation}
    q_i(e) \;=\;
    \begin{cases}
        q(e) & \text{edge } e \text{ is a decision edge of player } i,\\
        \bot & \text{otherwise (opponent or chance edge)},
    \end{cases}
    \label{eq:slotmap}
\end{equation}
which is a build-time constant. Its lane-free analogue $q_{\mathrm{val}}(e)$
($q(e)$ on decision edges, $\bot$ on chance edges) drives the backward pass.
Three operators then generate the iteration. For a
block $B_\ell$, a lane $i$, and vectors $x \in \R^{\numnodes}$, $y \in \R^{\numslots+1}$:
\begin{align}
    \fwd_{\ell,i}(x, y)(h)
        &=
        \begin{cases}
            x(p(e))\, y(q_i(e)) & h = c(e),\; e \in B_\ell,\\
            x(h) & \text{otherwise},
        \end{cases}
        \label{eq:opfwd} \\[2pt]
    \bwd_{\ell}(x, y)(h)
        &= x(h) + \sum_{e \in B_\ell \,:\, p(e) = h} y(q_{\mathrm{val}}(e))\, x(c(e)),
        \label{eq:opbwd}
\end{align}
\begin{equation}
    \acc_{\ell}(r, x, \rho)(q)
        = r(q)
        + \sum_{\substack{e \in B_\ell \,:\, q(e) = q \\ e \text{ a decision edge}}}\!\!
        s(e)\,\rho_{-i(e)}(p(e))\,\big(x(c(e)) - x(p(e))\big).
        \label{eq:opacc}
\end{equation}
Here $i(e)$ is the decision-edge owner, $s(e)=+1$ for player 1 and $-1$
for player 2, and $\rho_{-i(e)}$ selects the opponent's reach lane. The vector
$x$ holds chance-folded player-1 values, so no further chance factor is used.
Each operator uses gathers, elementwise arithmetic, and a scatter over $|B_\ell|$
entries, which is how \cref{alg:iteration} realizes them: \cref{eq:opfwd} as an
indexed assignment, \cref{eq:opbwd,eq:opacc} as \texttt{index\_add}. The forward pass is
$\fwd$ composed over $\ell \in \mathcal{D}$ in increasing depth, the backward pass is
$\bwd$ and $\acc$ composed in decreasing depth, and \cref{eq:opfwd} runs on the doubled
index range so both lanes advance in one call.

The forward pass assigns in \cref{eq:opfwd} where the backward pass accumulates.
Assignment is what frees it of buffer clearing. In a tree each non-root node is the
child of exactly one edge, so within a block every destination is written exactly
once and stale contents are unreachable. The backward operators accumulate instead, because a parent has several
children per block.

\paragraph{What the operators compute.}

\begin{proposition}[Folded chance is exact]
\label{prop:chance}
Initialize $\cfv \gets \cfv_{\mathrm{tmpl}}$ as in \cref{eq:template} and apply
$\bwd_\ell$ for $\ell \in \mathcal{D}$ in decreasing order, with every chance edge
assigned multiplier $1$ through $q_{\mathrm{val}}$. Then on termination, for every node
$h$,
\begin{equation}
    \cfv(h) \;=\; \sum_{z \succeq h,\; z \in \termnodes}
        \pi_c(h \!\to\! z)\, \pi_c(h)\, \reach^{\strat}_{1,2}(h \!\to\! z)\, u(z),
    \label{eq:chanceexact}
\end{equation}
that is, the chance-weighted continuation value at $h$. In particular at the root
$\cfv$ equals the expected payoff under $\strat$, and each terminal's chance product
$\pi_c(z)$ enters exactly once.
\end{proposition}

\Cref{prop:chance} says that folding saves work and leaves the computed values
unchanged. The chance chain is deleted from the dynamic pass and reappears only as
the build-time constant multiplying $u(z)$ in \cref{eq:template}.

\begin{proposition}[The dual-lane update is exact and branch-free]
\label{prop:lanes}
For every lane $i$, block $\ell$, and edge $e \in B_\ell$, the operator $\fwd_{\ell,i}$
of \cref{eq:opfwd} evaluated with the compiled slot map \cref{eq:slotmap} and
$y = \strat_{\mathrm{ext}}$ satisfies the case definition of \cref{eq:lanes} at
$c(e)$. The slot map is resolved at build time, so the per-iteration cost of the case
distinction is zero and no per-edge predicate is evaluated during execution.
\end{proposition}

\begin{proposition}[The depth schedule is shortest]
\label{prop:depth}
Call a schedule \emph{dependency-respecting} if it partitions the edges into ordered
groups such that no group contains two edges $e, e'$ with $c(e) = p(e')$, and $e$
precedes $e'$ whenever $c(e) = p(e')$. Then:
\begin{enumerate}[label=(\roman*),topsep=1pt,itemsep=0pt]
    \item the partition \cref{eq:blocks} is dependency-respecting and has $D$ groups;
    \item every dependency-respecting schedule has at least $D$ groups;
    \item the number of operator invocations per iteration is $c_1 + c_2 D$ for
    constants $c_1, c_2$ independent of $\numnodes$ and $\numedges$.
\end{enumerate}
\end{proposition}

Part~(iii) is the reason iteration time is nearly flat in tree size across four
decades (\cref{tab:matrix}). Growth enters the batch \emph{widths}, which the device
absorbs in parallel, while the number of dispatches stays fixed. Measured on the
suite, $c_1 \approx 31$ (the itemized dispatches of \cref{tab:opsteps} plus framework
bookkeeping) and $c_2 = 8$. The realized schedule applies $\acc$ once, flat over all
decision edges (\cref{sec:dataflow}), so only $\fwd$ and $\bwd$ contribute to $c_2$,
at four operations per level each. \Cref{tab:suite} confirms that $D$ matches the
tree depth on every game. The reference implementation's per-(stage, depth) partition
needs up to $100$ groups where \cref{eq:blocks} needs $8$.

Kim (2026) \cite{kim2026parallelizing} organizes its sequence-form iteration by
level as well, so its critical path is likewise proportional to depth. The
difference the timing measures is per-step work. There, one level applies a masked
sparse product against a precomputed level matrix, whose cost is set by the sparsity
pattern and the sparse kernel. Here, one level applies
\cref{eq:opfwd,eq:opbwd,eq:opacc} on dense contiguous ranges with build-time
indices. \cref{app:extra} localizes the resulting gap
to the sparse product and the index assignment inside its per-iteration
\texttt{normalize}.

\paragraph{The averaging horizon.}

Linear averaging \cite{tammelin2014solving} accumulates
$\bar{s} \mathrel{+}= t\,\reach_{\mathrm{own}}\,\strat$ on iteration $t$ and never
rescales. The accumulator therefore grows quadratically while each contribution stays
linear, and in finite precision the two rates collide.

\begin{proposition}[Representability horizon of unrescaled linear averaging]
\label{prop:horizon}
Let $S_T = \sum_{t=1}^{T} t\,x_t$ be the ideal real sum, with
$x_t\in[0,1]$, $S_T = \Theta(T^2)$, and $x_T\ge c>0$.
For a binary format with $p$-bit significand, while $S_T$ is in the normal range,
\begin{equation}
    \frac{T x_T}{\ulp(S_T)} = \Theta\!\left(\frac{2^p}{T}\right).
    \label{eq:horizon}
\end{equation}
Here $\ulp$ is the spacing in the binade containing the ideal sum (the
actual rounded accumulator may sit in a neighboring binade). Without $x_T\ge c$, only the $O(2^p/T)$ bound
follows. Thus $2^p$ is an order-of-magnitude resolution scale and should be read as such;
the exact stalling point depends on the rounded trajectory. A nonnegative increment is locally absorbed under
round-to-nearest if it is strictly below half the upward spacing
at the actual old accumulator; ties depend on the rounding rule.
\end{proposition}

GPU-CFR holds $\bar{s}$ in float64 for all update rules, while the other buffers
and iteration arithmetic use the configured dtype. The scales
$2^{24}\approx1.7\times10^7$ and $2^{53}\approx9\times10^{15}$ show the precision
difference. The cost is one wider slot vector with no change to the operator
sequence. This accumulator issue is separate from the graph-replay counter of
\cref{sec:cudagraph}. That float32 counter represents consecutive integers exactly
through $2^{24}$, and the host refills its starting index before each replay batch.

%% file: sections/5_setup.tex
\section{Experimental Setup}
\label{sec:setup}

\paragraph{Three evaluation suites.}
Each suite answers a different experimental question. The \emph{eight-game main
timing/regression suite} in \cref{tab:suite} contains six public games via OpenSpiel
\cite{lanctot2019openspiel} and two native HUNL-style subgames
\cite{ganzfried2015endgame,brown2017safe} (\cref{app:hunl}). The suite spans 54 to
275{,}983 Infosets and supports the headline timing, operation-count, and
reference-parity results. The \emph{twelve-game update-rule suite} adds four more
OpenSpiel games for the four-rule, two-order CPU study of \cref{tab:convergence}. The
\emph{extended poker suite} contains the released Libratus river endgames \cite{brown2018superhuman,steinberger2019pokerrl} and larger
exported poker instances of the head-to-head in \cref{app:libratus}.

\input{figures/TABLE_T1_suite.tex}

\paragraph{Backends.}
\emph{GPU-CFR} runs the compiled float32 solver on one NVIDIA A100 80GB in eager or
CUDA Graph mode. The two modes differ only in dispatch, and a vanilla CFR arm isolates
the engine from regret matching$^{+}$. \emph{Ours CPU} runs the same compiled tensor dataflow on eight threads
pinned to dedicated physical cores, with parallelism supplied by PyTorch's
intra-operator pool \cite{paszke2019pytorch}. \emph{LiteEFG} 0.1.5 \cite{liu2024liteefg} uses a single-threaded C++ backend. The
main matrix uses its simultaneous vanilla CFR mode, and \cref{app:extra} reports its
alternating CFR$^{+}$ preset. \emph{Kim (2026)}~\cite{kim2026parallelizing},
released as \texttt{noregret}, executes sequence-form
\cite{koller1997representations} CFR through CuPy sparse products on the same A100. The main matrix matches its update order
to GPU-CFR's simultaneous updates and reports its GPU backend. \emph{OpenSpiel} \cite{lanctot2019openspiel} contributes its Python CFR
implementation and exact best response. Every backend reports exploitability as NashConv$/2$. \Cref{tab:semantics} lists
the rule, update order, averaging, precision, and execution mode per arm.

\paragraph{Shared game representations.}
Public games are instantiated through OpenSpiel and converted to each backend's input
format. For the HUNL-style subgames and Libratus endgames, we export the compiled
intermediate representation to LiteEFG's generic extensive-form text format and load
the same checksummed tree through the Kim adapter. Structural counts and small-game behavior are checked across conversions
(\cref{app:protocol}), and \cref{app:libratus} reports the poker value-agreement
certificate.

\paragraph{Protocol.}
Timing rows run 1{,}000 CFR$^{+}$ iterations. OpenSpiel runs 200 CFR iterations and
is reported per iteration. GPU steady-state measurements use CUDA events \cite{guide2020cuda} after 50
warmup iterations. Graph capture and tree construction are measured separately, while
training-call rows include all work after solver construction. Each cell runs in an
independent low-load process, and tables report
medians over repeated runs. \Cref{app:protocol,app:correctness} gives the software
stack, repetition counts, load controls, and run-level traceability.

%% file: figures/TABLE_T1_suite.tex
\begin{table*}[!htb]
\centering
\caption{Benchmark suite and PyTorch/Aten operations issued per CFR$^+$ iteration
(load-immune count), reference implementation vs.\ compiled dataflow. Infosets counts
information sets, one per player at every non-chance history including terminals;
States counts the legal non-chance histories once both players' private information
is revealed. Blocks counts the sequential execution blocks per pass; the last column
is the maximum absolute regret difference after 22 float32 iterations.}
\label{tab:suite}
{
\setlength{\tabcolsep}{4pt}
\begin{tabular}{lrrrrrrrr}
\toprule
 & & & \multicolumn{2}{c}{Exec.\ blocks} & \multicolumn{2}{c}{Aten ops/iter} & & Regret \\
\cmidrule(lr){4-5} \cmidrule(lr){6-7}
Game & Infosets & States & ref & GPU-CFR & ref & GPU-CFR & Reduction & max diff \\
\midrule
Kuhn & 54 & 54 & 4 & 4 & 110 & 64 & 1.7$\times$ & 0.0 \\
Dark Hex 2x2 & 437 & 441 & 7 & 7 & 161 & 88 & 1.8$\times$ & 0.0 \\
HUNL river & 3,000 & 137,415 & 27 & 4 & 501 & 64 & 7.8$\times$ & 0.0 \\
Leduc & 4,620 & 9,300 & 11 & 11 & 229 & 120 & 1.9$\times$ & 0.0 \\
Goofspiel-5 & 13,293 & 26,931 & 9 & 8 & 195 & 96 & 2.0$\times$ & 0.0 \\
HUNL turn & 83,040 & 433,610 & 100 & 8 & 1,742 & 96 & 18.1$\times$ & 0.0 \\
Liar's Dice & 98,292 & 294,876 & 83 & 15 & 1,453 & 152 & 9.6$\times$ & 0.0 \\
Battleship & 275,983 & 209,941 & 25 & 10 & 467 & 112 & 4.2$\times$ & 0.0 \\
\bottomrule
\end{tabular}
}
\end{table*}

%% file: sections/6_results.tex
\section{Results}
\label{sec:results}
\input{figures/TABLE_T2_matrix.tex}

\begin{figure*}[!htb]
    \centering
    \begin{minipage}[t]{0.48\textwidth}
        \centering
        \includegraphics[width=\linewidth]{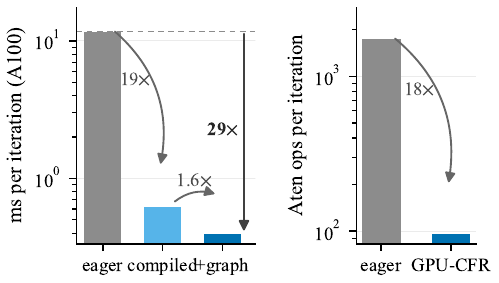}
        \Description{Two bar panels. The left panel shows per-iteration time dropping
        from 11.667 milliseconds for the eager reference to 0.619 for the compiled
        dataflow and 0.397 for the CUDA graph, with arrows marking the 19 times, 1.6
        times, and overall 29 times reductions. The right panel shows Aten operations
        per iteration dropping from 1742 to 96, an 18 times reduction.}
        \captionof{figure}{Per-iteration time (left) and PyTorch/Aten operations (right)
        on the HUNL turn subgame.}
        \label{fig:ladder}
    \end{minipage}\hfill
    \begin{minipage}[t]{0.465\textwidth}
        \centering
        \includegraphics[width=\linewidth]{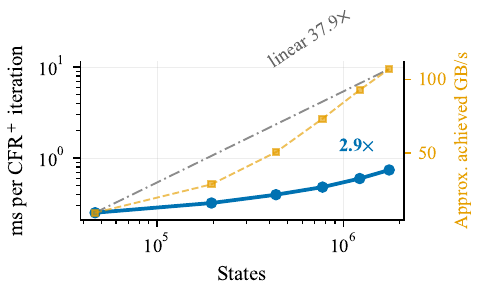}
        \Description{Log-log plot of milliseconds per CFR-plus iteration against the
        number of States. The measured solid curve rises only 2.9 times while a dash-dot
        reference line of slope one rises 37.9 times; a dashed second axis shows
        achieved bandwidth reaching about 107 gigabytes per second at the largest
        size.}
        \captionof{figure}{HUNL turn scaling, $K=4$ to $24$ hands per player (46k to
        1.76M States). Solid: steady-state graph time; dashed: achieved bandwidth.}
        \label{fig:scaling}
    \end{minipage}
\end{figure*}

\paragraph{Cross-framework timing.}
\label{sec:matrix}
\Cref{tab:matrix} and \cref{fig:hero} show that the preferred execution model
changes with game size. On the smallest trees the single-threaded C++ library wins.
The GPU pays a fixed launch-and-replay floor of about $0.113$\,ms per iteration, and
a 54-Infoset game cannot fill it. The compiled dataflow overtakes the CPU frameworks
as the game grows, and the margin widens from the HUNL turn subgame upward. There
graph replay is \liteGraphTurn$\times$ faster per iteration than LiteEFG. Across nearly four decades of game
size its per-iteration time changes by only about $5\times$. The number of dispatches
is set by tree depth, and additional states only widen each batched operation.

Kim (2026)~\cite{kim2026parallelizing} runs sequence-form CFR$^{+}$ on the same
A100, so the gap to it measures execution representation alone. The graph-replayed path is \kimGraphRange$\times$
faster across the eight-game suite, with a median of \kimGraphMedian$\times$, and the
eager compiled path is \kimEagerRange$\times$ faster. The flat schedule therefore carries most of the margin and graph replay adds the
rest. \Cref{tab:ttt} confirms the same ordering at matched exploitability
thresholds. The representation carries most of the
speedup. On eight CPU threads and no accelerator, the compiled
solver is \kimCpuRange$\times$ faster than the A100 baseline on all eight games, with
a median of \kimCpuMedian$\times$, and graph replay is a further
\graphCpuLargestRange$\times$ ahead of that CPU arm on the four largest games. Eight threads is a fair CPU ceiling for this dataflow. The sweep to 28 pinned cores
in \cref{fig:cpuscaling} saturates early, because a depth block is too small for the
intra-operator thread pool to split further. Against LiteEFG \cite{liu2024liteefg},
the general compiled CPU library, the same pattern holds. The eight-thread arm is already \cpuLiteLargestRange$\times$ faster on the four
largest games, and graph replay extends the margin to \liteGraphLargestRange$\times$.

\paragraph{Where the speed comes from.}
\label{sec:ladder}
\input{figures/TABLE_T6_overhead.tex}
The HUNL turn ablation in \cref{fig:ladder} times the two stages of the design in one
process under one software stack. Compiling the reference iteration into the
depth-level dataflow cuts Aten operations by $18\times$ and time by $19\times$. Graph
replay then removes the remaining host submissions for a further $1.6\times$, giving a
$29\times$ ladder end to end. The Nsight traces in \cref{app:profile} show the
mechanism. The eager path issues 87 kernel launches per iteration and keeps the GPU
busy a third of the time. Graph replay issues one graph launch per iteration and keeps
the GPU busy 94\% of the time. Recording 100 iterations per graph or handing the same body to
\texttt{torch.compile} \cite{paszke2019pytorch} changes nothing further. The
remaining time is the execution of the same 87 small gather and scatter kernels. The prior GPU baseline \cite{kim2026parallelizing} cannot take the second step at
all. Its iteration calls cuSPARSE through CuPy, which refuses to run under stream
capture. Its per-level row slicing and index assignment also move data between host
and device on every iteration. Capture therefore fails on the first sparse product
even at Kuhn size (\cref{app:profile}). Capturability is a property of the representation.

\Cref{fig:graphablation} shows graph-vs-eager steady-state times for every game. The
multiplier shrinks as trees grow, from 3.6$\times$ on Kuhn (54 Infosets,
fully launch-bound) to 1.6$\times$ on the HUNL turn subgame (83k Infosets), because
kernel execution time grows with data volume while launch overhead is fixed.

\begin{figure}[!htb]
  \begin{minipage}[b]{0.48\linewidth}\centering
    \includegraphics[width=\linewidth]{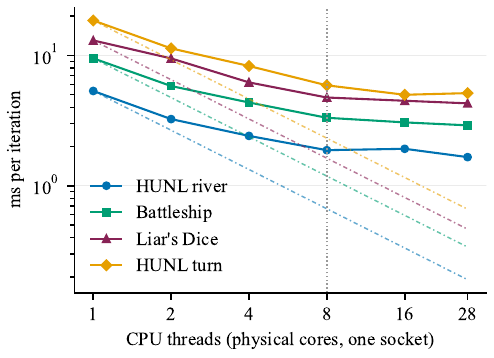}
    \Description{Log-log plot of milliseconds per iteration against CPU thread count
    for four games; each measured curve flattens after eight threads while a dash-dot
    ideal-scaling line keeps falling.}
    \caption{CPU thread scaling of the compiled solver (eager, CFR$^{+}$); dash-dot
    lines mark ideal scaling from the one-thread point, the dotted line the 8-thread
    arm of \cref{tab:matrix}.}
    \label{fig:cpuscaling}
  \end{minipage}\hfill
  \begin{minipage}[b]{0.48\linewidth}\centering
    \includegraphics[width=\linewidth]{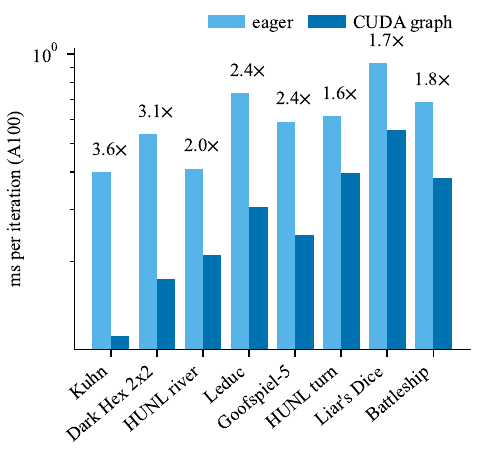}
    \Description{Grouped bar chart with one pair of bars per game on a log scale:
    a light bar for eager execution and a dark bar for CUDA-graph replay. The ratio
    printed above each pair shrinks from 3.6 times on Kuhn to 1.6 times on the HUNL
    turn subgame.}
    \caption{CUDA-graph replay vs.\ eager execution: steady-state ms per CFR$^{+}$
    iteration for every game (log scale). Labels give the eager/graph ratio.}
    \label{fig:graphablation}
  \end{minipage}
\end{figure}

The scaling sweep in \cref{fig:scaling} explains why the compiled schedule keeps
paying off on larger trees. Growing the turn subgame from 4 to 24 hands per player
multiplies the node count by $37.9\times$ while iteration time grows only
$2.95\times$. The extra hands add data volume and leave tree depth nearly unchanged,
so the operation sequence stays fixed and each operation simply gets wider. Achieved
bandwidth from a bytes-moved model reaches $5.5\%$ of the A100 peak at the largest
size, and hardware counters put measured traffic at $17\%$ (\cref{app:profile}). The
sweep is governed by launch latency and depth, not by bandwidth.

Precision and memory follow the same logic. Float64 costs $1.08\times$ on the graph
path for the turn subgame and nothing on Leduc, with eager times unchanged. Iteration
time is set by dispatch, and the arithmetic is a small part of it.
\Cref{tab:memory} (\cref{app:extra}) lists peak GPU memory per game against the
analytic model of \cref{sec:compiled}. Across the scaling sweep, allocated memory per tree node stays within
a $3\%$ band while the tree grows $37.9\times$, so memory scales linearly over the
tested range.

\begin{figure}[!htb]
    \centering
    \includegraphics[width=0.95\linewidth]{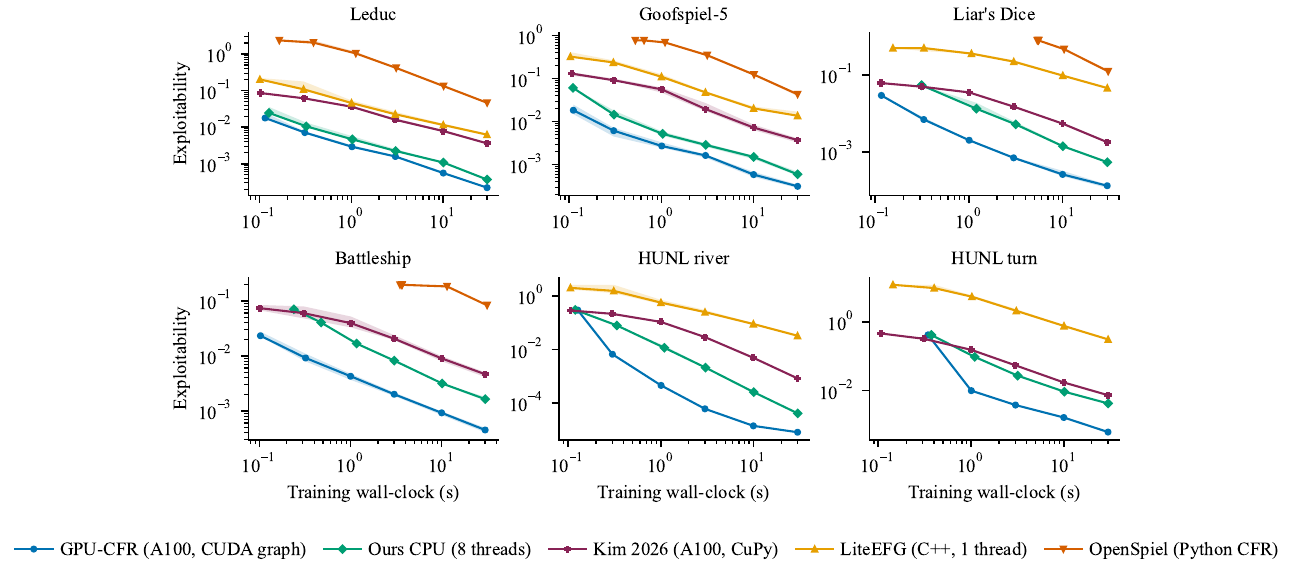}
    \Description{Six log-log panels, one per game, plotting exploitability against
    training wall-clock seconds for up to five backends. In every panel the GPU-CFR
    CUDA-graph curve sits lowest, reaching lower exploitability earlier; the CPU arm
    of the same compiled solver follows, then the Kim baseline, LiteEFG, and
    OpenSpiel.}
    \caption{Exploitability versus training wall-clock (both axes logarithmic) on four public
    games and the native HUNL subgames.}
    \label{fig:curves}
\end{figure}

\paragraph{Quality against wall-clock.}
\label{sec:quality}
\Cref{fig:curves} plots exploitability against training wall-clock. Each backend
runs its library's native preset (\cref{tab:matrix,tab:semantics}; the update-rule
study is in \cref{sec:convergence}), so the LiteEFG and OpenSpiel traces also differ
in algorithm. The matched-semantics comparison is GPU-CFR against Kim. The HUNL
panels use the shared exports and compare GPU-CFR on CUDA, our compiled CFR$^{+}$
implementation on CPU, Kim's NoRegret CFR$^{+}$ on CUDA, and LiteEFG's native vanilla
CFR. Each arm runs 10 independent processes at six log-spaced wall-clock targets from
0.1 to 30 seconds. Points are medians and bands are the full range across processes,
which \cref{app:correctness} traces to reduction order. In every panel the
graph-replayed curve sits lowest, and the lower per-iteration cost translates
directly into time. \Cref{tab:ttt} shows GPU-CFR reaching thresholds of 10, 3 and 1.5
times Kim's 30\,s exploitability \tttRatioMin--\tttRatioMax$\times$ sooner, with
disjoint bootstrap intervals. Kuhn and Dark Hex (converged within milliseconds) and
LiteEFG's Battleship trace (accessor stuck near $10^{-15}$) are omitted.

\Cref{tab:ttt} reports, for the matched GPU-CFR versus Kim pair of \cref{fig:curves},
the distribution over independent processes of the estimated wall-clock at which each
process first crosses a fixed exploitability. Thresholds are 10, 3 and 1.5 times Kim's
30\,s median exploitability; times are read off each trace by log-log interpolation between the
bracketing checkpoints, and the intervals are 10{,}000-resample bootstraps over ten
processes. The CPU column reports the median crossing time over ten independent
processes of the same compiled solver on eight threads.
\input{figures/TABLE_T13_ttt.tex}

\paragraph{Update rule and update order.}
\label{sec:rules}
The rules differ only in how the two accumulators evolve. CFR$^{+}$
\cite{tammelin2014solving} clamps accumulated regrets at zero after each update and
weights the strategy average linearly by iteration index. DCFR
\cite{brown2019solving} instead discounts the accumulators: positive and negative
regrets decay at separate iteration-indexed rates, and the strategy average decays
polynomially. Predictive CFR$^{+}$ (PCFR$^{+}$) \cite{farina2021faster} keeps the
CFR$^{+}$ regret update but matches the strategy against the accumulated regret plus
the current instantaneous regret as a predictor. Under \emph{simultaneous} updates
both players' regrets are computed from one strategy profile; under
\emph{alternating} updates the first player's update is visible to the second
player's update in the same sweep \cite{burch2019revisiting}.

The compiled engine supports these four rules under both update orders, allowing one
execution path to expose their convergence behavior. The full grid has 96 cells. At
equal reported iteration counts, alternating updates are tighter in 44 of the 48
game/rule comparisons; the four exceptions are on Dark Hex, where both orders are
within $1.3\times10^{-6}$ of zero. After charging alternating for its approximately
two player updates, it remains better in 41 of 48 cells. The timing matrix therefore fixes
one common update order, while this study measures the algorithmic effect of
changing it.

Update order also changes which rule leads. PCFR$^{+}$ is tightest on eleven of the
twelve games under simultaneous updates, with DCFR leading only on Dark Hex. Under
alternating updates, PCFR$^{+}$ leads six games and DCFR the other six. The interaction shows that the compiler
supports a family of CFR updates without making the systems result depend on one
preferred rule. \Cref{sec:convergence} checks the ranking against LiteEFG and OpenSpiel.

\input{figures/TABLE_T9_convergence.tex}

\paragraph{Repeated solves of one tree.}
\label{sec:hunl}
\Cref{tab:lifecycle,tab:overhead} break the HUNL turn tree into setup and per-solve
phases. Building and compiling the specification dominates, graph capture is an order
of magnitude cheaper, and each later 1{,}000-iteration solve is subsecond. Setup is
repaid within the first solve. The \turnFirstSolveSec\,s from specification to
solution is below the \turnKimThousandSec\,s of Kim's A100 implementation and the
\turnLiteThousandSec\,s of LiteEFG for the same 1{,}000 iterations
(\cref{tab:matrix}).

The export mechanism extends the comparison to Libratus endgame~4
\cite{brown2018superhuman,steinberger2019pokerrl} at full range. This game has
96{,}159 Infosets on a tree of 36 million states that is almost entirely terminal
(\cref{tab:sg4threeway}). At matched iteration counts GPU-CFR is $870\times$ faster
than LiteEFG, which inverts the small-scale C++ advantage. Kim's sequence-form path
is slightly faster here, because a terminal-heavy tree leaves little depth-dependent
work for the compiled schedule to remove. The three backends agree on the value to $2.5$ chips
(\cref{app:libratus}).
\input{figures/TABLE_T8_sg4threeway.tex}

%% file: figures/TABLE_T2_matrix.tex
\begin{table*}[!htb]
\centering
\caption{Per-iteration wall-clock in milliseconds, lower is better. Entries are
medians over independent processes (repetition counts and spread in
\cref{tab:spread}); bold marks the fastest system on each game, GPU-CFR entering
as its graph path. Baselines: Kim (2026)~\cite{kim2026parallelizing}, LiteEFG~\cite{liu2024liteefg}, OpenSpiel~\cite{lanctot2019openspiel}. ``n/a$^\ddag$'' marks a backend limitation
(\cref{app:hunl}).}
\label{tab:matrix}
{
\setlength{\tabcolsep}{3.5pt}
\begin{tabular}{l rrrrrrrr}
\toprule
Arm & Kuhn & \shortstack{Dark Hex\\2x2} & \shortstack{HUNL\\river} & Leduc & \shortstack{Goof-\\spiel-5} & \shortstack{HUNL\\turn} & \shortstack{Liar's\\Dice} & \shortstack{Battle-\\ship} \\
\midrule
Infosets & 54 & 437 & 3,000 & 4,620 & 13,293 & 83,040 & 98,292 & 275,983 \\
\midrule
\textbf{GPU-CFR} (A100) &  &  &  &  &  &  &  &  \\
\quad graph & 0.113 & 0.174 & \textbf{0.210} & \textbf{0.304} & \textbf{0.245} & \textbf{0.397} & \textbf{0.556} & \textbf{0.380} \\
\quad eager & 0.401 & 0.536 & 0.411 & 0.738 & 0.591 & 0.617 & 0.933 & 0.688 \\
\quad vanilla CFR & 0.136 & 0.200 & 0.311 & 0.328 & 0.266 & 0.701 & 0.591 & 0.403 \\
Ours CPU (8 threads) & 0.177 & 0.243 & 1.697 & 0.612 & 0.765 & 5.360 & 4.378 & 3.138 \\
Kim (2026) & 9.057 & 11.691 & 9.642 & 12.860 & 13.180 & 11.893 & 16.545 & 14.393 \\
LiteEFG (1 thread) & \textbf{0.008} & \textbf{0.065} & 16.857 & 1.233 & 3.375 & 102.160 & 102.880 & 80.952 \\
OpenSpiel (Python) & 0.973 & 8.369 & n/a$^\ddag$ & 139.236 & 230.114 & n/a$^\ddag$ & 1,542.649 & 2,528.514 \\
\bottomrule
\end{tabular}
}
\end{table*}

%% file: figures/TABLE_T6_overhead.tex
\begin{table*}[!htb]
\centering
\caption{One-time costs (CUDA, graph on): build (specification, compilation, solver construction; split in \cref{tab:lifecycle}), three eager warmup iterations, graph capture with first replay, and steady-state iteration time; medians over $n=20$ processes.}
\label{tab:overhead}
\begin{tabular}{lrrrrr}
\toprule
Game & Infosets & Build (s) & Warmup (ms) & Capture (ms) & Steady (ms/iter) \\
\midrule
HUNL river & 3,000 & 1.192 & 4.6 & 88.9 & 0.210 \\
HUNL turn & 83,040 & 4.233 & 5.2 & 301.0 & 0.398 \\
Liar's Dice & 98,292 & 3.837 & 6.8 & 43.3 & 0.555 \\
Battleship & 275,983 & 2.683 & 6.3 & 36.8 & 0.378 \\
\bottomrule
\end{tabular}
\end{table*}

%% file: figures/TABLE_T13_ttt.tex
\begin{table*}[!htb]
\centering
\caption{Time-to-threshold behind \cref{fig:curves}. For each game the
thresholds are 10$\times$, 3$\times$ and 1.5$\times$ the median exploitability Kim (2026)'s
A100 arm reaches at its 30\,s checkpoint; rows in which fewer than eight of ten
processes of either matched arm reach the threshold within the budget are omitted. Cells give the median over
independent processes of the estimated crossing time, the second at which a process's trace,
interpolated log--log between the bracketing checkpoints, first passes the threshold,
with a 10,000-resample bootstrap 95\% interval in brackets; the ratio
column is Kim's median over GPU-CFR's. GPU-CFR and Kim run
matched CFR$^{+}$ semantics; the CPU column is the same compiled solver on eight threads,
reported as a median over ten processes without an interval. \emph{Reached} counts
processes that attain the threshold within 30\,s (GPU-CFR, Kim).}
\label{tab:ttt}
{
\setlength{\tabcolsep}{4pt}
\begin{tabular}{lrrrrr}
\toprule
Threshold $\epsilon$ & GPU-CFR $t_\epsilon$ (s) & Kim (2026) $t_\epsilon$ (s) & Kim/GPU-CFR & CPU $t_\epsilon$ (s) & Reached \\
\midrule
\multicolumn{6}{l}{\emph{Leduc}} \\
$3.60\times10^{-2}$ & 0.115 [0.112, 0.115] & 1.015 [1.003, 1.074] & 8.846 & 0.126 & 10/10, 10/10 \\
$1.08\times10^{-2}$ & 0.198 [0.195, 0.207] & 5.828 [5.553, 6.188] & 29.454 & 0.314 & 10/10, 10/10 \\
$5.41\times10^{-3}$ & 0.450 [0.445, 0.467] & 16.971 [16.450, 17.302] & 37.732 & 0.828 & 10/10, 10/10 \\
\multicolumn{6}{l}{\emph{Goofspiel}} \\
$3.73\times10^{-2}$ & 0.111 [0.108, 0.113] & 1.585 [1.373, 1.760] & 14.321 & 0.156 & 10/10, 10/10 \\
$1.12\times10^{-2}$ & 0.182 [0.156, 0.190] & 6.139 [5.636, 6.626] & 33.793 & 0.415 & 10/10, 10/10 \\
$5.60\times10^{-3}$ & 0.358 [0.287, 0.391] & 15.721 [13.512, 17.551] & 43.970 & 0.945 & 10/10, 10/10 \\
\multicolumn{6}{l}{\emph{Liar's Dice}} \\
$1.81\times10^{-2}$ & 0.166 [0.165, 0.168] & 2.373 [2.347, 2.386] & 14.334 & 0.897 & 10/10, 10/10 \\
$5.42\times10^{-3}$ & 0.415 [0.414, 0.417] & 10.084 [9.949, 10.207] & 24.312 & 3.042 & 10/10, 10/10 \\
$2.71\times10^{-3}$ & 0.775 [0.773, 0.777] & 20.062 [19.909, 20.394] & 25.884 & 5.666 & 10/10, 10/10 \\
\multicolumn{6}{l}{\emph{Battleship}} \\
$4.65\times10^{-2}$ & 0.103 [0.102, 0.106] & 0.658 [0.515, 0.754] & 6.411 & 0.400 & 10/10, 10/10 \\
$1.39\times10^{-2}$ & 0.193 [0.182, 0.210] & 5.163 [4.988, 5.773] & 26.743 & 1.499 & 10/10, 10/10 \\
$6.97\times10^{-3}$ & 0.469 [0.434, 0.523] & 14.942 [13.839, 15.959] & 31.885 & 3.768 & 10/10, 10/10 \\
\multicolumn{6}{l}{\emph{HUNL river}} \\
$8.16\times10^{-3}$ & 0.287 [0.287, 0.290] & 7.023 [7.009, 7.038] & 24.472 & 1.344 & 10/10, 10/10 \\
$2.45\times10^{-3}$ & 0.470 [0.469, 0.473] & 15.216 [15.186, 15.232] & 32.398 & 2.756 & 10/10, 10/10 \\
$1.22\times10^{-3}$ & 0.641 [0.641, 0.646] & 23.357 [23.312, 23.362] & 36.411 & 4.100 & 10/10, 10/10 \\
\multicolumn{6}{l}{\emph{HUNL turn}} \\
$7.39\times10^{-2}$ & 0.566 [0.565, 0.570] & 2.174 [2.158, 2.187] & 3.840 & 1.362 & 10/10, 10/10 \\
$2.22\times10^{-2}$ & 0.801 [0.800, 0.806] & 7.710 [7.637, 7.813] & 9.624 & 3.980 & 10/10, 10/10 \\
$1.11\times10^{-2}$ & 0.978 [0.976, 0.983] & 17.821 [17.639, 18.008] & 18.213 & 8.432 & 10/10, 10/10 \\
\bottomrule
\end{tabular}}
\end{table*}

%% file: figures/TABLE_T9_convergence.tex
\begin{table*}[!htb]
\centering
\caption{Exploitability after $1{,}000$ iterations for four update rules
under both update orders, all on the same compiled engine. Bold marks the tightest
rule per game and order; cells are bit-reproducible CPU measurements, and a
repetition disagreement withholds the cell.}
\label{tab:convergence}
{\small
\setlength{\tabcolsep}{3pt}
\begin{tabular}{lrrrrrrrr}
\toprule
 & \multicolumn{4}{c}{Alternating} & \multicolumn{4}{c}{Simultaneous} \\
\cmidrule(lr){2-5} \cmidrule(lr){6-9}
Game & CFR & CFR$^{+}$ & DCFR & PCFR$^{+}$ & CFR & CFR$^{+}$ & DCFR & PCFR$^{+}$ \\
\midrule
Kuhn & $6.18\!\times\!10^{-4}$ & $8.73\!\times\!10^{-5}$ & $1.99\!\times\!10^{-4}$ & $\mathbf{2.20\!\times\!10^{-8}}$ & $1.32\!\times\!10^{-2}$ & $2.83\!\times\!10^{-3}$ & $5.05\!\times\!10^{-3}$ & $\mathbf{4.56\!\times\!10^{-5}}$ \\
Dark Hex 2x2 & $1.28\!\times\!10^{-6}$ & $1.28\!\times\!10^{-6}$ & $\mathbf{4.47\!\times\!10^{-9}}$ & $4.47\!\times\!10^{-9}$ & $1.26\!\times\!10^{-6}$ & $1.26\!\times\!10^{-6}$ & $\mathbf{1.62\!\times\!10^{-9}}$ & $1.62\!\times\!10^{-9}$ \\
Liar's Dice 1$\times$4 & $1.17\!\times\!10^{-4}$ & $4.52\!\times\!10^{-5}$ & $3.43\!\times\!10^{-5}$ & $\mathbf{1.46\!\times\!10^{-6}}$ & $1.86\!\times\!10^{-2}$ & $2.48\!\times\!10^{-3}$ & $2.54\!\times\!10^{-3}$ & $\mathbf{1.79\!\times\!10^{-5}}$ \\
Leduc & $1.04\!\times\!10^{-1}$ & $2.73\!\times\!10^{-4}$ & $\mathbf{1.57\!\times\!10^{-4}}$ & $7.13\!\times\!10^{-4}$ & $1.46\!\times\!10^{-1}$ & $6.65\!\times\!10^{-3}$ & $8.59\!\times\!10^{-3}$ & $\mathbf{1.46\!\times\!10^{-3}}$ \\
Goofspiel-5 & $1.44\!\times\!10^{-2}$ & $3.86\!\times\!10^{-4}$ & $2.73\!\times\!10^{-4}$ & $\mathbf{6.21\!\times\!10^{-5}}$ & $9.55\!\times\!10^{-2}$ & $5.41\!\times\!10^{-3}$ & $7.58\!\times\!10^{-3}$ & $\mathbf{1.16\!\times\!10^{-3}}$ \\
Liar's Dice 1$\times$5 & $5.57\!\times\!10^{-3}$ & $4.60\!\times\!10^{-5}$ & $\mathbf{1.73\!\times\!10^{-5}}$ & $8.27\!\times\!10^{-5}$ & $3.42\!\times\!10^{-2}$ & $2.68\!\times\!10^{-3}$ & $2.79\!\times\!10^{-3}$ & $\mathbf{3.02\!\times\!10^{-4}}$ \\
HUNL river & $5.35\!\times\!10^{-3}$ & $5.84\!\times\!10^{-4}$ & $\mathbf{8.54\!\times\!10^{-5}}$ & $5.51\!\times\!10^{-4}$ & $9.45\!\times\!10^{-2}$ & $5.47\!\times\!10^{-3}$ & $4.24\!\times\!10^{-3}$ & $\mathbf{9.11\!\times\!10^{-4}}$ \\
Battleship & $4.03\!\times\!10^{-3}$ & $8.25\!\times\!10^{-4}$ & $8.18\!\times\!10^{-4}$ & $\mathbf{2.19\!\times\!10^{-7}}$ & $6.66\!\times\!10^{-2}$ & $8.34\!\times\!10^{-3}$ & $9.13\!\times\!10^{-3}$ & $\mathbf{1.73\!\times\!10^{-3}}$ \\
Liar's Dice 1$\times$6 & $3.37\!\times\!10^{-3}$ & $1.27\!\times\!10^{-4}$ & $\mathbf{5.94\!\times\!10^{-5}}$ & $3.95\!\times\!10^{-4}$ & $3.54\!\times\!10^{-2}$ & $3.40\!\times\!10^{-3}$ & $3.58\!\times\!10^{-3}$ & $\mathbf{9.09\!\times\!10^{-4}}$ \\
HUNL turn & $5.25\!\times\!10^{-2}$ & $9.33\!\times\!10^{-4}$ & $\mathbf{2.70\!\times\!10^{-4}}$ & $9.94\!\times\!10^{-4}$ & $2.09\!\times\!10^{-1}$ & $1.62\!\times\!10^{-2}$ & $1.76\!\times\!10^{-2}$ & $\mathbf{1.52\!\times\!10^{-3}}$ \\
Liar's Dice 2$\times$3 & $2.20\!\times\!10^{-3}$ & $3.54\!\times\!10^{-5}$ & $1.93\!\times\!10^{-5}$ & $\mathbf{1.16\!\times\!10^{-5}}$ & $4.95\!\times\!10^{-2}$ & $4.05\!\times\!10^{-3}$ & $5.01\!\times\!10^{-3}$ & $\mathbf{4.57\!\times\!10^{-4}}$ \\
Goofspiel-6 & $2.05\!\times\!10^{-2}$ & $4.62\!\times\!10^{-4}$ & $4.37\!\times\!10^{-4}$ & $\mathbf{3.66\!\times\!10^{-4}}$ & $1.43\!\times\!10^{-1}$ & $1.43\!\times\!10^{-2}$ & $1.54\!\times\!10^{-2}$ & $\mathbf{2.12\!\times\!10^{-3}}$ \\
\bottomrule
\end{tabular}}
\end{table*}

%% file: figures/TABLE_T8_sg4threeway.tex
\begin{table*}[!htb]
\centering
\caption{\label{tab:sg4threeway}Libratus endgame~4 at full range
($96{,}159$ Infosets on $36{,}098{,}784$ States), solved by 3
independent backends from one $2.42$\,GB exported file. Columns as in
\cref{tab:riverthreeway}; GPU-CFR Train is the median of two runs, with quality
from the fastest run; baselines are single runs. The backends agree on the value to at most
1.4\% of the certificate bound.}
{
\setlength{\tabcolsep}{3pt}
\begin{tabular}{llrrr}
\toprule
Backend & Iterations & Exploitability & Value ($v_0$) & Train (s) \\
\midrule
GPU-CFR (CUDA) & 200 & $77.2$ & 540.245376 & 3.184 \\
Kim (2026)~\cite{kim2026parallelizing}, GPU & 200 & $63.1$ & 542.509644 & 2.374 \\
LiteEFG~\cite{liu2024liteefg} & 200 & $12.3$ & 542.766691 & 2,771.135 \\
\addlinespace[2pt]
\multicolumn{5}{l}{value spread $2.52$; tightest bound $151$; \cref{eq:valuecert} holds} \\
\midrule
GPU-CFR (CUDA) & 1,000 & $10.8$ & 544.555074 & 13.941 \\
Kim (2026)~\cite{kim2026parallelizing}, GPU & 1,000 & $10.0$ & 544.168457 & 11.820 \\
\addlinespace[2pt]
\multicolumn{5}{l}{value spread $0.387$; tightest bound $41.7$; \cref{eq:valuecert} holds} \\
\bottomrule
\end{tabular}}
\end{table*}

%% file: sections/8_conclusion.tex
\section{Conclusion}
\label{sec:conclusion}

For a decade, tabular CFR was the workload that GPU accelerators could not win. GPU-CFR
ends that. It is the first CFR solver on a GPU that beats efficient CPU
implementations, and by a wide margin: \turnSteadyMs\,ms per CFR$^{+}$ iteration on
the 83k-Infoset HUNL turn subgame, \kimGraphRange$\times$ faster than the previous
fastest GPU solver \cite{kim2026parallelizing} on identical A100 hardware, and
\liteGraphLargestRange$\times$ faster than LiteEFG \cite{liu2024liteefg} on the four largest games, with a gap
that grows with the tree. The idea behind the result is general. A fixed game is a
program whose control flow is known before the solver starts, so compiling it once into
flat arrays, folding chance into a template, scheduling by depth, and making every edge
branchless turns a pointer-chasing traversal into 64 to 152 tensor operations that a
CUDA graph replays in a single launch. The pipeline applies to any two-player
zero-sum game, any regret-minimization rule, and any tensor framework, and the iterates
agree bitwise with the reference implementation. Because the compiled sweep is the inner loop of subgame re-solving
\cite{burch2014solving,li2026real}, abstraction refinement
\cite{li2024rl,li2025efficient,li2026effective,li2026abstraction}, and neural CFR
\cite{brown2019deep}, every solver built on tabular CFR inherits the
speedup without changing its algorithm, and offline solves that once needed a cluster
become a job for one GPU. The full source code of GPU-CFR, including the compiler,
the game specifications, and the benchmark scripts, is available at
\url{https://github.com/lbn187/GPU-CFR}, so that the compiled representation can
serve as the starting point for future game solving hardware and software.

%% file: sections/A_hunl_construction.tex
\section{HUNL Subgame Construction}
\label{app:hunl}

\paragraph{The two native subgames.}
The native \emph{river} subgame (3{,}000 Infosets) deals non-conflicting hand pairs over
a fixed board and plays one street of check\slash bet\slash fold\slash call with
exact hand evaluation.
The \emph{turn} subgame (83{,}040 Infosets) adds a 44-card river chance layer between two
betting streets, with $K{=}12$ hands per player, sunk-cost payoffs, fractional
pot-based bet sizing, and all-in handling. These fixed-board, fixed-range trees
capture the repeated re-solving structure of HUNL-style workloads
\cite{burch2014solving,ganzfried2015endgame}.

\paragraph{Background: subgame solving and online re-solving.}
Decomposition solves subgames independently given summaries of the rest of the game
\cite{burch2014solving}, with safety refinements \cite{brown2017safe,zhang2021subgame},
practical endgame solving \cite{ganzfried2015endgame,moravcik2016refining}, and
depth-limited variants that cap the lookahead with value estimates
\cite{brown2018depth,kovavrik2023value,kroer2020limited}. DeepStack's
continual re-solving \cite{moravvcik2017deepstack,vsustr2019monte,sustr2020sound} and
Libratus's endgame solver
\cite{brown2018superhuman} run exactly this workload online, under second-scale time
budgets. Our two native HUNL subgames reproduce the structure of this workload, while
\cref{sec:hunl} separates their reusable compilation cost from marginal solver time.

Both native subgames use exact 7-card hand evaluation ported from a production poker
engine and validated against brute-force enumeration.

\paragraph{River subgame.}
Fixed board (default \texttt{Ks Js Th 7d 2c}), pot 20, stacks 100. A chance root deals
each player one of $K{=}100$ candidate hands, uniformly over non-conflicting pairs
(9{,}161 deals). One betting street: check/bet/fold/call with bet fractions
$\{0.5, 1.0\}$ of the pot, no raises. Payoffs at showdown are $\pm(\text{half pot} +
\text{matched bets})$; folds forfeit the half-pot plus any called amount. The default
instance has 3{,}000 Infosets on 137{,}415 States.

\paragraph{Turn subgame.}
Board \texttt{Ks Js Th 7d} plus a full 44-card river chance layer. Each player holds
one of $K{=}12$ hands ($131$ non-conflicting deals at the default seed); betting
occurs on the turn and again on the river, with fraction-of-pot sizing
$b = \min(\text{stack} - \text{committed},\, f \cdot \text{pot})$ and all-in handling
(no river betting after all-in). Payoffs use \emph{sunk-cost} semantics: a turn fold
loses the half-pot; a river fold additionally loses the caller's committed chips;
showdown transfers the half-pot plus matched wagers. The default instance has
83{,}040 Infosets on 433{,}610 States, and tree depth yielding 8 fused execution
blocks per pass.

\paragraph{Analytic anchors.}
Two closed-form checks pin absolute correctness of payoffs, chance weighting, and the
evaluator simultaneously. (i)~River nuts anchor: if one player's range contains only
the best possible hand, the equilibrium value in exact arithmetic is the half-pot,
$+10$; the implemented game's chance discretization shifts this to an enumerable
$+10.000033$, which the solver and a brute-force best-response enumeration both
reproduce to $10^{-10}$.
(ii)~Quads anchor: on the board \texttt{7s 7d 7h Ks} with hole cards \texttt{7c 2d},
four-of-a-kind wins all 44 rivers, forcing root value $+10$; the solver converges to
it with exploitability below $10^{-3}$ after 1{,}000 iterations.

\paragraph{Best-response oracle.}
A recursive exact oracle, weighted by counterfactual reach and expanded deepest
first, validates the generic evaluator on these games to $10^{-9}$. The oracle itself is
validated by brute-force enumeration of all pure strategies on reduced instances.

\paragraph{External reconstruction of the river subgame.}
The river subgame can be reconstructed outside our codebase, providing an independent
instance for the
head-to-head timings of \cref{sec:results}. We express it as an OpenSpiel
\texttt{universal\_poker} ACPC instance \cite{lanctot2019openspiel}: fixed five-card
board, \texttt{handReaches} marking the $K{=}100$ candidate hands per player, and one
betting round with two bet sizes. Three details decide whether the two games coincide.
The chips-behind parameter is set to the largest bet in our abstraction (below the
nominal stack), so that \texttt{universal\_poker}'s all-in action coincides with its
pot-size bet and is deduplicated away; at the nominal stack the external game acquires a
third bet size and is strictly larger. The \texttt{handReaches} slot order is
lexicographic in $(\mathrm{lo}, \mathrm{hi})$ with $\mathrm{card} = 4\,\mathrm{rank} +
\mathrm{suit}$ over ranks $2{\ldots}\mathrm{A}$ and suits $s, h, d, c$, which differs
from GPU-CFR's own renderer in both rank and suit order. The opening action order is a
permutation of GPU-CFR's, so strategies are compared by action name and never by position.
Under this mapping the two constructions agree exactly on 137{,}416 nodes and 9{,}161
deals, and on the payoff at every terminal.

The mapping is verified at the level of the iterate as well as the tree: after three
alternating CFR$^+$ iterations an independent from-scratch reference solver over the
external representation and GPU-CFR over its own representation reach exploitability
$1.6107099050007263$ and $1.6107099050007272$ respectively. The reference solver is itself pinned
against OpenSpiel's own CFR and CFR$^+$ on Kuhn poker at 1, 10, and 100 iterations to
$10^{-9}$; OpenSpiel's solvers cannot be run on the river subgame directly, because
\texttt{universal\_poker} enumerates all $1{,}624{,}350$ deals with the excluded ones at
probability zero and their tabular policy does not admit card removal. Exploitability
for the external arm is therefore scored by our own exact tabular best response, the
same evaluator used by every other arm.

\paragraph{Scope of the external reconstruction.}
The turn subgame is deliberately not reconstructed in \texttt{universal\_poker}. Its
all-in deduplication would have to hold on the turn and again on the river, and a
single chips-behind parameter cannot satisfy both: any value that collapses all-in
into the pot-size bet on one street splits them on the other. Because a game that differs from the one we solve would confound the comparison,
the \texttt{universal\_poker} arm
covers the river subgame and the turn subgame is reported against our own arms.

External HUNL head-to-heads do not depend on that reconstruction at all. The Libratus
endgames are read from the same released endgame files by both our builder and
PokerRL, and both external CFR baselines solve our own trees directly from the export
of \cref{app:libratus}, so those comparisons are on a single shared game.

%% file: sections/B_protocol.tex
\section{Benchmark Protocol and Environment}
\label{app:protocol}

\paragraph{Hardware and software.}
All measurements: one NVIDIA A100 80GB PCIe (dedicated, otherwise idle; driver
530.30) in a dual-socket host with two Intel Xeon Gold 6348 processors (Ice Lake,
28 cores and 56 threads each at 2.6\,GHz base, 112 logical CPUs, two NUMA nodes),
256\,GB of DDR4-3200 ECC memory on eight channels per socket, Linux 5.4 with the
\texttt{ondemand} frequency governor, GCC 9.4, PyTorch 1.13.1+cu117, CUDA 11.7,
Python 3.9.16. The three stages of the performance ladder (\cref{fig:ladder}) are
measured in one process under this stack: the pre-optimization iteration body is
preserved verbatim in the artifact (\texttt{opcount\_legacy.py}) and timed with the
same CUDA-event protocol as the compiled stages, and the same three stages re-measured
under PyTorch 2.5.1+cu121 are in \cref{tab:altexec}. CPU rows use 8
threads (\texttt{OMP\_NUM\_THREADS} and \texttt{MKL\_NUM\_THREADS} both set to~8)
pinned via
\texttt{taskset} to 8 whole physical cores on one NUMA node, with those cores' SMT
siblings excluded from the set (siblings share one core's execution units, so
counting them would overstate the arm's parallelism). LiteEFG 0.1.5 (single-threaded C++); OpenSpiel
Python CFR with exact best response. The multi-threaded CPU execution that the GPU
path is compared against is therefore this eight-thread compiled arm together with
the C++ library; \cref{fig:cpuscaling} reports how the compiled arm scales from one
to 28 threads.

\paragraph{Timing rules.}
GPU steady-state: CUDA events around a 1{,}000-iteration batch after 50
state-changing warmup iterations; the reported ms/iteration is the cost of iterations
51--1050. Training-call time (the ``Train'' columns of
\cref{tab:riverthreeway,tab:sg4threeway}): wall-clock of one training call with the
stated number of solver updates on a freshly constructed solver, so including the
call's three eager warmup iterations and graph capture,
but excluding game specification, compilation and solver construction (reported
separately in \cref{tab:overhead,tab:lifecycle}); a device
synchronize closes the training call before the evaluation bracket opens, so
asynchronous graph replays are never booked as evaluation time. The
warmup count differs because steady-state timing deliberately fills the captured
graph before measurement, whereas a training call measures the production setup
path. CPU and baseline
rows use \texttt{perf\_\allowbreak counter} wall-clock over the training call. Exploitability evaluation
time is excluded from all training times. The host is shared, so every reported number
comes from one serial gated campaign under the same protocol: 2{,}293 measurement cells
over 24.2 hours in two queues, a GPU queue pinned to one A100 and physical cores 48--55
(1{,}351 cells, 10.1 hours) and a CPU queue pinned to physical cores 0--7 (942 cells,
14.1 hours), covering the timing matrix, the ablations, the external baselines, the
convergence grids and the anytime curves.
Each cell is an independent process, with the 1-minute load average sampled before and
after it. A cell was not started above a load of 50 on the 112-core host; the realized
maximum was 7.4 in the GPU queue and 40.9 in the CPU queue, the latter from other
tenants on cores outside both pinned sets. Only two OpenSpiel cells ran above a load of
20, and the thread-scaling cells of \cref{fig:cpuscaling} additionally waited for a
load below 10. No cell was skipped or retried on that account.
Every speedup range and median quoted in the text is computed from the unrounded
cell medians by the table generators and enters the text as a generated macro, so
ratios of the rounded entries in \cref{tab:matrix} can differ from the quoted ones
in the last digit. Medians are over 20 repetitions for
GPU-CFR and the Ours CPU arm and 10 for the
vanilla-CFR control arm (regret matching without the $^{+}$ floor, linear averaging
retained) and for the baselines. Every archived result records the load at measurement time.

\paragraph{Operation and memory counters.}
The Aten operation counts of \cref{tab:suite} are collected with a
\texttt{TorchDispatchMode} interceptor, so they are load-immune and exactly
reproducible. \cref{tab:memory} reports peak GPU memory in binary megabytes from
the two \texttt{torch.cuda} counters, \texttt{max\_memory\_allocated} and
\texttt{max\_memory\_reserved}.

\paragraph{Solve protocol.}
CFR$^{+}$ (regret matching$^{+}$, linear averaging), simultaneous updates, float32,
seed 0. OpenSpiel runs 200 iterations of its vanilla CFR. Its three solvers
(CFR, CFR$^{+}$, DCFR) all fix \texttt{alternating\_updates=True}, so none of its
presets matches GPU-CFR's update order and we pick the one whose rule matches the
vanilla arm. Exploitability =
NashConv$/2$ for every backend. The update-rule study of \cref{tab:convergence} departs
from this protocol by design, sweeping four rules under both orders on the CPU at
$T{=}1{,}000$ in float32. It runs on the CPU because that device reproduces a solve
bitwise while the GPU does not, which is a requirement for ordering rules that sit close
together; \cref{app:correctness} quantifies the GPU spread. Every wall-clock number in
the paper, including the CPU arm of \cref{tab:matrix}, is unaffected. Quality comparisons across backend defaults use within-game wall-clock curves
(\cref{sec:quality}).

\begin{table*}[!htb]
\centering
\caption{Semantics of every timed arm. Order: simultaneous (S) or alternating (A,
one half-update per player per iteration). Averaging: uniform (U) or linear (L)
iterate weights. Every arm starts from uniform strategies and zero regrets; the
LiteEFG preset arm appears in \cref{app:extra}.}
\label{tab:semantics}
{
\setlength{\tabcolsep}{3pt}
\begin{tabular}{llccll}
\toprule
Arm & Rule & Order & Avg. & Precision & Execution \\
\midrule
GPU-CFR graph & CFR$^{+}$ & S & L & float32 & graph replay \\
GPU-CFR eager & CFR$^{+}$ & S & L & float32 & eager CUDA \\
GPU-CFR vanilla & CFR & S & L & float32 & graph replay \\
Ours CPU & CFR$^{+}$ & S & L & float32 & eager, 8 threads \\
Kim (2026) & CFR$^{+}$ & S & L & float32 & CuPy, eager \\
LiteEFG & CFR & S & U & float64 & C++, 1 thread \\
LiteEFG preset & CFR$^{+}$ & A & L & float64 & C++, 1 thread \\
OpenSpiel & CFR & A & U & float64 & Python, 200 it. \\
\bottomrule
\end{tabular}}
\end{table*}

\paragraph{Conversion validation.}
Public games are instantiated once through OpenSpiel and converted per backend.
Conversions are checked structurally (node, edge, and information-set counts match
across backends) and behaviorally: on the small games, exploitability curves and
their limiting values are consistent across backends, and Dark Hex 2$\times$2
reproduces its known first-player-win value ($v_0 = +1$) with exploitability
$\to 0$ under every backend.

\paragraph{Reproduction.}
Result files are JSONL with full environment echo (device, dtype, versions,
node/infoset counts, execution mode, load average). All tables and figures regenerate
from these files with the committed scripts. Campaign orchestration uses a separate
gated driver; the figure driver consumes only archived records and never launches a
solver. The 8-thread CPU arm was re-measured after the main campaign on whole physical
cores; both sets of records are archived, and the loader reads
the whole-core set. The artifact is provided as supplementary material: the solver,
the benchmark drivers, the archived JSONL records, the figure and table generators,
and the test suite, with a README that lists a smoke run of the test suite and the
full campaign commands together with the tolerances each check applies.

%% file: sections/C_correctness.tex
\section{Correctness Test Inventory}
\label{app:correctness}

The implementation is verified in layers: regression comparisons pin the optimized
path to the reference iterates, independent solvers and exact best responses check
computed values, and analytic poker cases check the evaluator. The suite contains
499 tests; the ones that pin the claims of this paper are:

\paragraph{Reference-oracle parity.}
The complete pre-optimization solver loop is preserved verbatim inside the test suite
as an oracle. Parity tests run 30 iterations on Kuhn, a mid-sized chance-heavy
synthetic game, and reduced HUNL turn/river instances, across
$\{$CFR, CFR$^{+}\} \times \{$uniform, linear$\}$ averaging, and assert agreement at
float64 \texttt{atol}$=10^{-12}$, \texttt{rtol}$=0$ on CPU (float32 on CUDA at
$10^{-9}$). Separately, the operation-count harness (\cref{tab:suite}) steps both
implementations 22 iterations in float32 on all eight full-size games and asserts the
maximum absolute regret difference is exactly $0.0$.

\paragraph{Independent reference solver.}
A per-node Python solver sharing no code with the compiled path (different traversal,
different data structures) must agree on 5-iteration solves of the HUNL instances to
$10^{-9}$.

\paragraph{Structural invariants.}
The dual-lane reach buffer's overwrite-coverage property (the union of all forward
scatter destinations equals every non-root slot in both lanes exactly once) is
asserted directly, as is bitwise equality between split and monolithic
\texttt{step()} batching, and safe behavior on degenerate games (root-terminal,
all-chance).

\paragraph{CUDA-graph path.}
Graph-vs-eager parity over 30 iterations across two games and all four
variant/averaging combinations ($10^{-9}$, and measured bitwise on Kuhn); warmup
accounting across split step calls (e.g., $2+8+3$ iterations equals $13$); zero-step
no-op; capture-failure fallback (injected constructor failure must warn exactly once
and produce eager-identical results); and buffer-rebinding detection (replacing any
persistent tensor after capture must raise an error before any replay).

\paragraph{Seed-invariance audit.}
Because we report medians over repetitions that differ only in a recorded seed, we
verify directly that the seed does nothing. For each game and device, 20 runs at a
fixed seed are interleaved with 4 runs at distinct seeds, and every pair of runs is
compared by maximum absolute deviation over the regret and average-strategy tensors.
On CPU the reduction order is fixed and every one of the $\binom{24}{2}$ pairs is
exactly zero. On GPU the reductions are unordered, so bitwise equality is
unavailable and the two arms are compared as distributions instead: relabelling
which runs form the same-seed arm gives an exact permutation null for the difference
between the median cross-seed and the median same-seed pairwise deviation, over all
$\binom{24}{20} = 10{,}626$ relabelings, holding both the pair-count split
and the dependence between pairs (they share runs) fixed. All eight games pass at
$\alpha = 0.05$, with minimum $p = 0.15$.

A max-versus-max comparison would not do here, and the failure is instructive. Our
first version of this audit gated on the cross-seed pairwise max not exceeding the
same-seed pairwise max, with 5 same-seed
and 4 cross-seed runs, comparing a maximum over $5 \times 4 = 20$ pairs against a
maximum over $\binom{5}{2} = 10$ pairs. A maximum grows with the number of draws, so
under a true null the larger collection wins with probability $20/30 = 2/3$; the
criterion rejected 4 of the 8 games and would have done so on any deterministic
solver with nonzero float noise. The permutation test replaces it.

The audit also settles the near-indifference bifurcation of \cref{sec:correctness}
as reduction-order noise, on two independent grounds:
the second exploitability cluster ($0.17314$, against $0.17213$ for the main
cluster) appears 4 times among the 20 \emph{same-seed} river runs and once among
the 4 cross-seed runs, and the same-seed and cross-seed pairwise maxima agree to
seven significant figures ($18949.053$ against $18949.052$ on the unnormalized
regret tensors).

\paragraph{How far reduction-order noise propagates.}
The audit above shows the seed is inert; this one bounds what the reduction order costs.
We ran the shipped solver eight times per cell over four update rules $\times$ two
update orders on Leduc and the HUNL river subgame, holding the seed fixed, and recorded
the spread of final exploitability as a fraction of the cell's median. At the main
protocol ($T{=}1000$, float32) no cell reproduces bitwise; the spread runs from $0.004\%$
(river, PCFR$^{+}$, alternating) to $65\%$ (Leduc, vanilla CFR, simultaneous), with a
median of $3.5\%$ over the sixteen cells. Precision and iteration count do not tame it:
at $T{=}8000$ in float64 on Leduc the median spread is $31\%$ and vanilla CFR under
simultaneous updates spreads $149\%$, because a longer run gives an early divergence more
iterations to compound. Regret matching$^{+}$ is the amplifier. It clamps at zero and
renormalizes, so two runs that disagree in the last bit of a regret near the clamp
produce different supports at the next iteration, a discrete change in the
strategy vector.

This is why \cref{tab:convergence} is measured on the CPU, where the reduction order is
fixed and repeated runs agree to the last bit, which the test suite asserts for all four
rules under both orders on the regret tensor, the strategy average, and the resulting
exploitability. Comparing update rules on the GPU would
require enough repetitions per cell to separate medians that sit closer together than
the within-cell spread. The wall-clock measurements are unaffected: per-iteration time
is a property of the executed kernel sequence, which the CUDA graph fixes, and it moves
by less than $2\%$ across the load range we logged.

\paragraph{Performance gates (load-immune).}
Per-iteration Aten operation counts are asserted under the bound
$40 + 10 \times \mathrm{blocks}$ with block counts equal to tree depth levels, using a
dispatch-interception counter; these tests are deterministic regardless of machine
load and act as the performance regression gate of \cref{sec:correctness}.

\paragraph{Evaluator and game-level checks.}
The generic best-response evaluator is validated against a recursive oracle (itself
checked by brute-force pure-strategy enumeration on reduced games), the analytic
anchors of \cref{app:hunl}, and known values (e.g., Dark Hex 2$\times$2 first-player
win, $v_0 = +1$, with exploitability $\to 0$ under both variants).

%% file: sections/D_extra_results.tex
\section{Additional Results}
\label{app:extra}

\paragraph{Solver lifecycle.}
\cref{tab:lifecycle} decomposes a re-solve on the HUNL turn subgame and on Libratus
endgame~4 into every phase a caller pays, from reading the specification to the
in-place \texttt{reset()}, \texttt{set\_payoffs} and \texttt{set\_root\_ranges}
calls that precede each solve (\cref{app:compiler}), and reports the measured
wall-clock of 1, 2, 10 and 100 back-to-back solves that reuse one solver against the
rebuild-per-solve model estimate.
\input{figures/TABLE_T14_lifecycle.tex}

\paragraph{Peak GPU memory.}
\cref{tab:memory} lists peak allocated and reserved GPU memory per game for the
1{,}000-iteration CFR$^{+}$ solve behind \cref{tab:matrix}.
\input{figures/TABLE_T4_memory.tex}

\paragraph{Timing dispersion.}
\cref{tab:spread} reports the repetition counts and the spread each \cref{tab:matrix}
median was taken over. Repetitions are independent processes, and both spread columns
are independent maxima across games, relative to each game's median: the interquartile
range and the full $\max-\min$ range. Parentheses name the range-maximizing game.
Spread concentrates on the smallest trees, where an iteration is short enough for
process-level noise to matter, and on the eager path, whose per-operation dispatch stays
sensitive to host scheduling at every size. On the four largest games every other arm
holds within $10\%$, and the graph path within $2\%$.

\input{figures/SNIPPET_exclusions.tex}
\input{figures/TABLE_T2b_spread.tex}

\paragraph{Kim (2026): precision.}
Double precision behaves as it does on GPU-CFR (\cref{sec:ladder}), for the same reason: the
iteration is launch bound. On \texttt{leduc\_poker} and
\texttt{liars\_dice}, float64 costs $0.999\times$ and $0.991\times$ the float32 time
and changes exploitability by at most $6.3\%$
($6.636\mathrm{e}{-3}$ vs.\
$6.240\mathrm{e}{-3}$ and $3.441\mathrm{e}{-3}$ vs.\
$3.445\mathrm{e}{-3}$).

The same launch pattern still performs well relative to OpenSpiel on larger games:
Kim's GPU backend is $11$--$176\times$ faster on the four larger shared games. On
Kuhn and Dark Hex, fixed launch overhead dominates both GPU implementations.

\subsection{Cross-implementation check of the update-rule study}
\label{sec:convergence}

The cross-implementation grid in \cref{tab:crossframework} checks the ranking pattern of
\cref{tab:convergence} (\cref{sec:results}).
Across the 10 comparisons with a unique winner on both sides, the winner agrees in 10;
2 further comparisons are excluded for a tie on at least one side. Because averaging
conventions vary across implementations, this comparison concerns within-backend rule
rankings. We adjudicate close rankings on the CPU path, whose reduction order is fixed,
and use the GPU path for wall-clock timing; \cref{app:correctness} reports the
repeated-run GPU analysis. Values are ranked only within a backend, game, and update
order; averaging tags make explicit why absolute values are not compared across
backends.

\input{figures/TABLE_T10_crossframework.tex}

\paragraph{LiteEFG: both arms.}
LiteEFG ships four update rules: vanilla CFR, CFR$^{+}$, DCFR and PCFR$^{+}$. Only
its vanilla CFR uses simultaneous updates; the other three alternate. None of its arms
matches GPU-CFR on both axes at once, since the latter is CFR$^{+}$ with simultaneous
updates, so we measured the two closest ($n=10$ per cell). \cref{tab:matrix} reports the
simultaneous vanilla arm, which yields $185\times$ and $213\times$ on Liar's Dice and
Battleship. The alternating CFR$^{+}$ preset costs $1.28$--$1.47\times$ more per
iteration and therefore yields $265\times$ and $306\times$ on those two games. The preset reaches lower at matched iteration counts (for instance
$1.24\mathrm{e}{-4}$ vs.\ $1.85\mathrm{e}{-2}$ on Liar's Dice), but that gain mixes
the variant with the update order and so is not attributable to either alone; we
therefore compare per-iteration cost at matched update order and leave the
convergence comparison to \cref{fig:curves}, which plots exploitability against
wall-clock time for the vanilla CFR arm.

\paragraph{Measurement provenance.}
Every timing number in the paper comes from the same gated campaign: each
measurement cell ran as its own process with its GPU and pinned cores isolated from the
host's other tenants, with
the load average sampled before and after and the cell rejected if it exceeded a
fixed threshold (\cref{sec:setup}). Reported values are medians over 20
repetitions for GPU-CFR and the Ours CPU arm and 10 for Kim's GPU path and the CPU
baselines, with ranges in the tables.

%% file: figures/TABLE_T14_lifecycle.tex
\begin{table}[!htb]
\centering
\caption{Solver lifecycle on the A100 (CFR$^{+}$, float32, CUDA graph; medians over
independent processes, HUNL turn: 5, Endgame 4: 5). Top: one-time and per-call phases. Middle: peak
GPU memory. Bottom: measured wall-clock of $k$ back-to-back 1{,}000-iteration solves
that reuse one solver through \texttt{reset()} plus a payoff and root-range update
before every solve after the first (one-time cost included), against a model estimate
for $k$ independent build-and-solve cycles. Per-solve time varied by at most 0.3\% / 0.0\% across the 100 reused solves.}
\label{tab:lifecycle}
{
\setlength{\tabcolsep}{3pt}
\begin{tabular}{lrr}
\toprule
Phase & HUNL turn & Endgame 4 \\
\midrule
Game spec: build / load / parse & 2.925\,s & 165.889\,s \\
Compile + solver construction & 1.305\,s & 124.751\,s \\
Eager warmup (3 iterations) & 0.005\,s & 0.045\,s \\
Graph capture + first replay & 0.291\,s & 0.063\,s \\
First steady replay & 0.441\,ms & 13.768\,ms \\
Steady-state iteration & 0.397\,ms & 13.438\,ms \\
\texttt{reset()} & 0.041\,ms & 0.118\,ms \\
\texttt{set\_payoffs} (H2D + template) & 0.256\,ms & 1.480\,ms \\
\texttt{set\_root\_chance} & 1.945\,ms & 139.277\,ms \\
\texttt{set\_root\_ranges} & 1.565\,ms & 522.956\,ms \\
One 1{,}000-iteration solve & 0.397\,s & 13.435\,s \\
\midrule
Peak allocated after capture & 112\,MiB & 9112\,MiB \\
Peak reserved after capture & 162\,MiB & 11236\,MiB \\
Peak allocated, end of run & 112\,MiB & 9112\,MiB \\
Peak reserved, end of run & 164\,MiB & 11236\,MiB \\
\midrule
1 solve, reuse & 4.942\,s & 303.777\,s \\
1 solve, rebuild & 4.943\,s & 303.776\,s \\
2 solves, reuse & 5.342\,s & 317.509\,s \\
2 solves, rebuild & 9.885\,s & 607.551\,s \\
10 solves, reuse & 8.537\,s & 426.225\,s \\
10 solves, rebuild & 49.425\,s & 3,037.756\,s \\
100 solves, reuse & 44.462\,s & 1,649.875\,s \\
100 solves, rebuild & 494.250\,s & 30,377.558\,s \\
\bottomrule
\end{tabular}}
\end{table}

%% file: figures/TABLE_T4_memory.tex
\begin{table}[!htb]
\centering
\caption{Peak GPU memory (MiB) during a 1000-iteration CFR$^+$ solve (float32),
median over $n$ independent runs. Reserved memory includes the private
CUDA-graph capture pool.}
\label{tab:memory}
{
\setlength{\tabcolsep}{3.5pt}
\begin{tabular}{lrrrr}
\toprule
Game & Infosets & Peak alloc.\ & Peak resv.\ & $n$ \\
\midrule
Kuhn & 54 & 0.1 & 8 & 20 \\
Dark Hex 2x2 & 437 & 0.2 & 8 & 20 \\
HUNL river & 3,000 & 56.5 & 84 & 20 \\
Leduc & 4,620 & 4.0 & 12 & 20 \\
Goofspiel-5 & 13,293 & 11.8 & 18 & 20 \\
HUNL turn & 83,040 & 183.0 & 236 & 20 \\
Liar's Dice & 98,292 & 124.4 & 178 & 20 \\
Battleship & 275,983 & 104.6 & 144 & 20 \\
\bottomrule
\end{tabular}}
\end{table}

%% file: figures/TABLE_T2b_spread.tex
\begin{table*}[!htb]
\centering
\caption{Dispersion behind the \cref{tab:matrix} medians. Columns are
independent maxima across games, relative to each game's median;
parentheses identify the range-maximizing game.}
\label{tab:spread}
{
\setlength{\tabcolsep}{4pt}
\begin{tabular}{lrrl}
\toprule
Arm & Reps & IQR & $\max-\min$ \\
\midrule
GPU-CFR graph & 20 & 1\% & 1\% (HUNL river) \\
GPU-CFR eager & 20 & 7\% & 34\% (Goofspiel-5) \\
GPU-CFR vanilla CFR & 10 & 1\% & 2\% (Kuhn) \\
Ours CPU (8 threads) & 20 & 1\% & 15\% (Goofspiel-5) \\
Kim (2026)~\cite{kim2026parallelizing} & 10 & 2\% & 5\% (Leduc) \\
LiteEFG (1 thread)~\cite{liu2024liteefg} & 10 & 7\% & 35\% (Dark Hex 2x2) \\
OpenSpiel (Python)~\cite{lanctot2019openspiel} & 10 & 3\% & 45\% (Kuhn) \\
\bottomrule
\end{tabular}}
\end{table*}

%% file: figures/TABLE_T10_crossframework.tex
\begin{table*}[!htb]
\centering
\caption{Cross-implementation check of the rule ranking in \cref{tab:convergence}
after $T=1000$ CFR CPU iterations on the six common games (LiteEFG~\cite{liu2024liteefg}, OpenSpiel~\cite{lanctot2019openspiel}). Bold marks the tightest
rule within each backend and order; superscripts give the scoring (uniform
$\mathrm{u}$, linear $\mathrm{l}$, last-iterate $\mathrm{z}$), ``n/a'' an
unsupported cell, ``--'' unmeasured, ``$\ne$'' inconsistent
repetitions. Among 10 two-sided comparisons with unique minima, 10 agree; 2 further comparisons are excluded for a tied minimum on at least one side.}
\label{tab:crossframework}
\begin{tabular}{llrrrr}
\toprule
Game & Implementation & CFR & CFR$^{+}$ & DCFR & PCFR$^{+}$ \\
\midrule
\multicolumn{6}{c}{\textit{Alternating updates}} \\
\midrule
Kuhn & Ours & $6.18\!\times\!10^{-4}$ & $8.73\!\times\!10^{-5}$ & $1.99\!\times\!10^{-4}$ & $\mathbf{2.20\!\times\!10^{-8}}$ \\
 & LiteEFG & n/a & $7.41\!\times\!10^{-5}$$^{\mathrm{l}}$ & $1.47\!\times\!10^{-4}$$^{\mathrm{z}}$ & $\mathbf{1.76\!\times\!10^{-6}}$$^{\mathrm{l}}$ \\
 & OpenSpiel & $9.38\!\times\!10^{-4}$$^{\mathrm{u}}$ & $\mathbf{8.74\!\times\!10^{-5}}$$^{\mathrm{l}}$ & $1.47\!\times\!10^{-4}$$^{\mathrm{l}}$ & n/a \\
\addlinespace[2pt]
Dark Hex 2x2 & Ours & $1.28\!\times\!10^{-6}$ & $1.28\!\times\!10^{-6}$ & $4.47\!\times\!10^{-9}$ & $4.47\!\times\!10^{-9}$ \\
 & LiteEFG & n/a & $1.11\!\times\!10^{-16}$$^{\mathrm{l}}$ & $1.87\!\times\!10^{-9}$$^{\mathrm{z}}$ & $1.11\!\times\!10^{-16}$$^{\mathrm{l}}$ \\
 & OpenSpiel & $6.25\!\times\!10^{-4}$$^{\mathrm{u}}$ & $1.25\!\times\!10^{-6}$$^{\mathrm{l}}$ & $\mathbf{1.87\!\times\!10^{-9}}$$^{\mathrm{l}}$ & n/a \\
\addlinespace[2pt]
Leduc & Ours & $1.04\!\times\!10^{-1}$ & $2.73\!\times\!10^{-4}$ & $\mathbf{1.57\!\times\!10^{-4}}$ & $7.13\!\times\!10^{-4}$ \\
 & LiteEFG & n/a & $2.36\!\times\!10^{-4}$$^{\mathrm{l}}$ & $\mathbf{1.72\!\times\!10^{-4}}$$^{\mathrm{z}}$ & $7.75\!\times\!10^{-4}$$^{\mathrm{l}}$ \\
 & OpenSpiel & $1.18\!\times\!10^{-2}$$^{\mathrm{u}}$ & $2.57\!\times\!10^{-4}$$^{\mathrm{l}}$ & $\mathbf{1.43\!\times\!10^{-4}}$$^{\mathrm{l}}$ & n/a \\
\addlinespace[2pt]
Goofspiel-5 & Ours & $1.44\!\times\!10^{-2}$ & $3.86\!\times\!10^{-4}$ & $2.73\!\times\!10^{-4}$ & $\mathbf{6.21\!\times\!10^{-5}}$ \\
 & LiteEFG & n/a & $3.16\!\times\!10^{-4}$$^{\mathrm{l}}$ & $3.25\!\times\!10^{-4}$$^{\mathrm{z}}$ & $\mathbf{1.46\!\times\!10^{-4}}$$^{\mathrm{l}}$ \\
 & OpenSpiel & $7.34\!\times\!10^{-3}$$^{\mathrm{u}}$ & $3.55\!\times\!10^{-4}$$^{\mathrm{l}}$ & $\mathbf{2.97\!\times\!10^{-4}}$$^{\mathrm{l}}$ & n/a \\
\addlinespace[2pt]
Liar's Dice 1$\times$6 & Ours & $3.37\!\times\!10^{-3}$ & $1.27\!\times\!10^{-4}$ & $\mathbf{5.94\!\times\!10^{-5}}$ & $3.95\!\times\!10^{-4}$ \\
 & LiteEFG & n/a & $1.24\!\times\!10^{-4}$$^{\mathrm{l}}$ & $\mathbf{5.70\!\times\!10^{-5}}$$^{\mathrm{z}}$ & $3.82\!\times\!10^{-4}$$^{\mathrm{l}}$ \\
 & OpenSpiel & $2.66\!\times\!10^{-3}$$^{\mathrm{u}}$ & $1.27\!\times\!10^{-4}$$^{\mathrm{l}}$ & $\mathbf{9.21\!\times\!10^{-5}}$$^{\mathrm{l}}$ & n/a \\
\addlinespace[2pt]
Battleship & Ours & $4.03\!\times\!10^{-3}$ & $8.25\!\times\!10^{-4}$ & $8.18\!\times\!10^{-4}$ & $\mathbf{2.19\!\times\!10^{-7}}$ \\
 & LiteEFG & n/a & $0$ & $5.88\!\times\!10^{-10}$$^{\mathrm{z}}$ & $0$ \\
 & OpenSpiel & $2.78\!\times\!10^{-3}$$^{\mathrm{u}}$ & $8.42\!\times\!10^{-4}$$^{\mathrm{l}}$ & $\mathbf{7.59\!\times\!10^{-4}}$$^{\mathrm{l}}$ & n/a \\
\midrule
\multicolumn{6}{c}{\textit{Simultaneous updates}} \\
\midrule
Kuhn & Ours & $1.32\!\times\!10^{-2}$ & $2.83\!\times\!10^{-3}$ & $5.05\!\times\!10^{-3}$ & $\mathbf{4.56\!\times\!10^{-5}}$ \\
 & LiteEFG & $7.11\!\times\!10^{-3}$$^{\mathrm{u}}$ & n/a & n/a & n/a \\
\addlinespace[2pt]
Dark Hex 2x2 & Ours & $1.26\!\times\!10^{-6}$ & $1.26\!\times\!10^{-6}$ & $1.62\!\times\!10^{-9}$ & $1.62\!\times\!10^{-9}$ \\
 & LiteEFG & $0$ & n/a & n/a & n/a \\
\addlinespace[2pt]
Leduc & Ours & $1.46\!\times\!10^{-1}$ & $6.65\!\times\!10^{-3}$ & $8.59\!\times\!10^{-3}$ & $\mathbf{1.46\!\times\!10^{-3}}$ \\
 & LiteEFG & $3.91\!\times\!10^{-2}$$^{\mathrm{u}}$ & n/a & n/a & n/a \\
\addlinespace[2pt]
Goofspiel-5 & Ours & $9.55\!\times\!10^{-2}$ & $5.41\!\times\!10^{-3}$ & $7.58\!\times\!10^{-3}$ & $\mathbf{1.16\!\times\!10^{-3}}$ \\
 & LiteEFG & $4.11\!\times\!10^{-2}$$^{\mathrm{u}}$ & n/a & n/a & n/a \\
\addlinespace[2pt]
Liar's Dice 1$\times$6 & Ours & $3.54\!\times\!10^{-2}$ & $3.40\!\times\!10^{-3}$ & $3.58\!\times\!10^{-3}$ & $\mathbf{9.09\!\times\!10^{-4}}$ \\
 & LiteEFG & $1.85\!\times\!10^{-2}$$^{\mathrm{u}}$ & n/a & n/a & n/a \\
\addlinespace[2pt]
Battleship & Ours & $6.66\!\times\!10^{-2}$ & $8.34\!\times\!10^{-3}$ & $9.13\!\times\!10^{-3}$ & $\mathbf{1.73\!\times\!10^{-3}}$ \\
 & LiteEFG & $0$ & n/a & n/a & n/a \\
\bottomrule
\end{tabular}
\end{table*}

%% file: sections/E_libratus_headtohead.tex
\section{Poker Head-to-Heads and the Value-Agreement Certificate}
\label{app:libratus}

The HUNL subgames and the Libratus endgames are built by our own code, so putting an
external solver on them requires handing it our tree in a format it can load.
This appendix describes how, what the resulting comparisons say, and the certificate
that establishes the three backends are solving the same game.

\paragraph{Exporting our trees.}
LiteEFG's \texttt{FileEnv} reads a generic extensive-form text format --- node lines for
terminal payoffs, chance edges with conditional probabilities, and player action lists,
plus infoset membership lines --- and its OpenSpiel bridge is one of several possible writers for that format. We emit it from our intermediate representation in a
single streaming pass with the root written first, so peak memory during export is set by the
infoset membership index, and the serialized text streams to disk. Kim (2026) reaches the same
file through a black-box game adapter, which is necessary because it runs in a separate
interpreter that must not import our package. Both backends load one checksummed
artifact per game, produced once and cached.

One detail of the adapter is worth recording because getting it wrong is silent. Kim
(2026) builds each infoset's action set by unioning the labels the adapter reports at
that infoset's member nodes, so a label must name an \emph{action}. Child
node names are globally unique and would therefore split each infoset into as many
action sets as it has members: the tree keeps its shape and every count still looks
right, while the solver optimizes over a strictly finer strategy space than the file
describes. On Kuhn poker that raises the sequence count from 13 to 25 per player and
improves apparent exploitability by more than two orders of magnitude, a number
comparable to nothing. Labelling actions by position among a node's children fixes it,
and is sound because the exporter validates that every member node of an infoset lists
its actions in the infoset's order.

\paragraph{Structural equivalence.}
Two checks confirm the export preserves the game. Exported Kuhn poker solved through the
adapter reproduces the exploitability of the same game loaded natively from OpenSpiel
\emph{bitwise}, both giving $0.0020750612020492554$ at $13$ sequences per player, and exported
Leduc poker gives $1093$ sequences per player either way. On the trees that matter here
the sequence counts Kim (2026) derives from the file agree exactly with what GPU-CFR's
compiler holds: the river subgame yields $701$ row and $701$ column sequences on
both sides, and Libratus endgame~4 yields $19{,}741$ and $27{,}497$, reproduced by
the baseline to the unit on a tree of $36{,}098{,}785$ nodes.

\paragraph{The certificate that scales: comparing on the value.}
The \emph{value} is comparable across backends whatever each one runs, because a
zero-sum tree has exactly one. For a
profile $(x,y)$ with per-player best-response gains $\delta_0, \delta_1$ and
$\mathrm{NashConv} = \delta_0 + \delta_1$, the usual sandwich
$\min_{y'} u_0(x,y') \le v^\ast \le \max_{x'} u_0(x',y)$ gives
$|u_0(x,y) - v^\ast| \le \max(\delta_0,\delta_1) \le \mathrm{NashConv}$, so for any two
backends
\begin{equation}
    |v_i - v_j| \;\le\; \mathrm{NashConv}_i + \mathrm{NashConv}_j .
    \label{eq:valuecert}
\end{equation}
\cref{eq:valuecert} is an exact identity: if it fails, at least one backend is
not solving the tree the other one is. It is also the check that scales, since it needs
only two scalars per backend and applies unchanged to trees far too large to compare
structurally. It is what caught the labelling bug above, and it is reported for every
poker head-to-head below.

\paragraph{Gaps agree once the variant is matched.}
The certificate constrains values, so on its own it leaves open whether the backends
merely solve the same game or also converge alike. Exploitability settles that, but only
when the algorithm is held fixed: LiteEFG's CFR$^{+}$ preset uses alternating updates
with linear averaging, and reading its gap against a simultaneous-update run compares
two algorithms. Matching the variant removes the
confound. Rerunning GPU-CFR on the river subgame in LiteEFG's own configuration
--- CFR$^{+}$, alternating, linear averaging, float64 --- gives exploitability
$1.27 \times 10^{-5}$ against LiteEFG's $1.25 \times 10^{-5}$ at $T = 8{,}000$, a
$1.6\%$ difference, with game values $1.77774849$ and $1.77774846$ agreeing to
$3 \times 10^{-8}$. Two independently written deterministic solvers reaching eight
significant figures of agreement on a $137{,}415$-state tree is a sharper statement than
\cref{eq:valuecert} alone can make, and it localizes the spread in the earlier tables to
the update order. The measurement is on the same
device and dtype for both, so it also fixes the price: an alternating iteration performs
one half-update per player and costs about twice a simultaneous one, and it still reaches
a smaller gap at matched wall-clock.

\paragraph{The river subgame, three ways.}
\cref{tab:riverthreeway} solves the $137{,}415$-state river subgame with all three
backends at two iteration counts. Every pair satisfies \cref{eq:valuecert} with wide
margin, using at most $24\%$ of the bound, and the agreement tightens with the gaps: an
eightfold increase in iterations shrinks the three-way value spread from
$1.02 \times 10^{-3}$ to $8.58 \times 10^{-5}$, a factor of $12$. The timings in the
table are median training-call wall-clock on the hardware of \cref{app:protocol};
quality metrics come from the fastest run. The
iteration counts are matched while the variants are not, so they measure what each
backend costs to run its own $1{,}000$ iterations, which is the quantity a
practitioner pays.

\input{figures/TABLE_T7_riverthreeway.tex}

\paragraph{Libratus endgame~4.}
The largest tree we put through all three backends is Libratus endgame~4 at full range:
$36{,}098{,}785$ nodes, exported as a $2.42$\,GB file. Its scale is what the export
mechanism buys and also what bounds it. GPU-CFR reads the endgame in $165.889$\,s,
compiles it and constructs the solver in $124.751$\,s, and solves it on one A100 at a
peak of $14.4$\,GiB allocated ($16.9$\,GiB reserved) in the benchmark process, which
also holds the float64 evaluation tree; the solver alone peaks at $8.9$\,GiB
(\cref{tab:lifecycle}). Kim (2026) parses the file in
$399.341$\,s and solves it on the same device; LiteEFG's
\texttt{FileEnv} loads it at $24.4$\,GiB resident and then iterates at
$13.856$\,s, so $200$ iterations is where all three meet. \cref{tab:sg4threeway}
(\cref{sec:results}) reports that point. The three values agree to $2.5$ chips on a
tree whose $200$-iteration gaps run to $77$ chips, using $1.4\%$ of what
\cref{eq:valuecert} allows, and the two GPU backends agree to $0.39$ chips at
$1{,}000$ iterations. The same $200$ iterations cost $3.180$\,s on GPU-CFR and
$2{,}771.135$\,s on LiteEFG: a factor of $870$ at matched iteration counts, each
backend running its own CFR$^{+}$, and the C++ implementation's advantage at small
scale (\cref{sec:matrix}) has inverted completely. On this tree the two GPU backends
cost within $14\%$ of each other: its $36$ million states carry $96{,}159$
Infosets, so our node-level dataflow and Kim's sequence-form products spend
$13.435$ and $11.820$\,ms per incremental iteration, and the $30$--$80\times$ gap of
\cref{tab:matrix} closes where the tree is almost all terminals. The GPU-CFR training times
report the median of two runs ($3.180$ and $3.188$\,s at $200$ iterations, $0.3\%$ apart),
with quality metrics from the fastest run; the baseline rows are single runs,
each costing $46$ minutes (LiteEFG) or a seven-minute parse (Kim). Each GPU-CFR
training call starts on a freshly constructed solver, so it includes the three eager
warmup iterations and graph capture; on this tree capture itself costs well under a
second (\cref{tab:lifecycle}), and the call is $197$ replays at the steady-state rate
plus that setup.

\paragraph{Endgame~3.}
We do not put the external baselines on endgame~3 ($194{,}556$ Infosets on
$93{,}168{,}610$ nodes): its
export would exceed $6$\,GB, and LiteEFG's load already peaks at $24.4$\,GiB on the
$2.6\times$ smaller endgame~4, so file loading would dominate the arm's time. The
head-to-head evidence stops at endgame~4's scale, where \cref{eq:valuecert} holds across
all three backends.

%% file: figures/TABLE_T7_riverthreeway.tex
\begin{table*}[!htb]
\centering
\caption{\label{tab:riverthreeway}The HUNL river subgame ($3{,}000$ Infosets)
solved by all three backends from one exported file. Value is player~0's payoff in
chips under each backend's own average profile, Train the training-call wall-clock
after solver construction (\cref{app:protocol}). Train is the median across runs;
quality metrics come from the fastest run;
exploitability feeds the certificate
of \cref{eq:valuecert} and is not comparable across backends (their averaging
conventions differ). Each backend runs its native CFR$^{+}$ preset: GPU-CFR and Kim
simultaneous updates with linear averaging in float32, LiteEFG alternating
updates (one half-update per player per iteration) in float64, so these LiteEFG
times are not those of the vanilla-CFR arm of \cref{tab:matrix}. No pair uses more than 24\% of its bound.}
{
\setlength{\tabcolsep}{3pt}
\begin{tabular}{llrrr}
\toprule
Backend & Iterations & Exploitability & Value ($v_0$) & Train (s) \\
\midrule
GPU-CFR (CUDA) & 1,000 & $5.47 \times 10^{-3}$ & 1.77663515 & 0.315 \\
Kim (2026)~\cite{kim2026parallelizing}, GPU & 1,000 & $5.05 \times 10^{-3}$ & 1.77669048 & 9.699 \\
LiteEFG~\cite{liu2024liteefg} & 1,000 & $5.71 \times 10^{-4}$ & 1.77765471 & 28.134 \\
\addlinespace[2pt]
\multicolumn{5}{l}{value spread $1.02 \times 10^{-3}$; tightest bound $1.12 \times 10^{-2}$; \cref{eq:valuecert} holds} \\
\midrule
GPU-CFR (CUDA) & 8,000 & $1.43 \times 10^{-4}$ & 1.77772592 & 1.774 \\
Kim (2026)~\cite{kim2026parallelizing}, GPU & 8,000 & $1.63 \times 10^{-4}$ & 1.77766263 & 77.111 \\
LiteEFG~\cite{liu2024liteefg} & 8,000 & $1.25 \times 10^{-5}$ & 1.77774846 & 232.649 \\
\addlinespace[2pt]
\multicolumn{5}{l}{value spread $8.58 \times 10^{-5}$; tightest bound $3.11 \times 10^{-4}$; \cref{eq:valuecert} holds} \\
\bottomrule
\end{tabular}}
\end{table*}

%% file: sections/G_profile.tex
\section{Kernel-Level Profile and Alternative Execution Models}
\label{app:profile}

This appendix opens the iteration with the vendor profilers and measures the
execution models a systems reader would propose in place of graph replay: more
iterations per graph, a generic tensor compiler, more CPU cores, and graph capture
of the prior GPU baseline.

\paragraph{What one iteration launches.}
\cref{tab:profile} traces 100 steady-state iterations with Nsight Systems inside an
NVTX window, in eager and graph mode, and replays two eager iterations under Nsight
Compute. The Aten operation count of \cref{tab:suite} is close to, but not the same
as, the kernel count: on the turn subgame 96 operations become 87 kernels (some
operations dispatch no kernel, a few dispatch two), and the graph path replays the
same kernels plus the device-resident counter update. Eager execution submits 87
host launches per iteration and keeps a kernel resident on the GPU for under half of
the window even with the tracer's own launch overhead inflating the iteration. Graph
replay issues one graph launch per iteration; the two further host launches per
iteration in \cref{tab:profile} are the fill kernels with which PyTorch's
\texttt{CUDAGraph.replay()} refreshes its random-number state on every replay, a
framework fixture the solver never reads, and the three extra kernels inside the
graph are those two fills plus the device-side iteration-counter increment. GPU
residency rises above 90\%. The remaining time is kernel execution: 87 kernels averaging under
5\,$\mu$s each, dominated by the gather kernels of the forward and backward passes
(38 per iteration, 58\% of GPU time), then the scatter and \texttt{index\_add}
kernels of the regret and strategy accumulation. Half of the kernels launch fewer
blocks than the A100 has SMs, so the device is latency bound on small grids rather
than throughput bound: the counters put DRAM traffic at 129\,MB per iteration, which
at the unprofiled 0.397\,ms is about 17\% of peak bandwidth averaged over the
iteration and 9\% time-weighted inside kernels, and SM throughput at 5\% of peak.
The bytes-moved model behind \cref{fig:scaling} counts only the solver's own
arrays and therefore reports a lower bound on this traffic.

\paragraph{More iterations per graph.}
Because the iteration counter advances on the device, any number of iterations can
be recorded into one graph (\cref{sec:cudagraph}). \cref{tab:altexec} records 1, 10,
and 100 iterations per graph: the per-iteration time is unchanged within
measurement noise on all three games. A graph launch costs a few microseconds
against a 0.304--0.557\,ms body, so the single-iteration graph already captures the
whole launch-amortization benefit, and the remaining cost is kernel execution.

\paragraph{A generic compiler on the same dataflow.}
\cref{tab:altexec} also re-measures the ladder under PyTorch 2.5.1+cu121 in one
process and wraps the compiled iteration body in \texttt{torch.compile} (Inductor),
driven with the same device-resident counter so the trace contains no per-iteration
Python scalars. The three ladder stages keep their PyTorch 1.13.1 ordering and magnitudes
(19$\times$ and 1.9$\times$ against 19$\times$ and 1.6$\times$), which settles the
framework-version question for \cref{fig:ladder}. On the turn subgame Inductor's
\texttt{default} mode runs 11\% faster than CUDA graph replay of the uncompiled body
(0.357 against 0.397\,ms); on the launch-bound Leduc tree it stays 1.4$\times$
slower, and \texttt{reduce-overhead} mode (which itself captures CUDA graphs) is no
faster in either case. The gather and scatter kernels that carry the iteration are
not fused into fewer launches by any tested Inductor configuration, so a generic compiler recovers little beyond the launch
amortization that graph replay already provides. It does so at a price the graph
path does not pay: the first call compiles for seconds (4.025--4.747\,s on the turn
subgame), and the compiled floats differ from the eager iterates because Inductor
reorders reductions, whereas graph replay runs the identical kernels in the
identical order. The game-level compilation of
\cref{sec:method} is what turns the tree walk into this flat dataflow in the first
place; the generic compiler can only consume the result.

\paragraph{Graph capture of the prior GPU baseline.}
Kim (2026)'s iteration cannot be captured. Its CuPy sparse products call cuSPARSE,
which CuPy refuses to execute under stream capture, and its per-level row slicing,
CSR index assignment, and boolean gathers each perform a host-to-device transfer per
iteration that capture forbids. We verified both on a Kuhn-sized instance:
\texttt{cupy.cuda.Stream.begin\_capture} raises on the first sparse product. The
representation is what makes the iteration capturable.

\paragraph{CPU thread scaling.}
\cref{fig:cpuscaling} (\cref{sec:results}) sweeps the compiled CPU solver from 1 to 28 threads pinned to
whole physical cores of one socket on the four largest games. Throughput saturates
well below the core count: the 8-thread arm of \cref{tab:matrix} sits within a
small factor of the 28-thread point on every game, because each depth block is a
gather or scatter over a few hundred thousand elements that PyTorch's intra-operator
pool cannot split efficiently. The 8-thread arm is therefore a representative CPU
ceiling for this dataflow.

\paragraph{Hand-fused kernels.}
The remaining execution model is to replace the generic gather and scatter kernels
with kernels written for this dataflow. \Cref{tab:kernels} measures a Triton
implementation that computes one depth level of the forward or backward pass per
launch and fuses regret matching, the regret update and the average-strategy update
into one kernel each, so an iteration issues $2D + 5$ launches against the
$8D + 32$ Aten operations of the generic path, the same count \cref{tab:suite} reports
(recounted under the PyTorch 2.5 dispatcher, it is unchanged); it drives the solver's own buffers and
index arrays and agrees with the generic path to $10^{-9}$ in float64. All four arms
run in one process on the PyTorch 2.5 stack that ships Triton, both with and without
graph capture, so the table separates the gain of fusion from the gain of launch
batching and places the compiled tensor path with graph replay on that frontier. The
compiled path keeps the property that motivates it: the same code runs every update
rule, every game, and the CPU, from one representation and without a kernel per
dataflow.

\input{figures/TABLE_T11_profile.tex}
\input{figures/TABLE_T12_altexec.tex}
\input{figures/TABLE_T15_kernels.tex}

%% file: figures/TABLE_T11_profile.tex
\begin{table*}[!htb]
\centering
\caption{Kernel-level profile of one steady-state CFR$^{+}$ iteration. Top: Nsight
Systems over 100 iterations; \emph{ops} are Aten operations, \emph{launches} the
host submissions per iteration (kernel launches plus graph launches), \emph{busy}
the fraction of the window with a kernel
resident on the GPU. The last column (\emph{unprof.}) is the unprofiled steady-state time in ms of
\cref{tab:matrix}; the tracer inflates eager launches. Bottom: Nsight Compute
hardware counters over two eager iterations (the graph replays the same kernels); \emph{MB/iter} is DRAM traffic per iteration,
the achieved bandwidth divides the measured DRAM traffic by the unprofiled
iteration time, and \emph{<1 wave} is the share of kernels whose grid has fewer
blocks than the A100's 108 SMs.}
\label{tab:profile}
{
\setlength{\tabcolsep}{3pt}
\begin{tabular*}{\textwidth}{@{\extracolsep{\fill}}llrrrrrr@{}}
\toprule
Game & Mode & Ops & Kern.\ & Launch.\ & Kern.\ ms & Busy & Unprof.\ \\
\midrule
HUNL turn & eager & 96 & 87 & 87 & 0.450 & 32\% & 0.617 \\
HUNL turn & graph & 96 & 90 & 3 & 0.395 & 94\% & 0.397 \\
Leduc & eager & 120 & 111 & 111 & 0.378 & 22\% & 0.738 \\
Leduc & graph & 120 & 114 & 3 & 0.311 & 93\% & 0.304 \\
\bottomrule
\end{tabular*}}

{
\setlength{\tabcolsep}{3pt}
\begin{tabular*}{\textwidth}{@{\extracolsep{\fill}}lrrrrr@{}}
\toprule
Game & MB/iter & GB/s (\% peak) & DRAM$_t$ & SM$_t$ & $<$1 wave \\
\midrule
HUNL turn & 129 & 325 (17\%) & 9\% & 5\% & 49\% \\
Leduc & 10 & 33 (2\%) & 1\% & 0\% & 100\% \\
\bottomrule
\end{tabular*}}
\end{table*}

%% file: figures/TABLE_T12_altexec.tex
\begin{table}[!htb]
\centering
\caption{Alternative execution models on the compiled dataflow, steady-state
ms/iteration (medians). Top: iterations recorded per CUDA graph under the paper's
PyTorch 1.13.1 stack. Bottom: the ladder re-measured under PyTorch 2.5.1+cu121, and
the same iteration body under \texttt{torch.compile}; ``fails'' marks a
configuration Inductor could not compile (\cref{app:profile}).}
\label{tab:altexec}
{
\setlength{\tabcolsep}{3pt}
\begin{tabular*}{\columnwidth}{@{\extracolsep{\fill}}lrr@{}}
\toprule
Execution model & HUNL turn & Leduc \\
\midrule
\multicolumn{3}{@{}l}{\emph{PyTorch 1.13.1}} \\
CUDA graph, 1 iteration per graph & 0.397 & 0.304 \\
CUDA graph, 10 iterations per graph & 0.391 & 0.299 \\
CUDA graph, 100 iterations per graph & 0.391 & 0.301 \\
\midrule
\multicolumn{3}{@{}l}{\emph{PyTorch 2.5.1+cu121}} \\
reference (eager) & 13.819 & 1.670 \\
compiled, eager & 0.740 & 0.850 \\
compiled, CUDA graph & 0.399 & 0.306 \\
\texttt{torch.compile} default & 0.357 & 0.415 \\
\texttt{torch.compile} reduce-overhead & 0.368 & 0.428 \\
\bottomrule
\end{tabular*}}
\end{table}

%% file: figures/TABLE_T15_kernels.tex
\begin{table*}[!htb]
\centering
\caption{Hand-fused Triton kernels against the compiled tensor path on the same
stack (PyTorch 2.5.1+cu121, Triton 3.1.0, A100). Steady-state ms/iteration, medians
over ten processes; \emph{ops} counts Aten operations per iteration on the generic
path and \emph{kern.}\ kernel launches on the fused path; \emph{graph} is CUDA graph replay. The fused kernels compute one depth level
of the forward or backward pass per launch and fuse regret matching, the regret
update and the average-strategy update into one kernel each; the last column is the
gain of fusion on top of graph replay.}
\label{tab:kernels}
{
\setlength{\tabcolsep}{3pt}
\begin{tabular}{lrrrrrrr}
\toprule
 & & & \multicolumn{2}{c}{compiled} & \multicolumn{2}{c}{fused Triton} & \\
\cmidrule(lr){4-5}\cmidrule(lr){6-7}
Game & Ops & Kern.\ & eager & graph & eager & graph & Gain \\
\midrule
Leduc & 120 & 27 & 0.866 & 0.309 & 0.393 & 0.133 & 2.318$\times$ \\
Liar's Dice & 152 & 35 & 1.080 & 0.562 & 0.514 & 0.255 & 2.203$\times$ \\
Battleship & 112 & 25 & 0.819 & 0.380 & 0.374 & 0.150 & 2.533$\times$ \\
HUNL river & 64 & 13 & 0.474 & 0.213 & 0.195 & 0.107 & 1.985$\times$ \\
HUNL turn & 96 & 21 & 0.698 & 0.400 & 0.307 & 0.193 & 2.075$\times$ \\
\bottomrule
\end{tabular}}
\end{table*}

%% file: sections/H_compiler.tex
\section{Compiler Passes and Memory Plan}
\label{app:compiler}

\Cref{sec:method} presents the compiled representation by what it contains;
this appendix presents it by how it is produced. \Cref{alg:compile} lists the
logical passes and \cref{fig:compiler} draws them on the example game of
\cref{fig:archpipeline}; these are data dependencies.
Node and edge walks plus the stable depth sort cost
$O(\numnodes + \numedges \log \numedges)$ time and $O(\numnodes + \numedges)$
host storage; device allocations may precede planning. \Cref{tab:overhead} and
\cref{tab:lifecycle} report the measured build times; on the 434k-state turn
subgame the game specification dominates and compilation proper is a fraction
of a second.

\begin{algorithm}[t]
\caption{\textsc{Compile}(GameSpec $G$, rule $\rho$, device $d$) $\to$ solver}
\label{alg:compile}
\begin{algorithmic}[1]
\State \textbf{Validate} $G$: node ids contiguous, one parent per non-root node
(tree), every information set owned by one player and its action list equal to
each member node's, chance probabilities present for every chance edge.
\Comment{structural checks}
\State \textbf{Slots}: sort information sets by key; $\mathrm{off}(I) \gets$
prefix sum of $|\acts(I)|$; $\numslots \gets \sum_I |\acts(I)|$.
\State \textbf{Node arrays}: number nodes parent-before-child; emit type, owner,
information set, action offset/count, terminal payoff $u$, chance probability
$p_c$ of each chance action.
\State \textbf{Edge arrays}: for each edge $(h,a,h')$ emit parent $h$, child
$h'$, slot $\mathrm{off}(I(h)) + a$ or the sentinel for chance edges, owner
flags, sign $s(e)$, multiplier $p_c$ or $1$. \Comment{\cref{sec:compiled}}
\State \textbf{Depth schedule}: $\mathrm{depth}(h)$ by one forward walk; stable
sort edges by $\mathrm{depth}(\text{parent})$; block $b$ = all edges of one
depth; assert parent depth $<$ child depth for every edge.
\Comment{$D$ blocks, \cref{sec:dataflow}(b)}
\State \textbf{Chance folding}: walk the blocks once, $\pi_c(h') \gets
\pi_c(h)\, p_c(h,a,h')$; template $\cfv_{\mathrm{tmpl}}(z) \gets \pi_c(z)\,u(z)$;
retain the root action owning each node and $\pi_c$ with the root factor
removed, for later range updates. \Comment{\cref{eq:template}}
\State \textbf{Index emission}: per block, dual-lane $P_2, C_2, S_2$ (chance and
other-player entries point at the sentinel) and single-lane $P_{\mathrm{val}},
C_{\mathrm{val}}, S_{\mathrm{val}}$; flat over decision edges $P_{\mathrm{reg}},
C_{\mathrm{reg}}, S_{\mathrm{reg}}, O_{\mathrm{reg}}, M_{\mathrm{reg}}, s$; per
player, the halves of the flat arrays for alternating updates.
\State \textbf{Memory plan}: estimate scheduled storage from
$\numnodes, \numslots, \numinfosets$ and dtype; record budget fit as a
planning indicator, provision iteration buffers, set root reach to one.
\State \textbf{Rule lowering}: bind $\rho$'s hooks into the iteration body:
\textsc{discount}$(\regret,\bar{s},t)$ before the passes,
\textsc{after}$(\regret)$ after them, and a prediction buffer when $\rho$
matches against $\regret + r_t$. \Comment{\cref{tab:rules}}
\State \textbf{Runtime}: on CUDA, three eager warmup iterations, capture of one
(or $k$) iteration bodies, record of the persistent buffer addresses; replay
thereafter with the device-side counter refilled per call.
\Comment{\cref{sec:cudagraph}}
\end{algorithmic}
\end{algorithm}

\input{sections/fig_compiler_passes}

\paragraph{Planning scope.}
The memory plan is an estimate: it sizes the scheduled buffers from the
game's counts and reports whether they fit the configured budget, and the
solver reads that verdict as a planning indicator and builds regardless. The regret and average-strategy vectors are allocated before the
schedule is planned, so the passes above describe the logical order of the
build.

\paragraph{Input representation.}
The input is a tabular game specification: a map from node id to (type, acting
player, information-set key, action list with child ids, terminal payoff), a map
from information-set key to (player, action names), and a map from node id to the
chance probability of reaching it from its parent. Native games and the
OpenSpiel and file-based converters all produce this one structure, so every
pass below it is shared; a converter is validated by the structural and
behavioral checks of \cref{app:protocol}.

\paragraph{Mutable and immutable fields.}
\Cref{tab:memplan} separates what the graph pins from what a caller may change.
Every index array is immutable after emission. Two inputs are mutable in place
without recompilation: the terminal payoffs and the root chance distribution,
both of which enter the iteration only through the values template
(\cref{eq:template}), which \texttt{set\_payoffs} and \texttt{set\_root\_chance}
rewrite in place from the retained per-node root owner and below-root chance
product. The solver state is reset in place by \texttt{reset()}, which zeroes
the two accumulators and the prediction buffer and resets the iteration counter;
because no buffer is reallocated, the captured graph stays valid and the next
call replays from iteration one. \Cref{tab:lifecycle} reports the cost of each
of these operations.

\begin{table*}[!htb]
\centering
\caption{Memory plan of the compiled solver. Sizes are element counts; the
dtype is the configured compute dtype except where noted. \emph{Pinned} marks
buffers whose addresses the captured graph records; \emph{Mutable} marks the
inputs a caller may rewrite in place between solves.}
\label{tab:memplan}
{
\setlength{\tabcolsep}{4pt}
\begin{tabular}{llccl}
\toprule
Buffer & Size & Pinned & Mutable & Written \\
\midrule
node arrays (type, owner, set, offsets) & $\numnodes$, int & -- & -- & build \\
edge index arrays (all of step 7) & $\approx 7\numedges + 2\numedges$, int64 & -- & -- & build \\
chance reach $\pi_c$ & $\numnodes$ & -- & via root & build \\
values template $\cfv_{\mathrm{tmpl}}$ & $\numnodes$ & yes & yes & build, updates \\
regrets $\regret$ & $\numslots$ & yes & reset & every iteration \\
strategy sum $\bar{s}$ (float64) & $\numslots$ & yes & reset & every iteration \\
prediction buffer (PCFR$^{+}$) & $\numslots$ & yes & reset & every iteration \\
strategy $\strat_{\mathrm{ext}}$ (+ sentinel) & $\numslots + 1$ & yes & -- & every iteration \\
dual-lane reach & $2\numnodes$ & yes & -- & every iteration \\
values $\cfv$ & $\numnodes$ & yes & -- & every iteration \\
set totals & $\numinfosets$ & yes & -- & every iteration \\
iteration counter $t$ (device) & 1, float32 & yes & reset & every replay \\
CUDA-graph pool & driver-managed & -- & -- & capture \\
\bottomrule
\end{tabular}}
\end{table*}

\paragraph{Update rules share one schedule.}
\Cref{tab:rules} counts the Aten operations each shipped rule issues per
iteration. The rules differ only in the hooks of step 9 (the dot matrix of
\cref{fig:compiler}), which act on the slot vectors: the delta against CFR$^{+}$ is a constant per rule (from $-1$ for vanilla
CFR, which drops the clamp, to $+29$ for DCFR's discount schedule) and
does not change between Leduc and the 434k-state turn subgame, while the
tree-dependent part of the schedule, the $8D$ operations of the two passes, is
shared verbatim. Every rule's body captures into a CUDA graph, because the hooks
read the iteration index from the device-side counter. The convergence of the rules is reported in \cref{app:extra}.

\input{figures/TABLE_T16_rules.tex}

%% file: sections/fig_compiler_passes.tex
\begin{figure*}[!htb]
\centering
\resizebox{\textwidth}{!}{%
\begin{tikzpicture}[
    font=\sffamily\small,
    dec1/.style={circle, draw=teal!60!black, fill=teal!45, minimum size=2.6mm, inner sep=0pt},
    dec2/.style={circle, draw=violet!60!black, fill=violet!40, minimum size=2.6mm, inner sep=0pt},
    chn/.style={diamond, draw=black!65, fill=black!22, minimum size=3.2mm, inner sep=0pt},
    ter/.style={rectangle, draw=black!70, fill=black!55, minimum size=2.0mm, inner sep=0pt},
    edge/.style={black!55, line width=0.7pt},
    flow/.style={-{Stealth[length=2.2mm]}, line width=1.0pt, black!60},
    arr/.style={-{Stealth[length=1.6mm]}, line width=0.9pt, black!60},
    fold/.style={-{Stealth[length=1.8mm]}, densely dashed, line width=0.8pt, black!55},
    ret/.style={-{Stealth[length=1.8mm]}, densely dashed, line width=0.9pt, orange!75!black},
    hook/.style={draw=orange!75!black, densely dashed, rounded corners=1.5pt, inner sep=1.5pt, text=orange!60!black},
    tick/.style={orange!80!black, line width=1.0pt},
    gtick/.style={orange!95!black, line width=1.6pt},
    note/.style={font=\sffamily\small, text=black!55},
    lanelabel/.style={font=\sffamily\small\bfseries, text=black!60},
    hdr/.style={font=\sffamily\small, text=black!65, anchor=west},
    badge/.style={circle, fill=orange!85!black, text=white,
                  font=\sffamily\small\bfseries, inner sep=1.0pt}
]
\def\cw{0.34}
\newcommand{\cell}[3]{\draw[black!60, fill=#3] ({#1},{#2}) rectangle ++(\cw,\cw);}
\newcommand{\seps}[3]{\foreach \s in {#3}
  \draw[black!85, line width=1.0pt] ({#1+\s*\cw},#2-0.03) -- ++(0,\cw+0.06);}
\newcommand{\blockstrip}[3]{%
  \foreach [count=\i from 0] \c in #3{
    \pgfmathsetmacro{\g}{(\i>1)*0.12 + (\i>5)*0.12}
    \cell{#1+\i*\cw+\g}{#2}{\c}}}
\def\Pfills{black!22, black!22, teal!45, teal!45, teal!45, teal!45, violet!40, violet!40, violet!40, violet!40}
\def\Cfills{teal!45, teal!45, violet!40, black!55, violet!40, black!55, black!55, black!55, black!55, black!55}
\def\Sfills{black!45, black!45, teal!45, teal!45, teal!45, teal!45, violet!40, violet!40, violet!40, violet!40}
\def\Qone{black!45, black!45, teal!45, teal!45, teal!45, teal!45, black!45, black!45, black!45, black!45}
\def\Qtwo{black!45, black!45, black!45, black!45, black!45, black!45, violet!40, violet!40, violet!40, violet!40}
\def\Nfills{black!22, teal!45, teal!45, violet!40, black!55, violet!40, black!55, black!55, black!55, black!55, black!55}
\def\Tfills{white, white, white, white, black!55, white, black!55, black!55, black!55, black!55, black!55}

\node[chn]  (c0)  at (1.45,-0.70) {};
\node[dec1] (p1a) at (0.85,-1.40) {};
\node[dec1] (p1b) at (2.05,-1.40) {};
\node[dec2] (q2a) at (0.50,-2.10) {};
\node[ter]  (t2a) at (1.20,-2.10) {};
\node[dec2] (q2b) at (1.70,-2.10) {};
\node[ter]  (t2b) at (2.40,-2.10) {};
\node[ter]  (t3a) at (0.30,-2.80) {};
\node[ter]  (t3b) at (0.70,-2.80) {};
\node[ter]  (t3c) at (1.50,-2.80) {};
\node[ter]  (t3d) at (1.90,-2.80) {};
\draw[edge, densely dashed] (c0) -- (p1a); \draw[edge, densely dashed] (c0) -- (p1b);
\draw[edge] (p1a) -- (q2a); \draw[edge] (p1a) -- (t2a);
\draw[edge] (p1b) -- (q2b); \draw[edge] (p1b) -- (t2b);
\draw[edge] (q2a) -- (t3a); \draw[edge] (q2a) -- (t3b);
\draw[edge] (q2b) -- (t3c); \draw[edge] (q2b) -- (t3d);
\node[note] at (1.45,-3.35) {game specification $G$};
\node[badge] at (3.05,-1.05) {1};
\node[diamond, draw=black!65, fill=blue!6, minimum size=5.2mm, inner sep=0pt] (gate) at (3.05,-1.65) {};
\draw[green!45!black, line width=1.0pt] (2.93,-1.65) -- (3.02,-1.76) -- (3.18,-1.54);
\draw[arr] (2.55,-1.65) -- (gate.west);
\draw[arr] (gate.east) -- (3.60,-1.65);
\def\sxa{3.65}
\node[badge] at (3.78,-0.42) {2};
\node[hdr] at (3.98,-0.42) {slots: prefix sum over sets};
\foreach [count=\i from 0] \c in {teal!45, teal!45, teal!45, teal!45, violet!40, violet!40, violet!40, violet!40}
  \cell{\sxa+\i*\cw}{-0.95}{\c}
\seps{\sxa}{-0.95}{0, 2, 4, 6, 8}
\node[badge] at (3.78,-1.45) {3};
\node[hdr] at (3.98,-1.45) {nodes, parent before child};
\foreach [count=\i from 0] \c in \Nfills \cell{\sxa+\i*\cw}{-1.98}{\c}
\seps{\sxa}{-1.98}{0, 1, 3, 7, 11}
\node[badge] at (3.78,-2.48) {4};
\node[hdr] at (3.98,-2.48) {edges: parent, child, slot};
\foreach [count=\i from 0] \c in \Pfills \cell{\sxa+\i*\cw}{-3.00}{\c}
\foreach [count=\i from 0] \c in \Cfills \cell{\sxa+\i*\cw}{-3.34}{\c}
\foreach [count=\i from 0] \c in \Sfills \cell{\sxa+\i*\cw}{-3.68}{\c}
\node[note, anchor=east] at (3.57,-2.83) {$P$};
\node[note, anchor=east] at (3.57,-3.17) {$C$};
\node[note, anchor=east] at (3.57,-3.51) {$S$};
\def\sxb{7.95}
\node[badge] at (8.08,-0.42) {6};
\node[hdr] at (8.28,-0.42) {template $\cfv_{\mathrm{tmpl}}$};
\foreach [count=\i from 0] \c in \Tfills \cell{\sxb+\i*\cw}{-0.95}{\c}
\seps{\sxb}{-0.95}{0, 1, 3, 7, 11}
\draw[arr] (7.45,-1.81) -- (7.70,-1.81) |- (7.90,-0.78);
\draw[arr] (7.45,-2.83) -- (7.90,-2.83);
\blockstrip{\sxb}{-3.00}{\Pfills}
\foreach \a/\b/\l in {7.95/8.63/0, 8.75/10.11/1, 10.23/11.59/2}{
  \draw[decorate, decoration={brace, mirror, amplitude=2pt}, black!55] (\a,-3.07) -- (\b,-3.07)
    node[note, midway, below=1.5pt] {$B_\l$};}
\draw[fold] (8.29,-2.96) -- (8.29,-0.99) node[note, pos=0.45, right=0.6mm] {fold chance};
\node[badge] at (8.08,-3.72) {5};
\node[hdr] at (8.28,-3.72) {depth schedule: stable sort by parent depth, $D{=}3$ blocks};
\def\sxc{12.45}
\node[badge] at (12.58,-0.42) {7};
\node[hdr] at (12.78,-0.42) {index emission per block (dark = $\bot$)};
\blockstrip{\sxc}{-0.95}{\Qone}
\blockstrip{\sxc}{-1.29}{\Qtwo}
\blockstrip{\sxc}{-1.63}{\Sfills}
\node[note, anchor=west] at (16.15,-0.78) {$q_1$};
\node[note, anchor=west] at (16.15,-1.12) {$q_2$};
\node[note, anchor=west] at (16.15,-1.46) {$q_{\mathrm{val}}$};
\draw[arr] (11.66,-2.83) -- (12.05,-2.83) |- (12.40,-1.29);
\draw[arr] (11.66,-2.83) -- (12.40,-2.83);
\foreach [count=\i from 0] \c in {teal!45, teal!45, teal!45, teal!45, violet!40, violet!40, violet!40, violet!40}
  \cell{\sxc+\i*\cw}{-3.00}{\c}
\node[hdr] at (12.45,-2.48) {flat over decision edges};
\foreach \a/\b/\l in {12.45/13.81/1, 13.81/15.17/2}{
  \draw[decorate, decoration={brace, mirror, amplitude=2pt}, black!55] (\a+0.02,-3.07) -- (\b-0.02,-3.07)
    node[note, midway, below=1.5pt] {player \l};}
\draw[flow] (16.95,-4.18) -- (16.95,-5.40);
\node[note, anchor=east] at (16.80,-4.78) {copy to the device, once};

\begin{scope}[yshift=-0.45cm]
\node[badge] at (0.35,-5.30) {8};
\node[hdr] at (0.55,-5.30) {memory plan};
\def\bx{1.95}      
\def\bs{0.085}     
\def\px{4.85}      
\def\mx{6.15}      
\def\ypred{-8.75}  
\node[note, anchor=base] at (\px,-5.68) {pinned};
\node[note, anchor=base] at (\mx,-5.68) {mutable};
\foreach [count=\r from 0] \l/\n/\f/\p/\m in {
  {index arrays}/26/black!22/0/0,
  {$\cfv_{\mathrm{tmpl}}$}/11/white/1/1,
  {$\regret$}/8/orange!30/1/1,
  {$\bar{s}$ (f64)}/8/orange!15/1/1,
  {prediction}/8/orange!15/1/1,
  {$\strat_{\mathrm{ext}}$}/9/green!25/1/0,
  {$\reach_{1,2}$}/22/teal!28/1/0,
  {$\cfv$}/11/gray!15/1/0,
  {set totals}/4/gray!15/1/0,
  {counter $t$}/1/gray!15/1/1}{
  \pgfmathsetmacro{\y}{-5.95-\r*0.30}
  \node[note, anchor=east] at (\bx-0.10,\y) {\l};
  \draw[black!60, fill=\f] (\bx,\y-0.10) rectangle ++({\n*\bs},0.20);
  \ifnum\p=1 \fill[blue!60!black] (\px,\y) circle (1.4pt); \fi
  \ifnum\m=1 \draw[orange!75!black, line width=0.9pt] (\mx,\y) circle (1.6pt); \fi
}
\fill[white] (3.35,-6.13) -- (3.42,-5.77) -- (3.52,-5.77) -- (3.45,-6.13) -- cycle;
\draw[black!60] (3.35,-6.13) -- (3.42,-5.77); \draw[black!60] (3.45,-6.13) -- (3.52,-5.77);
\node[note, anchor=west] at (4.02,-5.95) {${\approx}9\numedges$};

\node[badge] at (6.95,-5.30) {9};
\node[hdr] at (7.15,-5.30) {rule lowering};
\node[draw=black!30, rounded corners=4pt, fill=black!2, minimum width=20mm,
      minimum height=33mm, anchor=north] (body) at (8.10,-5.95) {};
\node[hook] at (8.10,-6.35) {discount};
\node[note] at (8.10,-6.95) {forward};
\node[note] at (8.10,-7.55) {backward};
\node[hook] at (8.10,-8.15) {after};
\node[hook] at (8.10,\ypred) {prediction};
\node[note, anchor=north] at ($(body.south)+(0,-0.05)$) {iteration body};
\def\rx{9.75}
\foreach [count=\j from 0] \r in {CFR, CFR$^{+}$, PCFR$^{+}$, DCFR}
  \node[note, rotate=60, anchor=west, inner sep=0pt] at ({\rx+\j*0.55},-6.10) {\r};
\foreach \y in {-6.35,-8.15,\ypred} \draw[black!25] (9.15,\y) -- (11.65,\y);
\foreach \j in {0,...,3} \foreach \y in {-6.35,-8.15,\ypred}
  \draw[black!35, fill=white] ({\rx+\j*0.55},\y) circle (1.5pt);
\foreach \j/\y in {1/-8.15, 2/-6.35, 2/-8.15, 2/\ypred, 3/-6.35}
  \fill[orange!85!black] ({\rx+\j*0.55},\y) circle (1.5pt);

\node[badge] at (12.35,-5.30) {10};
\node[hdr] at (12.60,-5.30) {runtime};
\def\wz{13.35}
\node[note, anchor=east] at (\wz-0.08,-6.05) {warm-up};
\draw[black!35] (\wz,-6.05) -- (16.05,-6.05);
\foreach \i in {0,1,2}{
  \foreach \k in {0,...,5}{
    \draw[tick] ({\wz+0.10+\i*0.90+\k*0.13},-5.92) -- ++(0,-0.26);}}
\node[draw=orange!80!black, fill=orange!25, rounded corners=1.5pt,
      minimum height=3.4mm, inner sep=1.6pt, anchor=west] (rec) at (16.15,-6.05) {record};
\node[note, anchor=east] at (\wz-0.08,-7.05) {replay};
\draw[black!35] (\wz,-7.05) -- (16.35,-7.05);
\foreach \i in {0,...,3}
  \draw[gtick] ({\wz+0.20+\i*0.90},-6.86) -- ++(0,-0.38);
\draw[flow, line width=0.9pt] (16.35,-7.05) -- ++(0.30,0);
\node[note, anchor=west] at (16.70,-7.05) {$\times T$};
\draw[ret] (16.35,-7.65) -- (16.35,-9.95) -- (\mx,-9.95) -- (\mx,-8.90);
\node[note, anchor=south east, text=orange!60!black, align=right] at (16.25,-9.90)
  {between solves:\\ rewrite mutable buffers only;\\ the captured graph stays valid};
\end{scope}

\begin{scope}[on background layer]
\node[draw=blue!45!black!30, rounded corners=3pt, inner sep=1.2mm, fill=blue!3,
      fit={(0.10,-0.25) (17.40,-3.98)}] (hostlane) {};
\node[draw=orange!70!black!30, rounded corners=3pt, inner sep=1.2mm, fill=orange!4,
      fit={(0.10,-5.55) (17.40,-10.55)}] (devlane) {};
\end{scope}
\node[lanelabel, anchor=south west] at (hostlane.north west)
  {Build on the host, once per game (\cref{alg:compile}, lines 1--7)};
\node[lanelabel, anchor=south west] at (devlane.north west)
  {Device: allocate once, replay every iteration (lines 8--10)};
\end{tikzpicture}}
\caption{The compiler passes of \cref{alg:compile} on the example game of
\cref{fig:archpipeline}; badges are algorithm lines. The host lane flattens
the tree into slot, node, and edge strips, sorts the edges into $D$ depth
blocks, folds chance into the values template, and emits the per-lane index
arrays, where dark cells point at the sentinel $\bot$. The device lane sizes
every buffer of \cref{tab:memplan} once (bars scale with element count; dots
mark pinned addresses, rings the inputs rewritten between solves), binds the
hooks each update rule uses (dots: bound), and after three eager iterations
records one graph that is replayed thereafter; between solves only the ringed
buffers are rewritten.}
\Description{Two stacked lanes. The top lane shows a small game tree passing a
diamond-shaped validation gate into three flat strips of colored cells for
slots, nodes, and edges, then an arrow into the same edges regrouped into three
depth blocks with braces, a dashed arrow folding chance into a template strip
whose terminal cells are dark, and finally three index strips per block where
opponent and chance cells are dark. A thick arrow drops into the bottom lane
labeled copy once. The bottom lane has a bar chart of buffer sizes with a
pinned column of blue dots and a mutable column of orange rings, a rounded
capsule listing the iteration body with three dashed hook slots and a dot
matrix of four update rules showing which hooks each binds, and two timelines:
a warm-up row with dense ticks, a record box, and a replay row with one thick
tick per iteration. A dashed orange return path leads from the replay row back
to the mutable column.}
\label{fig:compiler}
\end{figure*}
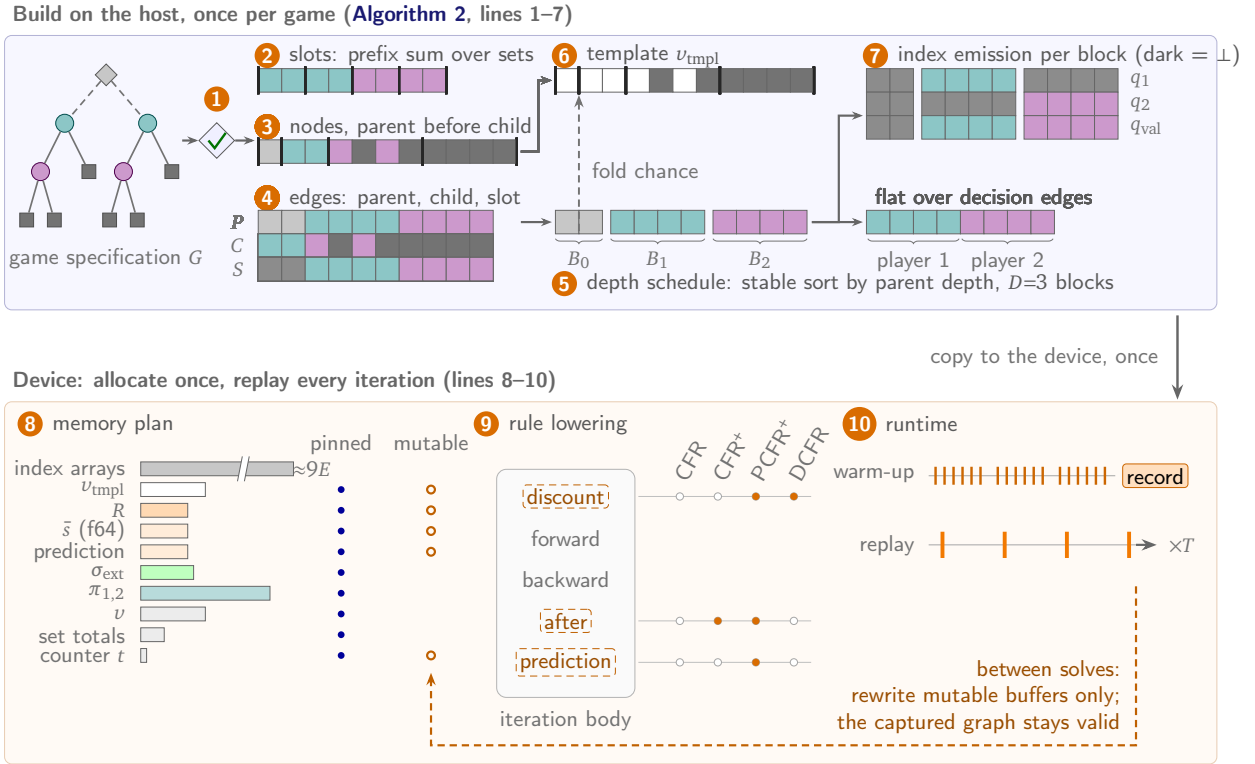

%% file: figures/TABLE_T16_rules.tex
\begin{table*}[!htb]
\centering
\caption{Generated operation schedule per update rule (Aten operations per iteration,
counted under \texttt{TorchDispatchMode}, load-immune). $\Delta$ is the rule's
difference from CFR$^{+}$; the last column lists the operations the rule adds to or
removes from the CFR$^{+}$ schedule; \emph{graph} records whether the rule's iteration
body captured into a CUDA graph on the A100. The delta is identical on both games: rule logic touches only the slot vectors, so its
cost is independent of tree size.}
\label{tab:rules}
{
\setlength{\tabcolsep}{4pt}
\begin{tabular}{lrrrcl}
\toprule
Rule & Leduc & HUNL turn & $\Delta$ & Graph & Schedule difference vs.\ CFR$^{+}$ \\
\midrule
CFR & 119 & 95 & -1 & yes & 0 kinds added, 1 removed \\
CFR$^{+}$ & 120 & 96 & 0 & yes & baseline \\
PCFR$^{+}$ & 133 & 109 & +13 & yes & 12 kinds added (2$\times$\texttt{where}, 1$\times$\texttt{add}, 1$\times$\texttt{div}, \ldots) \\
DCFR & 149 & 125 & +29 & yes & 10 kinds added, 1 removed (7$\times$\texttt{where}, 4$\times$\texttt{gt}, 3$\times$\texttt{div}, \ldots) \\
\bottomrule
\end{tabular}}
\end{table*}

%% file: sections/F_proofs.tex
\section{Proofs}
\label{app:proofs}

Throughout, the compiled game is a finite tree: every node other than the root is the
child of exactly one edge, and $d(c(e)) = d(p(e)) + 1$ for every edge $e$. We write
$z \succeq h$ for ``$z$ is a terminal in the subtree rooted at $h$'', $\pi_c(h)$ for the
product of chance probabilities on the root-to-$h$ path, $\pi_c(h \!\to\! z)$ for the
product over the $h$-to-$z$ path, and $\reach^{\strat}_{1,2}(h \!\to\! z)$ for the
product of both players' strategy probabilities over that same path. Notation and the
operators $\fwd, \bwd, \acc$ are those of \cref{sec:operators}.

\subsection{Proposition~\ref{prop:chance}: folded chance is exact}

\begin{proof}
Induct on the height $\eta(h)$ of $h$, the number of edges on the longest path from $h$
to a terminal below it. Because $\bwd_\ell$ is applied in decreasing $\ell$, when the
block containing $h$'s outgoing edges executes, every child of $h$ has already received
its final value.

\emph{Base case} $\eta(h) = 0$. Then $h \in \termnodes$, no block writes to $h$, and $h$
retains its template value $\cfv_{\mathrm{tmpl}}(h) = \pi_c(h)\, u(h)$ from
\cref{eq:template}. The right side of \cref{eq:chanceexact} has the single term
$z = h$, with $\pi_c(h \!\to\! h) = \reach^{\strat}_{1,2}(h \!\to\! h) = 1$, so it also
equals $\pi_c(h)\,u(h)$.

\emph{Inductive step.} Let $\eta(h) = k > 0$ and assume the claim for all nodes of
height below $k$; every child of $h$ has height at most $k-1$. Since $h \notin \termnodes$,
its template entry is $0$, so after $\bwd_{d(h)}$ the accumulated value is exactly
\begin{equation}
    \cfv(h) \;=\; \sum_{e \,:\, p(e) = h} \strat_{\mathrm{ext}}\big(q_{\mathrm{val}}(e)\big)\, \cfv\big(c(e)\big),
    \label{eq:proofrec}
\end{equation}
where $q_{\mathrm{val}}$ is the backward slot map of \cref{sec:operators}, which sends
chance edges to the sentinel. Two
cases.

If $h$ is a chance node, every outgoing edge is a chance edge, so
$\strat_{\mathrm{ext}}(q_{\mathrm{val}}(e)) = 1$ for all of them and
$\reach^{\strat}_{1,2}(h \!\to\! z) = \reach^{\strat}_{1,2}(c(e) \!\to\! z)$ for
$z \succeq c(e)$. Substituting the inductive hypothesis into \cref{eq:proofrec},
\begin{equation*}
    \cfv(h) = \sum_{e \,:\, p(e) = h} \sum_{z \succeq c(e)}
        \pi_c(c(e) \!\to\! z)\, \pi_c(c(e))
        \reach^{\strat}_{1,2}(c(e) \!\to\! z)\, u(z).
\end{equation*}
For a chance edge $e$ with probability $\gamma(e)$ we have
$\pi_c(c(e)) = \gamma(e)\, \pi_c(h)$ and
$\pi_c(h \!\to\! z) = \gamma(e)\, \pi_c(c(e) \!\to\! z)$, hence
$\pi_c(c(e) \!\to\! z)\, \pi_c(c(e)) = \pi_c(h \!\to\! z)\, \pi_c(h)$ --- the factor
$\gamma(e)$ moves between the two arguments and the product is invariant. The summand
therefore already has the form required by \cref{eq:chanceexact}, and it is the chance
edge's \emph{multiplier} being $1$ that prevents $\gamma(e)$ from being applied a second
time here.

If $h$ is a decision node of player $i$, each outgoing edge $e$ is a decision edge with
$\strat_{\mathrm{ext}}(q_{\mathrm{val}}(e)) = \strat(q(e))$, and $\pi_c(c(e)) = \pi_c(h)$ because no
chance probability sits on a decision edge. Then
$\reach^{\strat}_{1,2}(h \!\to\! z) = \strat(q(e))\, \reach^{\strat}_{1,2}(c(e) \!\to\! z)$
and $\pi_c(h \!\to\! z) = \pi_c(c(e) \!\to\! z)$ for $z \succeq c(e)$, so substituting
the inductive hypothesis again gives summands of the required form.

In both cases the outer sum ranges over the children of $h$, whose terminal sets
$\{z : z \succeq c(e)\}$ partition $\{z : z \succeq h\}$ because the game is a tree.
Every terminal below $h$ therefore contributes to $\cfv(h)$ exactly once, which is the
no-double-counting claim, and collecting the terms yields \cref{eq:chanceexact}. Taking
$h$ to be the root, where $\pi_c(h) = 1$ and $\pi_c(h \!\to\! z) = \pi_c(z)$, gives the
expected payoff under $\strat$.
\end{proof}

The proof also identifies what would break the invariant: any dynamic multiplier other
than $1$ on a chance edge would apply $\gamma(e)$ twice, once at the template and once
in the recurrence.

\subsection{Proposition~\ref{prop:lanes}: the dual-lane update is exact}

\begin{proof}
Fix a lane $i$, a block $B_\ell$, and an edge $e \in B_\ell$. By \cref{eq:opfwd},
$\fwd_{\ell,i}$ writes $\reach_i(c(e)) = \reach_i(p(e))\, \strat_{\mathrm{ext}}(q_i(e))$
and leaves every node that is not a child in $B_\ell$ unchanged. Compare with
\cref{eq:lanes} by cases on the owner of $e$.

If $e$ is a decision edge of player $i$, then $q_i(e) = q(e)$ by \cref{eq:slotmap} and
$q(e) \neq \bot$, so $\strat_{\mathrm{ext}}(q_i(e)) = \strat(\mathrm{slot}(p(e), a))$
for the action $a$ labelling $e$. The written value
$\reach_i(p(e))\, \strat(\mathrm{slot}(p(e), a))$ is the first branch of
\cref{eq:lanes}.

Otherwise $e$ is a decision edge of the opponent or a chance edge, so $q_i(e) = \bot$
and $\strat_{\mathrm{ext}}(\bot) = 1$ by construction. The written value is
$\reach_i(p(e)) \cdot 1 = \reach_i(p(e))$, the second branch.

The two cases are exhaustive and agree with \cref{eq:lanes} pointwise, so the operator
computes the case definition exactly. Since $q_i$ is materialized as an index array at
compile time, execution performs one gather at those indices and no predicate; the two
lanes occupy disjoint halves of a length-$2\numnodes$ buffer, so a single invocation over
the doubled index range advances both. Finally, assignment semantics are sound: within
$B_\ell$ all parents lie at depth $\ell$ and all children at depth $\ell+1$, and in a
tree each child has a unique incoming edge, so no destination is written twice and each
$\reach_i(p(e))$ read in block $\ell$ was written in block $\ell - 1$ (or initialized at
the root).
\end{proof}

\subsection{Proposition~\ref{prop:depth}: the depth schedule is shortest}

\begin{proof}
(i) Let $e, e'$ satisfy $c(e) = p(e')$. Then $d(p(e')) = d(c(e)) = d(p(e)) + 1$, so $e$
and $e'$ fall in blocks $B_\ell$ and $B_{\ell+1}$ with $\ell = d(p(e))$: they are never
in the same group, and executing blocks in increasing $\ell$ places $e$ before $e'$.
The blocks are nonempty by the definition of $\mathcal{D}$ in \cref{eq:blocks} and there
are $|\mathcal{D}| = D$ of them.

(ii) Let $\ell_{\min} < \dots < \ell_{\max}$ enumerate $\mathcal{D}$. Pick any edge
$e_{\max} \in B_{\ell_{\max}}$ and walk from $p(e_{\max})$ to the root, collecting the
incoming edge of each node. This produces edges $e_{\ell}$ with $d(p(e_\ell)) = \ell$
for every $\ell < \ell_{\max}$ with $B_\ell \neq \emptyset$, chained so that
$c(e_\ell) = p(e_{\ell+1})$ whenever both are on the walk. Every depth below
$\ell_{\max}$ that carries an edge is met, because the walk passes through one node per
depth and that node's incoming edge has parent depth one less. The walk therefore yields
a chain of $D$ edges, consecutive ones related by $c(e) = p(e')$. Any
dependency-respecting schedule must put each pair in different groups and in the chain's
order, so it uses at least $D$ groups.

(iii) By \cref{eq:opfwd,eq:opbwd,eq:opacc} each block contributes a fixed number of
operator invocations --- one $\fwd$ forward, one $\bwd$ and one $\acc$ backward, each a
constant number of gather/multiply/scatter calls --- and the remaining work (regret
matching, template initialization, the averaging and clamp steps of
\cref{alg:iteration}) is a fixed sequence independent of the block count. Writing $c_2$
for the per-block count and $c_1$ for the rest gives $c_1 + c_2 D$. Both constants are
independent of $\numnodes$ and $\numedges$: tree size enters only the \emph{length} of the
index arrays each invocation is given.
\end{proof}

On the measured suite the fitted constants are $c_1 \approx 31$ and
$c_2 \approx 8$ Aten operations, and the regression gate of \cref{sec:correctness}
asserts the count stays under $40 + 10D$.

\subsection{Proposition~\ref{prop:horizon}: the averaging horizon}

\begin{proof}
For an ideal real value $S$ in the normal range, let $2^k\le S<2^{k+1}$.
Its binade spacing is $\ulp(S)=2^{k-p+1}$, hence
$2^{-p}S<\ulp(S)\le 2^{1-p}S$.
The assumption $S_T=\Theta(T^2)$ therefore gives
$\ulp(S_T)=\Theta(T^2 2^{-p})$. Since $cT\le T x_T\le T$,
dividing proves \cref{eq:horizon}. Without the lower bound on $x_T$,
the numerator is only $O(T)$, giving the stated upper bound.
These statements concern the ideal sum. The recursively rounded sum can
depart from pointwise monotonicity and from any exact crossing constant.

For local rounding, let $\widehat{S}_{T-1}$ be the actual nonnegative,
finite floating-point accumulator and let its next larger representable
value also be finite. Define its upward spacing by
\[
    g_T=\operatorname{nextUp}(\widehat{S}_{T-1})-\widehat{S}_{T-1}.
\]
For the actual nonnegative increment $\widehat{\delta}_T$ supplied to the
addition, round-to-nearest yields
\[
    \operatorname{fl}(\widehat{S}_{T-1}+\widehat{\delta}_T)
    =\widehat{S}_{T-1}
    \quad\text{if}\quad 0\le\widehat{\delta}_T<g_T/2.
\]
At equality, ties are resolved by the format's tie-breaking rule.
This criterion uses the old rounded accumulator.
A later, larger increment can register, so absorption on one step does not
establish permanent stagnation. Subnormal spacings and overflow lie outside
the normal-range asymptotic argument above.
\end{proof}

Uniform averaging also has finite-resolution loss: with constant unit
increments, a float32 sum reaches $2^{24}$ and then stays there under
round-to-nearest, ties-to-even, because its upward spacing is $2$.
More generally, an ideal uniform sum of order $T$ with constant-order
increments also has increment-to-spacing ratio $\Theta(2^p/T)$.
GPU-CFR therefore uses float64 strategy accumulation for every update rule,
including uniform averaging. This widens accumulator resolution without
changing the precision of the increments or the device-side counter.

%% file: sections/I_future_work.tex
\section{Future Work}
\label{app:future}

GPU-CFR fixes the inner loop of tabular CFR and leaves the algorithm that calls it
untouched, so the natural next steps are the solvers that call that loop many times.
Real-time subgame re-solving \cite{burch2014solving,li2026real} runs a fresh solve
under a per-decision budget on every turn; a captured graph turns each of those
solves into one launch, and the budget then buys iterations instead of dispatch.
Abstraction pipelines \cite{li2026effective,li2026abstraction} score candidate
abstractions by solving each one, and computing tournament continuations in place of
the ICM heuristic \cite{li2026icm} solves one subtree per stack configuration; both
are loops of short, structurally identical solves that the compile-once design serves
directly. On the evaluation side, the exact best-response and value-decomposition
passes that variance-reduced estimators need \cite{li2026correlated,li2026av,li2026agents}
share the compiled arrays of \cref{sec:compiled}. Finally, solver-guided language
agents \cite{li2026pokerskill,wang2026solver} query equilibrium strategies during
reasoning, and a single-GPU solver that answers those queries in milliseconds removes
the latency that today limits such agents to precomputed charts. Extending the compiler
to more than two players and to sampled (Monte Carlo) CFR variants, whose traversal is
no longer static, is the main open design question.